\documentclass[a4paper]{article}
\usepackage[utf8]{inputenc}

\usepackage{amssymb,amsthm,amsmath,amsfonts}
\usepackage{graphicx,hyperref,color,xspace}
\usepackage[letterpaper, margin=1in]{geometry}
\usepackage{paralist}
\usepackage{authblk}

\usepackage{thmtools}

\usepackage{mystyle}

\usepackage{xifthen}

\newcommand{\dfree}[3][\delta]{%
    \ifthenelse{\equal{#2}{} }
    {\cD_{#1}}
    {\cD^{(#2,#3)}_{#1}}
}

\newcommand{\dreach}[4][\delta]{%
    \ifthenelse{\isempty{#2}}%
    {\ifthenelse{\isempty{#3}}%
        {\cR_{#1}}%
        {\cR_{#1}^{(#3,#4)}}%
    }%
    {\ifthenelse{\isempty{#3}}%
        {\cR_{#1,#2}}%
        {\cR_{#1,#2}^{(#3,#4)}}%
    }%
}
\renewcommand{\Comment}[1]{\State //\textit{#1}}

\newcommand{\disk}[2][\delta]{
    \ifthenelse{\isempty{#2}}
    {\mathrm{b}_{#1}}
    {\mathrm{b}_{#1}(#2)}
}

\newcommand{\nei}[3][s]{ N^{#1}_{#2,#3}}
\newcommand{\dia}[3][s]{ D^{#1}_{#2,#3}}
\newcommand{\ver}[3][s]{ V^{#1}_{#2,#3}}

\newcommand{\eps}{\varepsilon}

\newtheorem{theorem}{Theorem}
\newtheorem{lemma}[theorem]{Lemma}
\newtheorem{observation}[theorem]{Observation}
\newtheorem{definition}[theorem]{Definition}

\newtheorem{corollary}[theorem]{Corollary}

\let\originalleft\left
\let\originalright\right
\def\left#1{\mathopen{}\originalleft#1}
\def\right#1{\originalright#1\mathclose{}}

\title{On Computing the $k$-Shortcut Fr\'echet Distance}

\author{Jacobus Conradi\thanks{Jacobus Conradi is supported by the Deutsche Forschungsgemeinschaft (DFG, German Research Foundation) under grant number AA 1111/2-2 (FOR 2535 Anticipating Human Behavior).}~}

\author{Anne Driemel\thanks{Anne Driemel is supported by the Hausdorff Center for Mathematics (DFG grant number EXC 2047).}}
\affil{Institute of Computer Science, University of Bonn, Germany}

\date{\today}

\begin{document}

\maketitle
\pagestyle{plain}

\begin{abstract}
	The Fréchet distance is a popular measure of dissimilarity for polygonal curves. It is defined as a min-max formulation that considers all direction-preserving continuous bijections of the two curves. Because of its susceptibility to noise, Driemel and Har-Peled introduced the shortcut Fréchet distance in 2012, where one is allowed to take shortcuts along one of the curves, similar to the edit distance for sequences. 
	We analyse the parameterized version of this problem, where the number of shortcuts is bounded by a parameter $k$. The corresponding decision problem can be stated as follows: Given two polygonal curves $T$ and $B$ of at most $n$ vertices, a parameter $k$ and a distance threshold $\delta$, is it possible to introduce $k$ shortcuts along $B$ such that the Fr\'echet distance of the resulting curve and the curve $T$ is at most $\delta$?
	We study this problem for polygonal curves in the plane. We provide a complexity analysis for this problem with the following results: (i) assuming the exponential-time-hypothesis (ETH), there exists no algorithm with running time bounded by $n^{o(k)}$; (ii)  there exists a decision algorithm with running time in $\cO(kn^{2k+2}\log n)$.
	In contrast, we also show that efficient approximate decider algorithms are possible, even when $k$ is large. We present a \mbox{$(3+\eps)$-approximate} decider algorithm with running time in $\cO(k n^2 \log^2 n)$ for fixed $\eps$. In addition, we can show that, if $k$ is a constant and the two curves are $c$-packed for some constant $c$, then the approximate decider algorithm runs in near-linear time.
\end{abstract}

\thispagestyle{empty}
\setcounter{page}{1}

\section{Introduction}
\pagenumbering{arabic}

With the prevalence of geographical data collection and usage, the need to process and compare polygonal curves stemming from this data arises.
A popular versatile distance measure for polygonal curves is the Fr\'echet distance~\cite{su2020survey}. The distance measure is very similar to the well-known Hausdorff distance for geometric sets, except that it takes the ordering of points along the curves into account by minimizing over all possible direction-preserving continuous bijections between the two curves. 
Intuitively, the distance measure can be defined as follows. Imagine two agents independently traversing the two curves with varying speeds. Let $\delta$ be an upper bound on the (Euclidean) distance of the two agents that holds at any point in time during the traversal. The Fr\'echet distance corresponds to the minimum value of $\delta$ that can be attained over all possible traversals. 

In practice, the distance measure may be distorted by outliers and measurement errors. As a remedy, partial similarity and distance measures have been introduced which are thought to be more robust.
Buchin, Buchin and Wang define a \emph{partial Fr\'echet distance} \cite{Buchin2009ExactAfPCMvtFD} which maximizes the portions of the two curves matched to one-another within some given distance threshold.
Driemel and Har-Peled suggested the \emph{shortcut Fr\'echet distance}~\cite{Driemel2012JaywalkingYD} in the spirit of the well-known edit distance for strings: a set of non-overlapping subcurves can be replaced by straight edges connecting the endpoints (so-called shortcuts) to minimize the Fr\'echet distance of the resulting curves. 
Akitaya, Buchin, Ryvkin and Urhausen~\cite{Akitaya2019TheKD} introduced a variant of the Fr\'echet distance, where a certain number of  ``jumps'' (backwards and forwards) are allowed during the traversal of the two curves. We note that it has been acknowledged in the literature that partial dissimilarity measures generally do not satisfy metric properties~\cite{bronstein2009partial, veltkamp2001shape, jacobs2000class}. 

It is conceivable that computing a partial dissimilarity based on the Fr\'echet distance should be more difficult than the standard Fr\'echet distance because of the structure of the optimization problems involved. 
While the (discrete or continuous) Fr\'echet distance can be computed in roughly $\Theta(n^{2\pm \eps})$ time for two polygonal curves of $n$ vertices and any $\eps > 0$~\cite{Alt1995ComputingtFdbTPC, aronov2006frechet,  buchin2017four,  Bringmann14,BringmannM16, BuchinOS19}, the overall picture on the computational complexity of the partial variants is very heterogeneous. 

De Carufel, Gheibi, Maheshwari, and Sack~\cite{DECARUFEL2014625} showed that the problem of computing the partial Fr\'echet distance is not solvable by radicals over $\bQ$ and that the degree of the polynomial equations involved is unbounded in general. On the other hand, some variants of the partial Fr\'echet distance can be computed exactly in polynomial time~\cite{Buchin2009ExactAfPCMvtFD}.
Computing the shortcut Fréchet distance was shown to be NP-hard \cite{Driemel2012JaywalkingYD} when shortcuts are allowed anywhere along the curve. On the other hand, the \emph{discrete Fr\'echet distance with shortcuts} was shown to be computable in strictly subquadratic time by Avraham, Filtser, Kaplan, Katz, and Sharir~\cite{AvrahamFKKS15}, which is even faster than computing the discrete Fr\'echet distance without shortcuts. 
The variant defined by Akitaya, Buchin, Ryvkin, and Urhausen~\cite{Akitaya2019TheKD} turns out to be NP-hard, but allows for fixed-parameter tractable algorithms.

\subparagraph{Our contribution}
In this paper, we study the computational complexity of a parameterized version of the shortcut Fréchet distance, where the maximum number of shortcuts that may be introduced on the curve is restricted by a parameter $k$. We show that assuming the Exponential-Time-Hypothesis (ETH), no fixed-parameter tractable running time is possible with $k$ being the parameter. For polygonal curves in the plane, we present an exponential-time exact algorithm and we show that near-linear time approximation algorithms are possible using certain realistic input assumptions on the two curves.

\subparagraph{Previous work}
Driemel and Har-Peled~\cite{Driemel2012JaywalkingYD} introduced the shortcut Fr\'echet distance and described a near-linear time $(3+\eps)$-approximation algorithm for the class of $c$-packed curves. However, they only allowed shortcuts that start and end at \emph{vertices} of the base curve.
Buchin, Driemel and Speckmann~\cite{Buchin2013ComputingTF} showed that, if shortcuts are allowed \emph{anywhere along the curve}, then the problem of computing the shortcut Fr\'echet distance exactly is NP-hard via reduction from SUBSET-SUM. They also describe a 
$3$-approximation algorithm for the decision problem with running time in $\cO(n^3\log n)$ for the case that shortcuts may start and end in the middle of edges.
Prior to our work, there has been no study of exact algorithms for either variant of the shortcut Fr\'echet distance. Our analysis of the exact problem therefore closes an important gap in the literature. Obtaining the exact algorithm was surprisingly simple, once the relevant techniques were combined in the right way.

\subsection{Basic definitions}
\begin{definition}[curve]
	A curve $T$ is a continuous map from $[0,1]$ to $\bR^d$, where $T(t)$ denotes the point on the curve parameterized by $t\in[0,1]$.
	For $0\leq s < t\leq 1$ we denote the subcurve of $T$ from $T(s)$ to $T(t)$ by $T[s,t]$.
	A polygonal curve of complexity $n$ is given by a sequence of $n$ points in $\bR^d$.
	The curve is then defined as the piecewise linear interpolation between consecutive points.
\end{definition}

\begin{definition}[Fréchet distance]\label{def:frechet} 
	Given two curves $T$ and $B$ in $\bR^d$, their Fréchet distance  is defined as
	\[ d_\mathcal{F}(T,B) = \inf_{f,g:[0,1]\rightarrow[0,1]} ~ \max_{t\in[0,1]}\|T(f(t)) - B(g(t))\|,\]
	where $f$ and $g$ are monotone, continuous, increasing and surjective.
	We call a pair of such functions $(f,g)$ a traversal. Any such traversal has the cost $\max_{\alpha\in[0,1]}\|T(f(\alpha)) - B(g(\alpha))\|$ associated to it.
\end{definition}

In our definition of the Fr\'echet distance given above, we follow Alt and Godau~\cite{alt1995approximate}. Strictly speaking, this definition does not use bijections as for the sake of convenience the strict monotonicity of $f$ and $g$ is relaxed. 

\begin{definition}[$k$-shortcut curve]
	We call a line segment between two arbitrary points $B(s)$ and $B(t)$ of a curve $B$ a shortcut on $B$, where $s<t$ and denote it by $\overline{B}[s,t]$.
	A $k$-shortcut curve of $B$ is the result of replacing $k$ subcurves $B[s_i,t_i]$ of $B$  for $1\leq i\leq k$ by shortcuts $\overline{B}[s_i,t_i]$ connecting their start and endpoint, with $t_i\leq s_{i+1}$ for $1\leq i\leq k-1$ .
\end{definition}

\begin{definition}[$k$-shortcut Fréchet Distance]
	Given two polygonal curves $T$ and $B$, their $k$-shortcut Fréchet distance $d^k_\mathcal{S}(T,B)$ is defined as the minimum Fréchet distance between $T$ and any $k'$-shortcut curve of $B$ for some $0 \leq k' \leq k$. In this context, we call $B$ the base curve (where we take shortcuts) and $T$ the target curve (which we want to minimize the Fr\'echet distance to).
\end{definition}

\subsection{Overview of this paper}\label{our-results}

In Section \ref{exactchapter} we present an exact algorithm for deciding if the $k$-shortcut Fréchet distance is smaller than a given threshold $\delta$. The algorithm can also be used for the non-parameterized variant by setting $k=n$. Our first main result is the following theorem.

\begin{restatable}{theorem}{kexact}
\label{k-exact}
\label{thm:k-exact}
Let $T$ and $B$ be two polygonal curves in the plane with overall complexity $n$, together with a value $\delta>0$. There exists an algorithm with running time in $\cO\left(kn^{2k+2}\log n\right)$ and space in $\cO\left(kn^{2k+2}\right)$ that  decides whether $d^k_\mathcal{S}(T,B) \leq \delta$.
\end{restatable}

Our algorithm for Theorem~\ref{k-exact} iterates over the free-space diagram by Alt and Godau~\cite{Alt1995ComputingtFdbTPC} in $k$ rounds. Within the free-space diagram, a direction-preserving continuous bijection between two curves corresponds to a monotone path starting at $(0,0)$ and ending at $(1,1)$.
In each round, we compute the set of points in the parametric space of the two curves that are reachable by using one additional shortcut. For computing the set of eligible shortcuts spanning a fixed set of edges, we make use of the so-called line-stabbing wedge introduced by Guibas, Hershberger, Mitchell and Snoeyink~\cite{Guibas1994ApproximatingPSM}. Line-stabbing wedges were also used in the approximation algorithm by Buchin, Driemel, and Speckmann~\cite{Buchin2013ComputingTF}. In our case, since we perform exact computations, the  reachable space may be fragmented into a number of components, and this number may grow exponentially with the number of rounds.

In Section \ref{non-fpt} we give some evidence that this high complexity due to fragmentation is not an artifact of our algorithm, but may be inherent in the problem itself.
For this, we assume that the exponential time hypothesis (ETH) holds. The ETH states that $3$-SAT in $n$ variables cannot be solved in $2^{o(n)}$ time \cite{Impagliazzo1999}. Our second main result is the following conditional lower bound.

\begin{restatable}{theorem}{hardness}
\label{thm:ethhardness}
Unless ETH fails, there is no algorithm for the $k$-shortcut Fréchet distance decision problem in $\bR^d$ for $d\geq 2$, with running time $n^{o(k)}$.
\end{restatable}

Our conditional lower bound of Theorem~\ref{thm:ethhardness} is obtained via reduction from a variant of the $k$-SUM problem, which is called $k$-Table-SUM. In particular, we construct a $(4k+2)$-shortcut Fréchet distance decision instance for a given $k$-Table-SUM instance. Our construction is based on the NP-hardness reduction by Buchin, Driemel and Speckmann~\cite{Buchin2013ComputingTF}. Their reduction was from SUBSET-SUM and could not be directly applied to obtain our result. 
The construction implicitly encodes partial solutions for the SUBSET-SUM instance as reachable intervals on the edges of one of the curves. In this way, each shortcut taken by the optimal solution implements a choice for an element to be included in the sum. The reduction by Buchin et al.\ implemented this in the form of a binary choice, thereby leading to a number of shortcuts that is linear in $n$. In our case, the number of shortcuts taken should only depend on  $k$ and not $n$. Therefore, we give a new construction for a choice gadget, that allows to choose an element from a set to be included in a partial solution while using only a constant number of shortcuts for this choice.

In light of the above results, it is interesting to consider approximation algorithms and realistic input assumptions for this problem.  In Section \ref{c-packed} we show that there is an efficient approximation algorithm for this problem. If we can assume that the input curves are well-behaved, we even obtain a near-linear time algorithm for constant $k$. To formalize this, we consider the class of $c$-packed curves, see also~\cite{Driemel2010ApproximatingTF}.

\begin{restatable}[$c$-packed curves]{definition}{cpacked}
For $c>0$, a curve $X$ is called $c$-packed if the total length of $X$ inside any ball is bounded by $c$ times the radius of the ball.
\end{restatable}

The following is our third main result. Since any polygonal curve of complexity $n$ is $c$-packed for some $c \leq 2n$, the theorem also implies a running time of $\cO\left(kn^2\eps^{-5} (\log^2\left( n \eps^{-1} \right)\right)$ for polygonal curves in the plane---without any input assumptions.

\begin{restatable}{theorem}{ckapproximation}
\label{thm:ckapproximation}
Let $T$ and $B$ be two $c$-packed polygonal curves in the plane with overall complexity $n$, together with values $0<\eps\leq1$ and $\delta>0$. There exists an algorithm with running time in 
$ \cO\left(kcn\eps^{-5}\log^2\left( n\eps^{-1} \right)\right) $
and space in $\cO\left(kcn\eps^{-4}\log^2\left(\eps^{-1}\right)\right)$ which outputs one of the following: (i) $d^k_\mathcal{S}(T,B) \leq (3+\eps)\delta$ or (ii) $d^k_\mathcal{S}(T,B) > \delta$. In any case, the output is correct.
\end{restatable}

The main ideas that go into the proof of Theorem~\ref{thm:ckapproximation} can be sketched as follows. The first observation is that a highly fragmented reachable space that leads to the high running time of the exact algorithm of Theorem~\ref{thm:k-exact} can be approximated by limiting the number of shortcuts that the algorithm may take. To show that the algorithm still takes the right decisions (within the approximation bounds), we make use of a property of shortcut prices that was first observed by Driemel and Har-Peled~\cite{Driemel2012JaywalkingYD}. Namely,  the price of a shortcut is approximately monotone and it suffices in each round to take the `shortest` feasible shortcut among all shortcuts that are available.
Now, the main challenge as compared to the algorithm in~\cite{Driemel2012JaywalkingYD} is that this shortcut may still start in the middle of an edge. Evaluating the cost of this shortcut using line-stabbing wedges would be too expensive.
Instead, we use a data structure by Driemel and Har-Peled~\cite{Driemel2012JaywalkingYD} that allows to query the Fréchet distance of a line segment to a subcurve. We use this to implicitly approximate the line-stabbing wedge using a convex hull of a set of grid points.
However, this is still not enough, as the free-space may have quadratic complexity. To obtain a near-linear running time for small $c$, we make use of the property of $c$-packed curves as observed by Driemel, Har-Peled and Wenk~\cite{Driemel2010ApproximatingTF}, that the complexity of the free-space diagram of two $c$-packed curves is only linear in $c\cdot n$ when the curves are appropriately simplified.

\section{Preliminaries}

\begin{definition}[Free-space diagram]
	Let $T$ and $B:[0,1]$ be two polygonal curves in $\bR^2$.
	The free-space diagram of $T$ and $B$ is the joint parametric space $[0,1]^2$ together with a not necessarily uniform grid, where each vertical line corresponds to a vertex of $T$ and each horizontal line to a vertex of $B$ (refer to Figure~\ref{fig-fsd}).
	We call the cell of the parametric space corresponding to the $i$th edge of the target curve and the $j$th edge of the base curve $C_{i,j}$.
	The $\delta$-free-space of $T$ and $B$ is defined as \[ \dfree{}{}(T,B) = \left\{(x,y)\in[0,1]^2 \mid \|T(x) -B(y)\|\leq\delta \right\} \] 
	This is the set of points in the parametric space whose corresponding points on $B$ and $T$ are at a distance at most $\delta$.
	Denote by $\dfree{i}{j}(T,B) = \dfree{}{}(T,B) \cap C_{i,j}$ the $\delta$-free-space inside the cell $C_{i,j}$. 
\end{definition}%

In the following $T$ and $B$ will often be fixed, thus we will simply write $\dfree{}{}$.
It is known that $\dfree{i}{j}$ is convex and has constant complexity.
More precisely, it is an ellipse intersected with the cell $C_{i,j}$.
Furthermore the Fréchet distance between two curves is less than or equal to $\delta$ if and only if there exists a monotone path (in $x$ and $y$) in the free-space that starts in the lower left corner $(0,0)$ and ends in the upper right corner $(1,1)$ cf. \cite{Alt1995ComputingtFdbTPC}. 
In the case of the $k$-shortcut Fréchet distance we need to also consider shortcuts when traversing the parametric space.
When considering any $k$-shortcut curve $B'$ of $B$ and any traversal $(f,g)$ of $B'$ and $T$ with associated cost $\delta$, then $(f,g)$ induces traversals $(f',g')$ with associated cost at most $\delta$ on every shortcut $\overline{B}[s,t]$ and some corresponding subcurve $T[u,v]$ of $T$.
To capture this, we use the notion of \textit{tunnels} which was introduced in~\cite{Driemel2012JaywalkingYD} and is defined as follows.

 \begin{figure}
 \begin{center}
 	\includegraphics{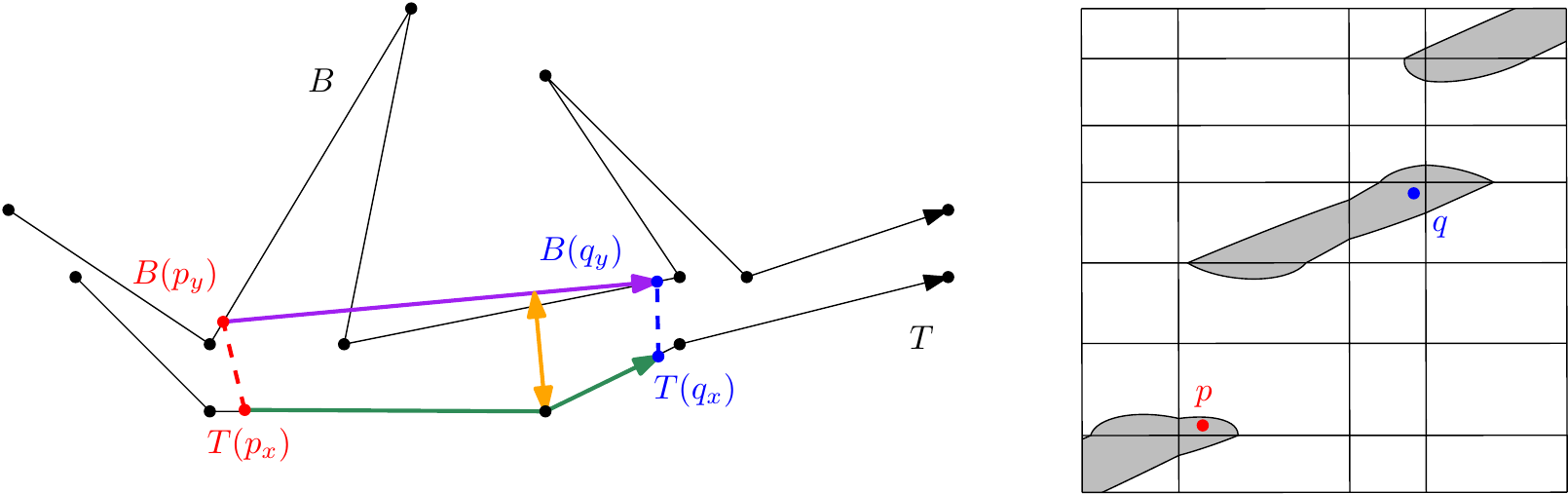}
 \end{center}
 	\caption{Free-space diagram for a base curve $B$ and a target curve $T$, with the $\delta$-free-space in gray. The figure shows a feasible proper tunnel $\tau(p,q)$. The shortcut $\overline{B}[p_y,q_y]$ is shown in purple, and the subcurve $T[p_x,p_y]$ in green. The price of $\tau(p,q)$ is the Fréchet distance of the shortcut and the subcurve. }
 	\label{fig-fsd}
 \end{figure}

\begin{definition}[Tunnel]
	A tunnel $\tau(p,q)$ is a pair of points $p=(x_p,y_p)$ and $q=(x_q,y_q)$ in the parametric space of $B$ and $T$, with $x_p\leq x_q$ and $y_p\leq y_q$. 
    $\tau(p,q)$ is called feasible if $p$ and $q$ are in $\dfree{}{}$. We say that a tunnel is proper, if the endpoints of the shortcut do not lie on the same edge of $B$.
	 We say a tunnel has a price $\prc(\tau(p,q)) = d_\cF(T[x_p,x_q],\overline{B}[y_p,y_q])$, refer to Figure~\ref{fig-fsd}.
\end{definition}

\begin{definition}[Reachable space]
	We define the $(\delta,s)$-reachable free-space of $T$ and $B$
	\[\dreach{s}{}{}(T,B) = \{(x_p,y_p)\in[0,1]^2\;|\;d^s_\cS(T[0,x_p],B[0,y_p])\leq\delta\}\]
	and again $\dreach{s}{i}{j}(T,B) = \dreach{s}{}{}(T,B)\cap C_{i,j}$.
	We call the intersection
	$\dreach{s}{i}{j}(T,B) \cap C_{a,b}$ for any $ (a,b) \in \{ (i-1,j), (i,j-1), (i+1,j), (i,j+1) \} $  a reachability interval of the cell $C_{i,j}$. In particular for $(a,b) \in \{ (i-1,j), (i,j-1) \} $ we call them incoming reachability intervals and  for $(a,b) \in \{ (i+1,j), (i,j+1) \} $ we call them outgoing reachability intervals.
\end{definition} 

We will simply write $\dreach{s}{}{}$ and $\dreach{s}{i}{j}$ whenever $T$ and $B$ are fixed.
Observe that the reachability intervals for every cell  $C_{i,j}$ and $s$ are contained in the boundary set  $\partial C_{i,j}$, and each reachability interval is described by a (possibly empty) single interval, since any two points in the reachability interval can be connected via a monotone path that stays inside the $\delta$-free-space. Furthermore, any tunnel $\tau(p,q)$ with $p=(x_p,y_p)$ and $q=(x_q,y_q)$, that is not proper, induces a traversal of $B[y_p,y_q]=\overline{B}[y_p,y_q]$ and $T[x_p,x_q]$. Thus we can omit the tunnel and replace it with a monotone path from $p$ to $q$ in $\dfree{}{}$. Therefore, in the following, we only consider monotone paths with proper tunnels. 

The $k$-shortcut Fréchet distance of $T$ and $B$ is at most $\delta$ if and only if $(1,1)\in\dreach{k}{}{}(T,B)$. We want to reduce the problem of deciding the $k$-shortcut Fr\'echet distance to the problem of finding a monotone path in the free-space diagram.

\begin{definition}[Monotone path with tunnels]
    A monotone path with $k$ proper tunnels in the $\delta$-free-space of two curves consists of $k+1$ monotone (in $x$ and $y$) paths in the $\delta$-free-space from $s_i$ to $t_i$ for $1\leq i \leq k+1$, with $s_1 = (0,0)$, such that $t_i$ lies to the left and below $s_{i+1}$, for $1\leq i \leq k$. The $k$ proper tunnels have the form $\tau(t_i,s_{i+1})$ for $1\leq i \leq k$. 
\end{definition}

\begin{observation}\label{obs-proper}
    Let $T$ and $B$ be two polygonal curves. The set $\dreach{s}{}{}(T,B)$ is exactly the set of points $p \in \dfree{}{}(T,B)$ such that there exists a monotone path  ending in  $p$ with at most $s$ proper tunnels, each of price at most  $\delta$. (By definition, these paths have to start at $(0,0)$).
\end{observation}

Note the following observation, used throughout the paper.

\begin{observation}\label{obs-trivial}
	Given line segments $\overline{ab}$ and $\overline{cd}$ in $\bR^d$, then for the Fréchet distance we have $d_\cF(\overline{ab},\overline{cd}) = \max(||a-c||,||b-d||)$.
\end{observation}
To decide whether a cell is reachable by a tunnel, we need to check if there exists a shortcut edge that stabs through an ordered set of disks centered at a subset of the vertices of the other curve. To formalize this, we use the notion of ordered stabbers and line-stabbing wedges as defined by Guibas, Hershberger, Mitchell and Snoeyink~\cite{Guibas1994ApproximatingPSM}.
 
 \begin{definition}[Line-stabbing wedge]\label{def:stabber}
    Given a sequence of $n$ convex objects $  O_1,\ldots,O_n$, an ordered stabber of this sequence is a line segment $l(x) = (1-x)s + xt$ from $s$ to $t$, such that points $0\leq x_1\leq x_2 \leq\ldots\leq x_n\leq 1$ exist with $p_i = l(x_i)\in O_i$. We call $p_i$ the realising points of $l$. We say that $l$ {stabs through} $O_1, \ldots, O_n$.
    We call the set of points $t$ that are endpoints of ordered stabbers of $O_1, \ldots, O_n$ the {line-stabbing wedge} of this sequence.   
 \end{definition}
 
In their paper, Guibas et al.\ give an algorithm to compute the line-stabbing wedge for a sequence of $n$ objects, with running time $\cO(n\log n)$. This line-stabbing wedge is described by $\cO(n)$ circular arcs, and two tangents that go to infinity (see Figure \ref{fig:stabber_example}).

\begin{observation}\label{obs:stabber}
    Let $T,B$ and $\delta$ be given. Denote by $v_k$ the vertices of $T$. For any feasible tunnel $\tau(p,q)$ with $p=(x_p,y_p)\in C_{a,b}$ and $q=(x_q,y_q)\in C_{i,j}$, it holds that $\overline{B}[y_p,y_q]$ stabs through the ordered set $\{\disk[\delta]{v_{a+1}},\ldots,\disk[\delta]{v_{i}}\}$, if and only if the price of $\tau(p,q)$ is at most $\delta$.
\end{observation}

\begin{figure}
    \centering
    \includegraphics[width=0.7\textwidth]{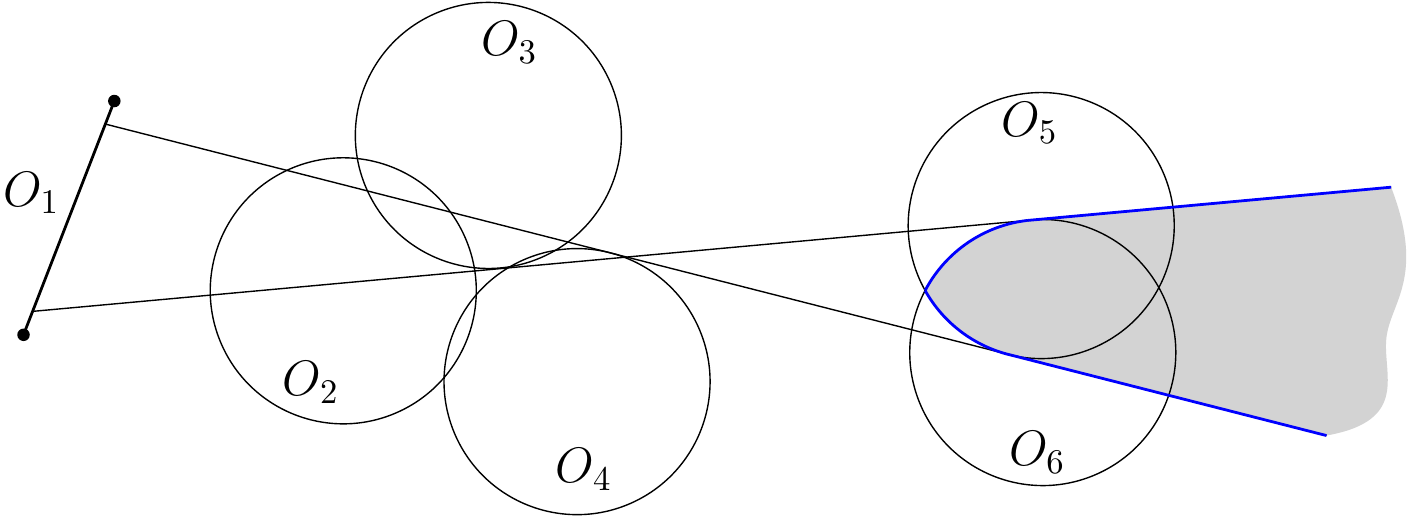}
    \caption{Example for the line-stabbing wedge for a line segment $O_1$ and disks $O_2,\ldots,O_6$. The line-stabbing wedge is shown in gray, with its boundary in blue.}
    \label{fig:stabber_example}
\end{figure}

\section{Exact decider algorithm}\label{exactchapter}

In this section we describe an exact decider algorithm for the $k$-shortcut Fréchet distance for two polygonal curves. The algorithm can also be used to solve the decision problem of the (unparameterized) shortcut Fr\'echet distance by setting $k=n$. We first describe the algorithm in Section~\ref{exact-alg-subsection} and then analyse its correctness and running time in Section~\ref{sec:analysis}.

\subsection{Description of the algorithm}\label{exact-alg-subsection}

We are given a parameter $k$, a value $\delta$ and the two polygonal curves $T$ and $B$ in the plane. 
Our algorithm iterates over the $\delta$-free-space diagram of $T$ and $B$ in $k$ rounds. In each round, based on the computation of the previous round, we compute the set of points that are reachable by using one additional shortcut. The goal is to compute the $(\delta,s)$-reachable space $\dreach{s}{}{}(T,B)$ in round $s$.
In each round, we handle the cells of the free-space diagram in a row-by-row order, and within each row from left to right. For every cell $C_{i,j}$ we consider three possible ways that a monotone path with proper tunnels can enter. 
\begin{enumerate}
\item The monotone path could enter the cell $C_{i,j}$ from the neighboring cell $C_{i-1,j}$ to the left or from the neighboring cell $C_{i,j-1}$ below. This does not directly involve a tunnel.
\item The monotone path could reach $C_{i,j}$ with a proper tunnel (only for $s \geq 1$). We distinguish between vertical and diagonal tunnels (compare~\cite{Buchin2013ComputingTF,Driemel2012JaywalkingYD} for a similar distinction).
\begin{enumerate}
    \item[(i)] The tunnel may start in any cell $C_{a,b}$ with $a<i$ and $b<j$.  We call this a diagonal tunnel.
    \item[(ii)] The tunnel may start in any cell $C_{i,b}$ for $b<j$. We call this a vertical tunnel.
\end{enumerate}
\end{enumerate}

Note that we do not consider (horizontal) tunnels starting in a cell $C_{a,j}$ with $a<i$, since we only consider proper tunnels.
Using this distinction, we will describe how to compute the set of points reachable by a monotone path with $s$ proper tunnels, for each cell of the diagram.  We denote the set computed by the algorithm for cell $C_{i,j}$ in round $s$ with $P^s_{i,j}$. The $(\delta,s)$-reachable space is then obtained by taking the union of these sets over all rounds
$\dreach{s}{i}{j}=\bigcup_{0\leq s'\leq s}P^{s'}_{i,j}$.

After $k$ rounds, the algorithm tests whether the point $(1,1)$ is contained in our computed set of reachable points. If this is the case, then the algorithm returns "$d^k_\cS(T,B)\leq\delta$", otherwise the algorithm returns "$d^k_\cS(T,B) > \delta$".

\begin{figure}
    \centering
    \includegraphics{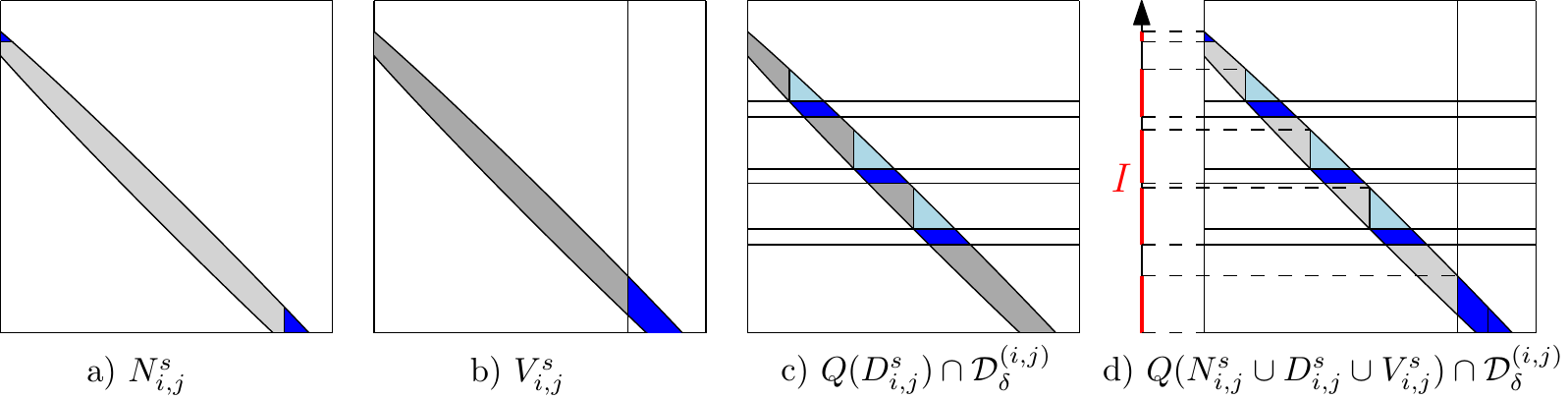}
    \caption{Example of the composition of the reachable space within a free-space cell. The fragmentation of the reachable space in this cell $P^s_{i,j}=Q(\nei[s]{i}{j}\cup\dia[s]{i}{j}\cup\ver[s]{i}{j})\cap\dfree{i}{j}$ results in a large family of intervals $I$ on the edge of $B$.}
    \label{fig:reachable_examples}
\end{figure}

\subparagraph{Propagating reachability within a cell}
To simplify the description of the algorithm, we use the following set function which receives a set $P \subseteq C_{i,j}$ for some cell $C_{i,j}$ and which 
extends $P$ to all points above and to the right of it.
\[ Q(P)=\{(x,y)\in [0,1]^2 \mid \exists (a,b)\in P \text{ such that }a\leq x\text{ and } b\leq y\}\]
We will usually intersect this set with $\dfree{i}{j}$ to obtain all points that are reachable from a point of $P$ by a monotone path that stays inside the $\delta$-free-space of this cell.  Figure~\ref{fig:reachable_examples} c) shows an example of the resulting set. Note that the boundary of the resulting set can be described by pieces of the boundary of $\dfree{i}{j}$, pieces of the boundary of $P$, and horizontal and vertical line segments. 

\subparagraph{Step 1: Neighbouring cells}\label{sec:neighbouring-cell}
Since we traverse the cells of the diagram in a lexicographical order, we have already computed the (possibly empty) sets $P_{i-1,j}^s$ and $P_{i,j-1}^s$, by the time we handle cell $C_{i,j}$ in round $s$. Therefore, we can compute the incoming reachability intervals  by intersecting $P_{i-1,j}^s$ and $P_{i,j-1}^s$ with $C_{i,j}$. 
Now we apply the function $Q$ to these sets and denote the result with $\nei[s]{i}{j}$:
\[ \nei[s]{i}{j} = \left( Q(P_{i-1,j}^s \cap C_{i,j}) \cup Q(P_{i,j-1}^s \cap C_{i,j})\right) \cap \dfree{i}{j} \]
Refer to Figure~\ref{fig:reachable_examples} a).

\subparagraph{Step 2 (i): Diagonal tunnels}
(only for $s \geq 1$) 
We invoke the following procedure for every $a<i$ and $b<j$ with $P^{s-1}_{a,b}$. We denote the union of resulting sets of points in $\dfree{i}{j}$ computed in this step with $\dia[s]{i}{j}$.

The procedure is given a set of points $P^{s-1}_{a,b}$ in the $\delta$-free-space $\dfree{a}{b}$ and computes all points in $\dfree{i}{j}$ that are endpoints of tunnels starting in $P^{s-1}_{a,b}$ with price at most $\delta$.
The procedure first projects $P^{s-1}_{a,b}$ onto the edge $e_b$ of the base curve. The resulting set consists of disjoint line segments $I=\{\overline{s_1\,t_1},\ldots\}$ along $e_b$ (refer to Figure \ref{fig:reachable_examples} d)\,).
The procedure then computes the line-stabbing wedge $W$ through $\overline{s_1\,t_1}$ and disks $\disk{v_{a+1}},\ldots,\disk{v_i}$ centered at vertices of $T$.
$W$ is then intersected with the edge $e_j$, resulting in a set $J$ on $e_j$ corresponding to a horizontal slab in $C_{i,j}$ (compare Figure \ref{fig:reachable_examples} c) and Figure \ref{fig:diagonalTunnel_example}\,). This resulting set is then intersected with $\dfree{i}{j}$ to obtain all endpoints of feasible shortcuts with price at most $\delta$ starting in $\overline{s_1\,t_1}$. 
The procedure performs the above steps for every line segment $\overline{s\,t}\in I$ and returns the union of these sets. The resulting set may look as illustrated in Figure \ref{fig:reachable_examples} c).

\subparagraph{Step 2 (ii): Vertical tunnels}\label{sec:vert-tunnel}
(only for $s \geq 1$) 
Let $p$ denote a point in $\bigcup_{l\leq j-1}P^{s-1}_{i,l}$ with minimal $x$-coordinate, i.e., a leftmost point in this set. A feasible vertical tunnel always has price at most $\delta$. Therefore, we simply take all points in the $\delta$-free-space to the right of $p$ in the cell $C_{i,j}$.
To do this, we compute the intersection of a halfplane that lies to the right of the vertical line at $p$ with the $\delta$-free-space in $C_{i,j}$. We denote this set with  $\ver[s]{i}{j}$. Refer to Figure~\ref{fig:reachable_examples} b) for an example.

\begin{figure}
    \centering
    \includegraphics[width=0.8\textwidth]{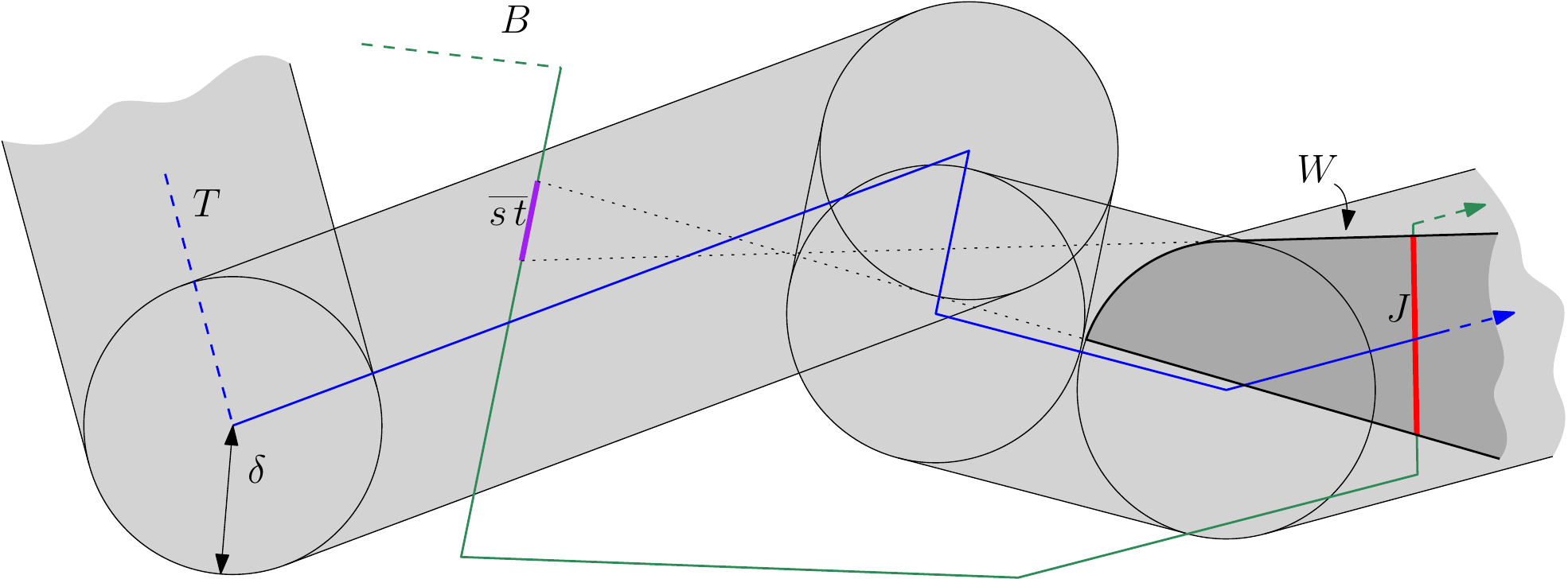}
    \caption{Example of the set $J$ (in red) computed by the \textsc{diagonalTunnel} procedure.}
    \label{fig:diagonalTunnel_example}
\end{figure}
\subparagraph{Putting things together}
Finally, we compute the set $P^s_{i,j}$ by taking the union of the computed sets and extending this set by using the function $Q$ defined above: 
\[ P^s_{i,j}=Q(\nei[s]{i}{j}\cup\dia[s]{i}{j}\cup\ver[s]{i}{j})\cap\dfree{i}{j}\]
It remains to specify the initialization: 
We define $P^{0}_{1,1}=\dfree{1}{1}$, if $(0,0) \in \dfree{}{}$, and otherwise $P^0_{1,1}=\emptyset$. Starting from this, we can compute the sets $P^{0}_{i,j}$ for $i,j \neq 1$ in a row-by-row fashion. For $s > 0 $ we continue in rounds, as described above.

\subsection{Analysis}\label{sec:analysis}
We now analyse the described algorithm.

\subsubsection{Correctness}

We argue that the structure of $P^s_{i,j}$ as computed by the algorithm is indeed as claimed. Namely for all $i,j$ and $s$ it holds that $\dreach{s}{i}{j} = \bigcup_{0\leq s'\leq s}P^{s'}_{i,j}$. We begin with two lemmas, before we prove this statement.

\begin{lemma}\label{exact:exact-diag}
    Let $T$ and $B$ be two polygonal curves with $n_1$ and $n_2$ edges respectively. For any $1\leq i\leq n_1$, $1\leq j\leq n_2$ and $1\leq s\leq k$ let $\dia{i}{j}$ be the set of endpoints of diagonal tunnels, as computed in the algorithm described in Section \ref{exact-alg-subsection}, and let $R=\bigcup_{a=1}^{i-1}\bigcup_{b=1}^{j-1}P_{a,b}^{s-1}$ be the set of reachable points by exactly $s-1$ proper tunnels in the lower-left quadrant of $C_{i,j}$. For any $q\in C_{i,j}$ the tunnel $\tau(p,q)$ has price $\prc(\tau(p,q))\leq\delta$ for some $p\in R$ if and only if $q\in \dia{i}{j}$.
\end{lemma}
\begin{proof}
    First let $a$ and $b$ be fixed and look at $P=P_{a,b}^{s-1}$.
    The diagonal tunnel procedure begins by projecting $P$ onto the edge $e_{b}$ of $B$, resulting in $I$.
    By the correctness of the procedure presented by Guibas et al.\, the diagonal tunnel procedure computes among other things the set of points in $\bR^2$ that are \emph{all} endpoints of stabbers through $I$ and $\disk{v_{a+1}},\ldots,\disk{v_{i}}$ centered at vertices of $T$.
    Intersecting this set with $e_{j}$ results in \emph{all} endpoints of stabbers through the ordered set ending on $e_{j}$, call this set $J$. 
    For every point $B(y_q)$ in $J$ there is at least one point $B(y_p)$ in $I$, such that $\overline{B}[y_p,y_q]$ stabs through $\{\disk{v_{a+1}},\ldots,\disk{v_{i}}\}$. 
    Hence, by Observation \ref{obs:stabber}, every point $p\in \dfree{a}{b}$, with $y$-coordinate $y_p$ and every point $q\in\dfree{i}{j}$ with $y$-coordinate $y_q$ form a feasible tunnel $\tau(p,q)$ with price at most $\delta$.
    Since the line-stabbing algorithm correctly computes \emph{all} possible endpoints of stabbers starting in $I$ and ending on $e_j$, the claim follows for $P_{a,b}^{s-1}$ by Observation \ref{obs:stabber}. That is, any $q$ such that there exists $p\in R$ with $\prc(\tau(p,q))\leq\delta$ also must be in $\dia{i}{j}$.
    As the algorithm iterates over all cells in the lower-left quadrant of $C_{i,j}$ and in the end defines $\dia{i}{j}$ as the union of above computed sets, the claim follows.
\end{proof}

\begin{lemma}\label{exact:exact-vertical}
		Let $T$ and $B$ be two polygonal curves with $n_1$ and $n_2$ edges, respectively. For any $1\leq i\leq n_1$, $1\leq j\leq n_2$ and $1\leq s\leq k$ let $\ver{i}{j}$ be the points reachable by a vertical tunnel as computed in the algorithm and let $R=\bigcup_{b=1}^{j-1}P_{i,b}^{s-1}$ be the set of reachable points by exactly $s-1$ proper tunnels in the column below $C_{i,j}$. For any $q\in C_{i,j}$ the tunnel $\tau(p,q)$ has price $\prc(\tau(p,q))\leq\delta$ for some $p\in R$ if and only if $q\in \ver{i}{j}$.
\end{lemma}
\begin{proof}
	Note that any vertical tunnel costs at most $\delta$ if it is feasible by Observation \ref{obs-trivial}.
	Furthermore note that the leftmost point $p$ in $R$ is stored in $l^{s-1}_{i,j-1}$ hence, we can retrieve $p$.
	Now assume $\tau(r,q)$ is an arbitrary vertical tunnel with $r\in R$ and $q\in C_{i,j}$.
	Since a tunnel must be monotone $x_r\leq x_q$.
	Because $p$ is the leftmost point in $R$ we have $x_p\leq x_r\leq x_q$.
	From the way we constructed $\ver{i}{j}$ (intersecting a vertical closed halfplane to the right of $p$ with $\dfree{i}{j}$) it follows that $q\in \ver{i}{j}$.
\end{proof}

\begin{theorem}\label{exact:correctness}
Let $T$ and $B$ be two polygonal curves in the plane with overall complexity $n$ together with a value $\delta>0$. Let $P_{i,j}^s$ be the set of reachable points with exactly $s$ proper tunnels as computed in the algorithm for all $i,j$ and $s$. It holds that 
\[ \bigcup_{s'\leq s}P_{i,j}^{s'} = \dreach{s}{i}{j}(T,B).\]
Thus the algorithm correctly decides, whether the $k$-shortcut Fréchet distance of $T$ and $B$ is at most $\delta$.
\end{theorem}
\begin{proof}
We show that the reachable space $\dreach{s}{i}{j}$ is correctly computed via induction in $i,j$ and $s$.
Note that $\dreach{s'}{1}{1}$ is computed correctly for all $s'\leq k$ since $\dfree{1}{1}$ is convex and the algorithm checks whether $(0,0)\in\dfree{1}{1}$.  Thus, if $(0,0)\in\dfree{i}{j}$, $\dreach{0}{1}{1}=\dfree{1}{1}=P^0_{1,1}$ is computed in the first step, by convexity of $\dfree{1}{1}$, otherwise it is empty. For $s'>0$ the set $P^{s'}_{1,1}$ is empty since no cell is below or to the left of it. Hence, $\dreach{s'}{1}{1}=\dreach{0}{1}{1}$ is also computed correctly.

By induction all cells $C_{\leq n_1,<j}$ and $C_{<i,j}$ and in particular $C_{i-1,j}$ and $C_{i,j-1}$ have been handled correctly up to round $s$ and $P_{i,j}^s$ and is stored for every correctly handled cell. Assume now that some point $q\in \dreach{s}{i}{j}$ is given. By Observation \ref{obs-proper}, the point $q$ corresponds to a monotone path with $s'\leq s$ proper tunnels. There are three possible ways via which this point in the parametric space is reachable.
The path reaching $q$ could take $s'$ shortcuts to reach $C_{i-1,j}$ or $C_{i,j-1}$, and enter via a monotone path through the boundary into $C_{i,j}$ at some point $a\in \partial C_{i,j}$. As $C_{i-1,j}$ and $C_{i,j-1}$ have been handled correctly for $s'$, the incoming reachability intervals on the boundary have been computed correctly containing $a$, thus $q$ is also in $P_{i,j}^{s'}$.

Alternatively the path could enter some cell $C_{i,l}$ with $s'-1$ shortcuts and then take a vertical shortcut into $C_{i,j}$ for some $l<j$ and then possibly taken another monotone path inside the cell to $q$.   Lemma~\ref{exact:exact-vertical} implies that $q$ is in $P_{i,j}^{s'}$.

Lastly the path could take a diagonal shortcut and could similarly end with a monotone path inside $C_{i,j}$ to $q$. Lemma~\ref{exact:exact-diag} implies that $q$ then again is in $P_{i,j}^{s'}$.

Now let $q\in P_{i,j}^{s'}$ for $s'\leq s$. Then $q$ is either in $(i)$ $\nei[s']{i}{j}$, $(ii)$ $\ver[s']{i}{j}$, $(iii)$ $\dia[s']{i}{j}$ or $(iv)$ is reachable by a monotone path from some point $q'$ in one of the three preceding cases. Thus we can reduce this to the first three cases.

However Cases $(i)$ follow immediately since cells $C_{i,j-1}$ and $C_{i-1,j}$ have been handled correctly up to round $s'$, and thus $q$ must also be in $\dreach{s'}{i}{j}$.

For Cases $(ii)$ and $(iii)$ Lemma~\ref{exact:exact-vertical} and Lemma~\ref{exact:exact-diag} imply that $q$ must be in $\dreach{s'}{i}{j}$ respectively. Thus $\dreach{s}{i}{j}(T,B) = \bigcup_{s'\leq s}P_{i,j}^{s'}$.

As we store the reachable space and the leftmost point, this information is available for all upcoming iterations.
\end{proof}

\subsubsection{Running time}

\begin{lemma}\label{runningtime}
    Let $T$ and $B$ be two polygonal curves in the plane with overall complexity $n$, together with a distance threshold $\delta>0$. The algorithm described in Section \ref{exact-alg-subsection} has running time in $\cO(kn^{2k+2}\log n)$ and uses $\cO(kn^{2k+2})$ space.
\end{lemma}

\begin{proof}
Note that the sets $\nei[s]{i}{j}$, $\dia[s]{i}{j}$ and $\ver[s]{i}{j}$ computed by the algorithm are described as intersections of $\dfree{i}{j}$ with halfplanes, and unions of these. For a fixed $P^s_{i,j}$ we define  $n_{i,j,s}$ as the total number of such operations from which $P^s_{i,j}$ was obtained. As such, $\cO(n_{i,j,s})$ bounds the complexity of this set.

The complexity of $\nei[s]{i}{j}$ and $\ver[s]{i}{j}$ is constant. The complexity of $\dia{i}{j}$  is bounded by the sum of the complexities of all cells to the lower left:  \[n_{i,j,s}\in\cO\left(\sum_{0\leq a<i}\sum_{0\leq b<j}n_{a,b,s-1}\right).\] As $i,j\leq n$, and $s\leq k$, and $n_{a,b,0}\in\cO(1)$ for all $a$ and $b$, it holds that $n_{i,j,s}\in\cO(n^{2k})$.

Computing $\dia[s]{i}{j}$ takes $\cO(\sum_{a<i}\sum_{b<j}n_{a,b,s-1}\log n + n^2\log n)=\cO(n^{2k}\log n)$ time. This follows from the fact, that we compute $\cO(n)$ line-stabbing wedges, and for every cell $C_{a,b}$ with $a<i$ and $b<j$ we handle $n_{a,b,s-1}$ line segments based on $P^{s-1}_{a,b}$.
Computing $\nei{i}{j}$ takes $\cO(n_{i-1,j,s}+n_{i,j-1,s})=\cO(n^{2k})$ time, as we need to compute the reachability intervals from neighbouring cells. Computing $\ver{i}{j}$ takes $\cO(\sum_{b<j}n_{i,b,s-1})=\cO(n^{2k-1})$ time, as we need to compute the leftmost point $l^{s-1}_{i,j-1}$.
The space required to store $P^s_{i,j}$ as required by latter iterations and cells is in $\cO(n^{2k})$. 
Computing $Q(\nei{i}{j}\cup \ver{i}{j} \cup \dia{i}{j})$ takes linear time in the complexity of $\nei{i}{j}\cup \ver{i}{j} \cup \dia{i}{j}$, i.e. $\cO(n^{2k})$. As we do this for every cell in every round, the running time overall is $\cO(kn^{2k+2}\log n)$, and the space is bounded by $\cO(kn^{2k + 2})$.
\end{proof}

Lemma \ref{runningtime} together with Theorem \ref{exact:correctness} correctness then imply Theorem \ref{k-exact}.

The algorithm can also be used for the (unparameterized) shortcut Fr\'echet distance by choosing $k=n$, since there can be at most $n$ proper tunnels. We obtain the following corollary.

\begin{restatable}{corollary}{nexact}
\label{n-exact}
Let $T$ and $B$ be two polygonal curves in the plane with overall complexity $n$, together with a value $\delta>0$. There exists an algorithm with running time in $\cO\left(n^{2n+3}\log n\right)$ and space in $\cO\left(n^{2n+3}\right)$ that 
decides whether the shortcut Fréchet distance of $T$ and $B$ is at most $\delta$.
\end{restatable}

\section{Approximate decision algorithms}\label{c-packed}

In this section we first describe a $(3+\eps)$-approximation algorithm for the decision problem of the $k$-shortcut Fr\'echet distance of two polygonal curves in the plane.
The algorithm has a near-quadratic running time in $n$.
In Section \ref{approx:c-packed} we show that the algorithm can be modified to have running time near-linear in $n$, for the special class of $c$-packed curves.

\subsection{Description of the algorithm}\label{sec:apx-alg}

We describe how to modify the algorithm of Section \ref{exactchapter} to circumvent the exponential complexity of the reachable space and obtain a polynomial-time approximation algorithm. 

Let two polygonal curves $T$ and $B$ be given, together with a distance threshhold $\delta$ and approximation parameter $\eps$.
As before, the algorithm (see Algorithm \ref{approx:algorydm}) iterates over the cells of the free-space diagram and computes sets $\nei[s]{i}{j}$, $\ver[s]{i}{j}$, and $\dia[s]{i}{j}$ for each cell $C_{i,j}$. The main difference now is that, instead of computing the exact set of points that can be reached by a diagonal tunnel, we want to use an approximation for this set. For this, we define an approximate diagonal tunnel procedure, see further below.  This procedure is called with the rightmost point $r^{s-1}_{i-1,j-1}$ in $\bigcup_{a<i;b<j}P^{s-1}_{a,b}$, $\eps$ and distance parameter $3\delta$. Crucially, the set resulting from one call to the procedure has constant complexity and is sufficient to approximate the set $\dia[s]{i}{j}$. 
We then compute $P^s_{i,j}=Q(\nei[s]{i}{j}\cup\dia[s]{i}{j}\cup\ver[s]{i}{j})\cap\dfree{i}{j}$,  similarly to Section \ref{exactchapter}. From this we compute $(i)$ the leftmost point $l^{s}_{i,j}$ in $\bigcup_{b\leq j}P^s_{i,b}$ based on $P^s_{i,j}$ and $l^s_{i,j-1}$, $(ii)$ the rightmost point $r^{s}_{i,j}$ in $\bigcup_{a\leq i;b\leq j}P^s_{a,b}$ based on $P^s_{i,j}$, $r^s_{i-1,j}$ and $r^s_{i,j-1}$, and $(iii)$ the outgoing reachability intervals of $P^s_{i,j}$. We store these variables to be used in the next round. Finally, after $k$ rounds, we check if $(1,1)$ is contained in the computed set of reachable points.

Our approximate diagonal tunnel procedure makes use of a data structure by Driemel and Har-Peled, which is summarized in the following lemma. This data structure needs to be built once on $T$ in the beginning and is then available throughout the algorithm.

\begin{restatable}[Distance oracle \cite{Driemel2012JaywalkingYD}]{lemma}{distanceoracle}\label{distance-oracle}
    Given a polygonal curve $Z$ with $n$ vertices in $\bR^d$ and $\eps>0$, one can build a data structure $\cF_\eps$ in $\cO\left(\chi^2n\log^2 n\right)$ time, that uses $\cO\left(n\chi^2\right)$ space such that given a query segment $\overline{p\,q}$ and any two points $u$ and $v$ on the curve, one can $(1+\eps)$-approximate $d_\mathcal{F}(\overline{p\,q},Z[u,v])$ in $\cO\left(\eps^{-2}\log n\log \log n\right)$ time, where $\chi=\eps^{-d}\log\left(\eps^{-1}\right)$.
\end{restatable}

\begin{algorithm}
    \caption{Approximate Decider}
    \label{approx:algorydm}
    \begin{algorithmic}[1]
        \Procedure{ApproximateDecider'}{curve $T$, curve $B$, $\delta>0$, $0<\eps\leq 1$}
        \If{$||T(0) - B(0)||>\delta$ or $||T(1) - B(1)||>\delta$} 
		    Return '$d^k_\cS(T',B')>\delta$'\label{approx:assert}\EndIf
        \State Let $\cF_\eps$ be the data structure of Lemma~\ref{distance-oracle} built on $T$ with $\eps$.
        \State Let $g^s_r$, $g^s_l$, $\mathcal{A}^s$ and $\bar{\mathcal{A}}^s$ be arrays of size $n_1$ for each $0\leq s\leq k$.
		
		\For{$s=0,\ldots,k$}
		
		    \For{$j=1,\ldots,n_2$}
		        \State Copy array $\mathcal{A}^s$ into $\bar{\mathcal{A}}^s$
		        
		        \For{$i=1,\ldots,n_1$}
		        
		            \State Compute $\dfree{i}{j}$
		            
		            \If{$i=1,j=1$ and $s=0$}
		                \State $P^s_{i,j} = \dfree[\delta']{i}{j}$
		            \Else 
		                \Comment{Compute set of points directly reachable from neighboring cells}
		                \State Let $ I_v = \emptyset$ and $I_h = \emptyset$
                        \If{$j>1$}  Let $I_v$ be the incoming reachability interval from  $\bar{\cA}^s[i]$ \EndIf
                        \If{$i>1$}  Let $I_h$ be the incoming reachability interval from  $\cA^s[i-1]$  \EndIf
						\State Let $\nei{i}{j}=(Q(I_v)\cup  Q(I_h))\cap\dfree{i}{j}$
						\If{$s > 0$} 
						\Comment{Approximate set of points reachable by diagonal tunnel}
						\State Retrieve rightmost point $r$ in the lower left quadrant of $C_{i,j}$ from $g^{s-1}_r$.
						\State Let $\dia{i}{j}=\textsc{apxDiagonalTunnel}(r,(i,j),\varepsilon,3\delta)$
						\Comment{Compute set of points reachable by vertical tunnel}
						\State Retrieve leftmost point $l$ in the column below $C_{i,j}$ from $g^{s-1}_l$.
						\State Let $\ver{i}{j}=\textsc{verticalTunnel}(l,(i,j),\delta)$
						\Else
						\State Let $\dia{i}{j}=\emptyset$ and $\ver{i}{j}=\emptyset$
						\EndIf
						\Comment{Putting things together}  
						\State $P^s_{i,j}=Q(\nei{i}{j}\cup \dia{i}{j}\cup \ver{i}{j})\cap\dfree{i}{j}$
		            \EndIf
		            
		            \If{$P^s_{i,j}\neq \emptyset$}
					    \State Store the rightmost point of $P^s_{i,j}$ in $g^s_r$
						\State Store the leftmost point of $P^s_{i,j}$ in $g^s_l$
						\State Compute  outgoing reachability intervals and using $P^s_{i,j}$ and store them in ${\mathcal{A}}^s[i]$.
					\EndIf
		        \EndFor
		    
		    \EndFor
		\EndFor
		
		\If{$(1,1)\in \mathcal{A}^s[n_1]$}
			\State Return '$d^k_\cS(T',B')\leq3(1+\eps)^2\delta$' with $s\leq k$ shortcuts
		\Else
		    \State Return '$d^k_\cS(T',B')>\delta $' with at most $k$ shortcuts
		\EndIf
        \EndProcedure
        \Procedure{ApproximateDecider}{curve $T$, curve $B$, $\delta>0$, $0<\eps\leq 1$}
            \State Let $\eps'=\eps/9$
            \State Return ApproximateDecider'($T$,$B$,$\delta$,$\eps'$)
        \EndProcedure
    \end{algorithmic}
\end{algorithm}
\begin{algorithm}
\caption{Approximate Diagonal Tunnel}
\label{approx:algoapprox}
\begin{algorithmic}[1]
    \Procedure{apxDiagonalTunnel}{$(r_T,r_B)$, $(i,j)$, $\varepsilon$, $\delta$}
    \State Let $r=B(r_B)$ \Comment{$r$ is the starting point of the shortcut}
    \For{$t\in\left( \bG_{\frac{\delta\eps}{\sqrt{2}}} \cap \disk[3(1+\varepsilon)\delta]{v_{i}}\right)$}
        \State Query $\cF_\eps$ for the distance $d_{\cF_\eps}(\overline{r\,t},T[r_T,v_i])$ and store the answer in $\delta'$		        \If{$\delta'\leq(1+\varepsilon)^2\delta$}
		    \State Mark $t$ as eligible \Comment{$t$ is an eligible endpoint of a shortcut}
		\EndIf
    \EndFor
    
    \State Compute the convex hull $H$ of eligible points
    \If{$r\in H$}
        \State Return $C=\dfree[\delta]{i}{j}$
    \Else
        \State Let $U$ be the cone with apex $r$ formed by tangents $t_1$ and $t_2$ from $r$ to $H$
    	\State Let $p_i \in H$ be a supporting point of the tangent $t_i$ for $i\in\{1,2\}$
    	\State Let $L$ be the subchain of $\partial H$ with endpoints $p_1$ and $p_2$ which is facing $r$
    	\State Let $H' \subset U$ be the set bounded by $L$ and the rays supported by $t_1$ and $t_2$ facing away from $r$
    	\State Let $C'$ be the intersection of $H'$ with $e_j$
    	\State Return $C=(e_i \times C')\cap \dfree[\delta]{i}{j}$
    \EndIf
    \EndProcedure
\end{algorithmic}
\end{algorithm}

\begin{definition}[Grid]
    We define the scaled integer grid $\bG_\delta=\left\{(\delta x, \delta y) \mid (x,y)\in\bZ^2\right\}$.
\end{definition}
\subparagraph{Approximate diagonal tunnel procedure}
The procedure (see Algorithm \ref{approx:algoapprox}) is provided with parameters $\eps$, $\delta$, some $r'=(r_T,r_B)$ in cell $C_{a,b}$ and the edge $e_j$ that is associated with a cell $C_{i,j}$.
We want to compute a set of stabbers starting at $r=B(r_B)$ that contains every stabber through the disks $\disk[\delta]{v_{a+1}},\ldots,\disk[\delta]{v_i}$, and is contained in the set of all stabbers through disks of radius $(1+\varepsilon)^2\delta$ centered at the same vertices. 
We approximate this set of stabbers as follows.

We iterate over all grid points $t$ in the disk $\bG_{\frac{\delta\eps}{\sqrt{2}}}\cap\disk[(1+\eps)\delta]{v_i}$, and make queries to the precomputed distance oracle $\cF_\eps$ to determine if the Fr\'echet distance of the query segment $\overline{r\,t}$ to the subcurve of $T$ from $T(r_T)$ to $v_i$ is sufficiently small. We mark $t$ if the approximate distance returned by the data structure is at most $(1+\varepsilon)^2\delta$. We then compute the convex hull $H$ of all marked grid points, and the two tangents $t_1$ and $t_2$ of $H$ through $B(r_B)$.
The true set of endpoints of stabbers is approximated by the set $H'$ of points that lie inside and 'behind' the convex hull $H$, from the perspective of $r$.  Figure~\ref{fig-epsilon-stabber} illustrates this.
We then intersect $H'$ with the edge $e_j$ resulting in a single horizontal slab in $C_{i,j}$. This resulting set is then intersected with $\dfree{i}{j}$ and returned.

\begin{figure}
\begin{center}
	\includegraphics[width=0.7\textwidth]{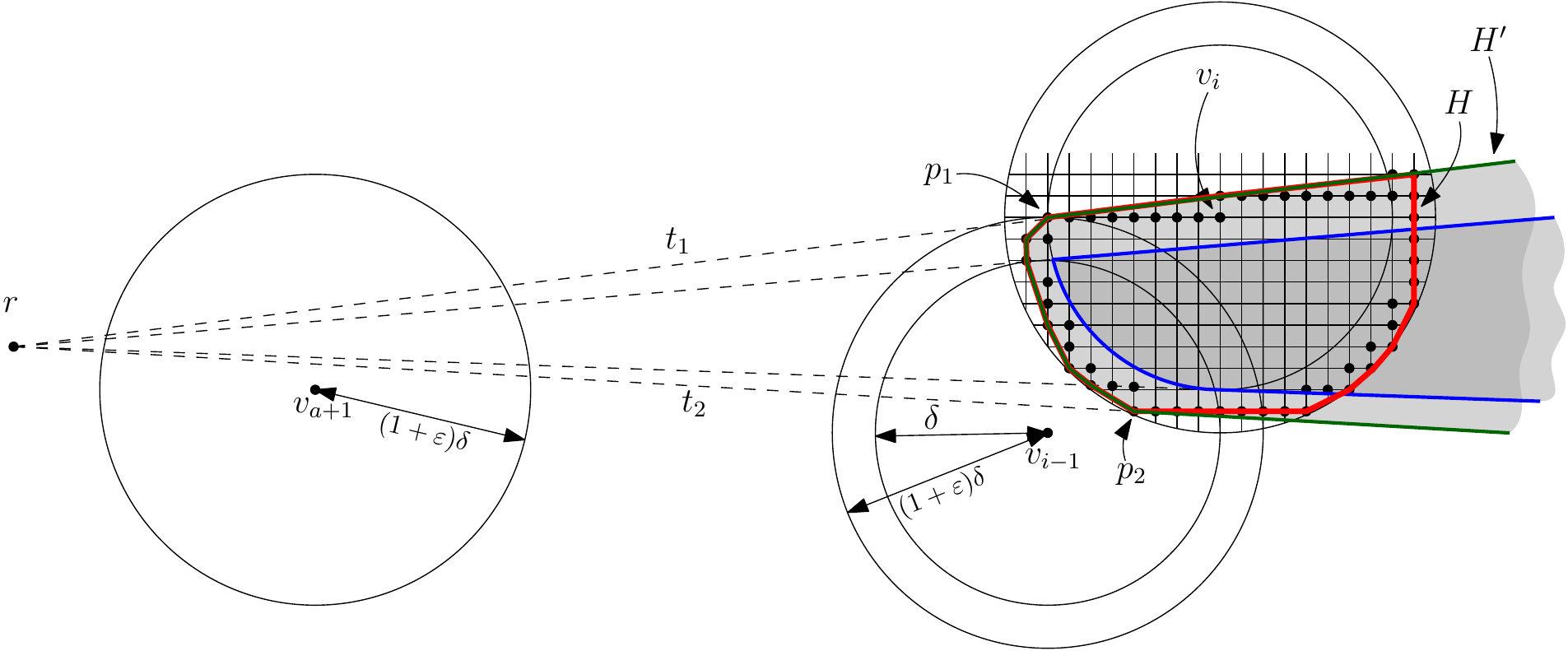}
	\end{center}
	\caption{Illustration to the approximate diagonal tunnel procedure. The true line-stabbing wedge for disks with radius $\delta$ is shown in blue. The convex hull of eligible grid points is shown in red. The approximate line stabbing wedge is shown in green.}
	\label{fig-epsilon-stabber}
\end{figure}

\subsection{Analysis of the approximation algorithm}\label{sec:apx-result}

We now analyse the described algorithm, namely the \textsc{ApproximateDecider} procedure.

\subsubsection{Correctness}

We argue that the structure of $P^s_{i,j}$ as approximated by the $\textsc{ApproximateDecider}$' procedure is indeed as claimed. Namely for all $i,j$ and $s$ it holds that $\dreach{s}{i}{j} \subset \bigcup_{0\leq s'\leq s}P^{s'}_{i,j}\subset \dreach[3(1+\eps)^2\delta]{s}{i}{j}$. We again consider any monotone path with $s$ proper tunnels ending in some cell and show the set inclusion by induction. 
Indeed, it suffices to consider the tunnel starting in the rightmost reachable point in the lower left quadrant of the cell, if we call the approximate diagonal tunnel procedure with a distance threshold $3\delta$. To prove correctness, we use  the following lemma  by Driemel and Har-Peled. The lemma states that if a feasible tunnel $\tau(r,q)$ costs more than $3\delta$ then any feasible tunnel $\tau(p,q)$ with $x_p\leq x_r$ costs more than $\delta$.

\begin{restatable}[monotonicity of tunnels \cite{Driemel2012JaywalkingYD}]{lemma}{centrallemma}\label{central}
	Given a value $\delta>0$ and two curves $T_1$ and $T_2$ such that $T_2$ is a subcurve of $T_1$, and given two line segments $\bar{B}_1$ and $\bar{B}_2$ such that $d_\cF(T_1,\bar{B}_1)\leq\delta$ and the start (resp. end) point of $T_2$ is within distance $\delta$ to the start (resp. end) point of $\bar{B}_2$, then $d_\cF(T_2,\bar{B}_2)\leq3\delta$.
\end{restatable}

In the following, we denote with $\{\disk[\delta]{v_i}\}_i$ a sequence of disks $\{\disk[\delta]{v_1},\ldots,\disk[\delta]{v_m}\}$ for some $m$.

\begin{lemma}\label{approx:convex}
Let $a,b_1,b_2\in\bR^d$ together with a sequence of vertices $v_1,\ldots,v_n$ be given. 
If $\overline{a\,b_1}$ stabs through disks $\{\disk[\delta]{v_i}\}_i$, 
 and $\overline{a\,b_2}$ stabs through $\{\disk[\delta]{v_i}\}_i$, then for any $t\in[0,1]$ the line segment $\overline{a\,b(t)}$ stabs through $\{\disk[\delta]{v_i}\}_i$, where $b(t)=(1-t)b_1 + tb_2$.
\end{lemma}
\begin{proof}
    Refer to Figure \ref{approx:convex-picture}.
    We can interpret the setting as a triangle with sides $(b_1-a)$, $(b_2-a)$, $(b_1-b_2)$, where the first two sides correspond to the original stabbers and the last side to $b(t)$. Note that any line segment $\overline{a\,b(t)}$ lies completely within this triangle with $(b_1-a)$ on the one and $(b_2-a)$ on the other side. Hence, for every $i$ and realising points $p_i$ of $\overline{a\,b_1}$ and $q_i$ of $\overline{a\,b_2}$, $p_i$ lies on the one and $q_i$ on the other side of $\overline{a\,b(t)}$. Since $\disk[\delta]{v_i}$ is convex and $p_i$ and $q_i$ are inside this disk, the intersection of $\overline{p_i\,q_i}$ and $\overline{a\,b(t)}$ is inside the disk as well. Call this intersection point $r_i$. The set $\{r_i\}_i$ are realising points for $\overline{a\,b(t)}$. This follows directly from the fact that $\{p_i\}_i$ and $\{q_i\}_i$ are ordered along their respective line segments, and thus $\overline{p_i\,q_i}$ never crosses another $\overline{p_j\,q_j}$. Thus for $i<j$, $r_i$ appears before $r_j$ along $\overline{a\,b(t)}$, implying the claim.
\end{proof}

\begin{lemma}\label{approx:triangle-stabber}
    Let $a_1,a_2,b_1,b_2\in\bR^2$ together with a sequence of vertices $v_1,\ldots,v_n$ be given. If $\overline{a_1\,b_1}$ stabs through $\{\disk[\delta]{v_i}\}_i$, and $\|a_1-b_1\|\leq\delta'$ and $\|a_2-b_2\|\leq\delta'$, then $\overline{a_1\,b_1}$ stabs through $\{\disk[\delta+\delta']{v_i}\}_i$.
\end{lemma}
\begin{proof}
    By Observation \ref{obs-trivial}, $d_\cF(\overline{a_1\,b_1},\overline{a_2\,b_2})\leq\delta'$, via the reparametrizaion $(f,g)$ with $f(t)=(1-t)a_1+ta_2$ and similarly $g(t)=(1-t)b_1+tb_2$.
    As $p=\overline{a_1\,b_1}$ stabs through $\{\disk[\delta]{v_i}\}_i$, there exist realising points $p_i$ along $p$, with $p_i$ lying in the $\delta$-disk centered at $v_i$.
    Then \[\|g(f^{-1}(v_i))-v_i\|\leq\|g(f^{-1}(v_i))-p_i\| + \|p_i-v_i\|\leq \delta' + \delta.\]
    Additionally, $q_i=g(f^{-1}(v_i))$ are ordered along $q=\overline{a_2\,b_2}$, proving the claim.
\end{proof}

\begin{lemma}\label{approx:eps-lemma}
    Given $r\in\bR^2$, $C_{i,j}$,$\eps$ and $\delta$ like in the \textsc{apxDiagonalTunnel} procedure.
	Denote by $S_{\delta}$ the set of endpoints of all $\delta$-stabbers (that is, stabbers through $\disk[\delta]{v_m}$ for $a+1\leq m\leq i$) on the edge $e$ starting at $r$ and let $C'$ be the point set computed by the algorithm. Then
	\[S_{\delta}\subseteq C' \subseteq S_{(1+\varepsilon)^2\delta}.\]
\end{lemma}
\begin{proof}
    Let $y\in C'$.
    Then $q=B(y)\in H'$ where $H'$ is set of points computed by the algorithm.
    Denote the intersection of $\overline{r\,q}$ and the boundary of $H'$ by $h$.
    $h$ is then a linear combination of at most two grid points whose stabbers from $r$ have been marked as eligible i.e. who are $(1+\varepsilon)^2\delta$-stabber.
    Hence, Lemma \ref{approx:convex} implies that $\overline{r\,q}$ is also a $(1+\varepsilon)^2\delta$-stabber, implying $C' \subseteq S_{(1+\varepsilon)^2\delta}$.
    
    Now let $q\in e$ be an arbitrary point such that $\overline{r\,q}$ is a $\delta$-stabber.
    Let $t$ be the last realising point of $\overline{r\,q}$. The line segment $\overline{r\,t}$ is a $\delta$-stabber and $t$ lies in $\disk[\delta]{v_i}$. We claim that $t$ lies in $H$.
    Consider the set $G= \bG_{\frac{\delta\eps}{\sqrt{2}}} \cap \disk[\varepsilon\delta]{t}$.
    By the properties of the grid, $t$ lies within the convex hull of $G$. Moreover $G \subset \disk[(1+\varepsilon)\delta]{v_i}$.
    Lemma \ref{approx:triangle-stabber} implies that $\overline{r\,t'}$ is a $((1+\varepsilon)\delta)$-stabber for any $t' \in G$.
    This in turn implies that for the first point $s'$ of $\overline{r\,t'}$ inside $\disk[\delta(1+\eps)]{v_a}$, $\overline{s'\,t'}$ is a $((1+\varepsilon)\delta)$-stabber, hence, $t'$ would have been marked as an eligible endpoint of a $((1+\varepsilon)^2\delta)$-stabber. Since $H$ is the convex hull of eligible points, it follows that $t\in\mathrm{conv}(G)\subset H$.
    Therefore $q\in H'$ and thus $q \in C'$.
\end{proof}

\begin{figure}[t]
    \centering
    \includegraphics[width=0.7\textwidth]{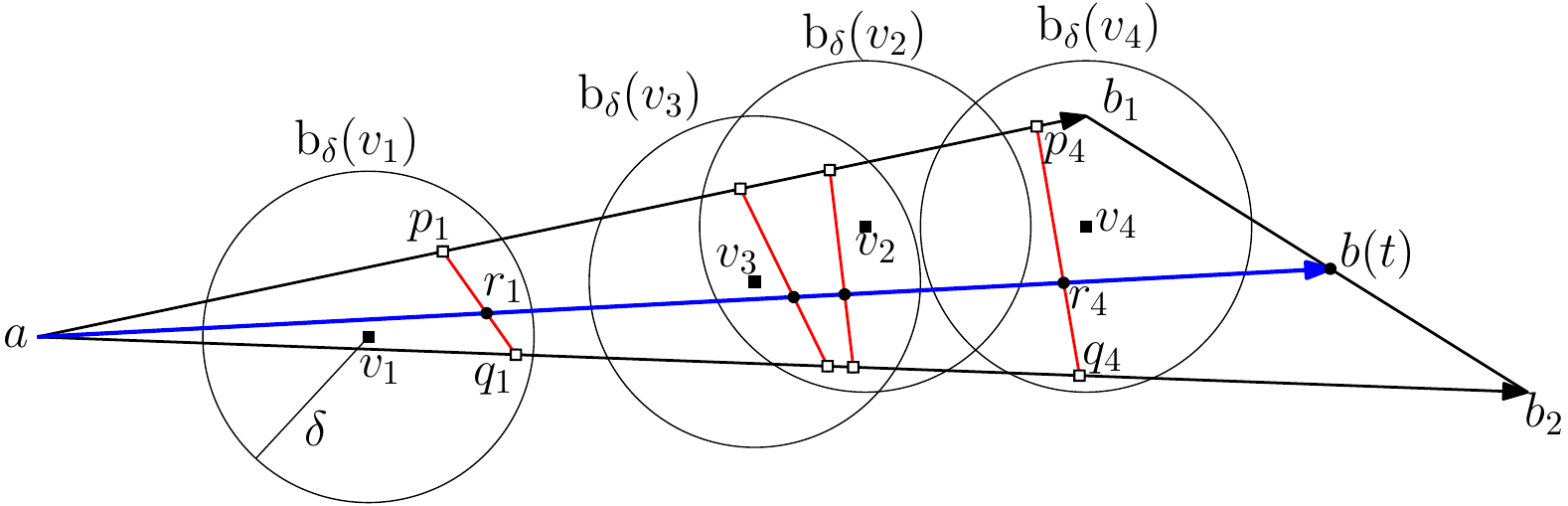}
    \caption{Linear interpolation between two $\delta$-stabbers starting in $a$. Illustrations to the proof of Lemma~\ref{approx:convex}. In blue $\overline{a\,b(t)}$, and in red $\overline{p_i\,q_i}$ is illustrated. Their intersections form the realising points $r_i$ of $\overline{a\,b(t)}$.}
    \label{approx:convex-picture}
\end{figure}

\begin{lemma}\label{approx:diagonal}
	For any $1\leq i\leq n_1$, $1\leq j\leq n_2$, $1\leq s\leq k$ let $\dia{i}{j}$ be the endpoints of diagonal tunnels as computed in the \textsc{ApproximateDecider}' procedure, and let $R=\bigcup_{a=1}^{i-1}\bigcup_{b=1}^{j-1}P_{a,b}^{s-1}$ be the set of reachable points by exactly $s-1$ proper tunnels in the lower left quadrant of the cell $C_{i,j}$. It holds that
	\begin{compactenum}[(i)]
		\item there exists a point $p\in R$ such that for any $q\in C_{i,j}$ the diagonal tunnel $\tau(p,q)$ has price $\prc(\tau(p,q)) \leq 3\delta$ then $q\in \dia{i}{j}$. If $q\in \dia{i}{j}$ then $\prc(\tau(p,q)) \leq 3(1+\eps)^2\delta$, and
		\item there exists no other $b\in C_{i,j}\setminus \dia{i}{j}$ that is the endpoint of a diagonal tunnel from $R$ with price at most $\delta$.
	\end{compactenum}
\end{lemma}
\begin{proof}
	The first part follows immediately from Lemma~\ref{approx:eps-lemma} together with the process described by the algorithm:
	The point $p$ is simply the rightmost point in $R$, which is maintained in $g'_r$ (by a trivial induction argument) at the time, where $C_{i,j}$ is processed.
	Assume $p=(x_p,y_p)$ lies in cell $C_{a,b}$.
    We call the \textsc{apxDiagonalTunnel} procedure with $p$ and the vertices $v_{a+1},\ldots,v_i$ between the $a$th and $i$th edge of the target curve.
	It returns points $q=(x_q,y_q)$ inside the $\delta$-free-space such that $\overline{B}[y_p,y_q]$ stabs through the sequence $\disk[3(1+\eps)^2\delta]{v_{a+1}},\ldots,\disk[3(1+\eps)^2\delta]{v_{i}}$. Since $p$ and $q$ are in the $\delta$-free-space of $B$ and $T$, $\|T(x_p)-B(y_p)\|\leq\delta$ and $\|T(x_q)-B(y_q)\|\leq\delta$ which together imply $d_{\cF}(\overline{B}[y_p,y_q],T[x_p,x_q])\leq3(1+\eps)^2\delta$.
	
	Assume for the sake of contradiction of the second part that such a point $b$ does exist and the start point of the shortcut is $s\in R$.
	Then by Lemma \ref{central} all tunnels $\tau(r,b)$ with $x_s<x_r$ have price at most $3\delta$.
	In particular $\prc(\tau(p,b))\leq3\delta$, but then $b$ would have been in $\dia{i}{j}$ already. 
\end{proof}

\begin{lemma}\label{approx:apxratiohelper}
	Given two polygonal curves $T$ and $B$ in the plane as well as parameters $\eps>0$ and $\delta>0$, the \textsc{ApproximateDecider}' computes a decision of either $d_\mathcal{S}^k(T,B)>\delta$ or $d_\mathcal{S}^k(T,B)\leq 3(1+\eps)^2\delta$.
\end{lemma}
\begin{proof}
    We show that \[\dreach[\delta]{s}{}{}(T,B)\subset \bigcup_{s'\leq s}P^{s'} \subset \dreach[3(1+\varepsilon)^2\delta]{s}{}{}(T,B),\]
	for $s\leq k$.
	
	This proof is by induction over the order of handled cells.
	We show the inclusions from the theorem for each cell, i.e. \[\dreach[\delta]{s}{i}{j}(T,B)\subset \bigcup_{s'\leq s}P^{s'}_{i,j} \subset \dreach[3(1+\varepsilon)^2\delta]{s}{i}{j}(T,B).\]
	Assume that $(0,0)\in\dfree[\delta]{1}{1}$, else the algorithm would have returned a correct decision in line \ref{approx:assert}.
	For $i=j=1$ we have that $P^0_{i,j}=\dfree[\delta]{1}{1}$ which is correct by convexity of $\dfree[\delta]{1}{1}$.
	For all other $s$ we have that $P^s_{1,1}=\emptyset$. 
	This follows from the fact that there is no points in the column below or in the lower left quadrant of $C_{1,1}$.
	Thus for $i=j=1$ we have $\bigcup_{s'\leq s}P^{s'}_{i,j} =  \dreach[\delta]{s}{i}{j}\subset \dreach[3(1+\varepsilon)^2\delta]{s}{i}{j}$.
	
	Consider the algorithm handling some cell $C_{i,j}$.
	By induction all cells $C_{\leq n_1,<j}$ and $C_{<i,j}$ and in particular $C_{i-1,j}$ and $C_{i,j-1}$ have been handled correctly up to $s$.
	Hence, their reachability intervals and left- and right most points can be/have been computed correctly and are stored in their respective arrays.
	We need to show that $\dreach[\delta]{s}{i}{j}\subset \bigcup_{s'\leq s}P^{s'}_{i,j}$.
	Thus let $q\in\dreach[\delta]{s}{i}{j}$ be the endpoint of a monotone path from $(0,0)$ walking monotonously through $\dfree[\delta]{i}{j}$ using $s'\leq s$ proper tunnels of cost $\delta$.
	There are three possibilities of how the path could have entered $C_{i,j}$.
	
	The path could have taken $s'$ shortcuts to enter a neighbouring cell and then walked into $C_{i,j}$ through its boundary at some point $a$.
	Since $C_{i-1,j}$ and $C_{i,j-1}$ have been handled correctly, $a$ is in the computed reachability interval of the neighbouring cell. Since the path must be monotone $q$ lies in the closed halfplane fixed at the lower left end of the reachability interval in the respective directions, thus $q$ is also in $P^{s'}_{i,j}$.
	
	Alternatively the path could have entered some cell $C_{i,l}$ with $s'-1$ shortcuts and then took a horizontal shortcut into $C_{i,j}$ for some $j<l$. By Lemma \ref{exact:exact-vertical} together with the induction hypothesis for $P_{i,<j}$ we have that $q$ is in $P^{s'}_{i,j}$.
	
	Similarly, if the path took a diagonal shortcut, we can apply Lemma \ref{approx:diagonal} together with the induction hypothesis for $P_{<i,<j}$, showing that $q$ is in $P^{s'}_{i,j}$ implying the left inclusion $\dreach[\delta]{s}{i}{j}\subset \bigcup_{s'\leq s}P^{s'}_{i,j}$.
	
	Now let us assume that $q\in P^{s'}_{i,j}$ for some $s'\leq s$. We want to show that $q\in\dreach[3(1+\varepsilon)^2\delta]{s}{i}{j}$.
	It must be that either (i) $q$ is in $\nei{i}{j}$, (ii) $q$ is in $\ver{i}{j}$, (iii) $q$ is in $\dia{i}{j}$, or (iv) $q$ is in the upper right quadrant of some point $p$, where $p$ satisfies (i), (ii) or (iii). Thus we can reduce this to the first $3$ cases.
	
	In Case (i) the claim follows immediately, because $P^{s'}_{i-1,j}$ and $P^{s'}_{i,j-1}$ have been computed correctly.
	
	In Case (ii) the claim follows immediately as well, due to the fact that $P^{s'}_{i,<j}$ have been computed correctly and the leftmost point is stored correctly in $\bar{g}_l^{s'-1}$, together with Lemma \ref{exact:exact-vertical}.
	
	Case (iii) follows rather straight forward as well, since $P^{s'}_{<i,<j}$ have been computed correctly and thus the rightmost point in the lower left quadrant of $C_{i,j}$ that was reachable by $s'-1$ shortcuts is correctly stored in $\bar{g}_r^{s'-1}$.
	By Lemma \ref{approx:diagonal} the \textsc{apxDiagonalTunnel} has precisely the guarantee that the endpoints of shortcuts are contained within the set of shortcuts with price at most $3(1+\varepsilon)^2\delta$, the claim follows as well.	
	Since $P^{s'}_{i,j}$ fulfils all these requirements and thus computes all its left- and right-most points and reachability intervals correctly, the induction follows.
	Hence, $\dreach[\delta]{s}{}{}(T',B')\subset \bigcup_{s'\leq s}P^{s'} \subset \dreach[3(1+\varepsilon)^2\delta]{s}{}{}(T',B')$.
	The algorithm output reflects the fact whether or not $(1,1)$ is in $P^{\leq k}$ proving the claim.
\end{proof}

\begin{theorem}\label{approx:apxratio}
    Let $T$ and $B$ be two polygonal curves in the plane with overall complexity $n$, together with values $0<\eps\leq1$ and $\delta>0$. The \textsc{ApproximateDecider} procedure correctly computes a decision of either $d_\mathcal{S}^k(T,B)>\delta$ or $d_\mathcal{S}^k(T,B)\leq (3+\eps)\delta$.
\end{theorem}
\begin{proof}
    This follows directly from the choice of $\eps'$ and Lemma \ref{approx:apxratiohelper}. This follows from the fact that $\eps \leq 1$ and $\eps'=\eps/9$ hold, which implies $d_\mathcal{S}^k(T,B)\leq3(1+\eps')^2\delta'<(3+\eps)\delta$.
\end{proof}

\subsubsection{Running time}

\begin{theorem}
Let $T$ and $B$ be two polygonal curves in the plane with overall complexity $n$, together with values $0<\eps\leq1$ and $\delta>0$. There exists an algorithm with running time in 
$ \cO\left(kn^2\eps^{-5}\log^2\left( n\eps^{-1} \right)\right) $
and space in $\cO\left(kn^2\eps^{-4}\log^2\left(\eps^{-1}\right)\right)$ which outputs one of the following: (i) $d^k_\mathcal{S}(T,B) \leq (3+\eps)\delta$ or (ii) $d^k_\mathcal{S}(T,B) > \delta$. In any case, the output is correct.
\end{theorem}
\begin{proof}
    We claim that the \textsc{ApproximateDecider} procedure fulfills these requirements.\\
    As $\eps' = \frac{\eps}{9}$, we can replace $\eps'$ with $\eps$ in the running time. For the precomputation we initialize the datastructure presented by Driemel and Har-Peled~\cite{Driemel2012JaywalkingYD} from Lemma \ref{distance-oracle}. This precomputation takes $\cO\left(\varepsilon^{-4}\log^2\left(\eps^{-1}\right)n\log^2(n)\right)$ time.
    We iterate over all $\cO(n^2)$ cells $k$ times, where in every iteration the only step that can not be handled in constant time, is invoking \textsc{apxDiagonalTunnel} procedure. This procedure iterates over $\cO\left(\varepsilon^{-2}\right)$ gridpoints, thus querries the data structure $\cO\left(\varepsilon^{-2}\right)$ times where each querry takes $\cO\left(\varepsilon^{-2}\log n\log\log n\right)$ time. Finally, we construct a convex hull and intersect it with a line, taking $\cO\left(\varepsilon^{-2}\log \varepsilon^{-1}\right)$ time as the complexity of the convex hull is $\cO\left(\varepsilon^{-2}\right)$. Thus the overall running time of the \textsc{apxDiagonalTunnel} procedure is $\cO\left(\varepsilon^{-4}\log n\log\log n\right)$. We call this procedure $\cO\left(kn^2\right)$ times.
    
    Thus the overall running time is
    
    \begin{align*}
    &\cO\left(\varepsilon^{-4}\log^2\left(\varepsilon^{-1}\right)n\log^2(n) + kn^2\varepsilon^{-1}\left(\varepsilon^{-4}\log n\log\log n\right)\right) \\
    =\;\;     &\cO\left(\varepsilon^{-5} n \log^2(n) + kn^2\varepsilon^{-5}\log n\log\log n\right) \\
    =\;\; &\cO\left(kn^2\eps^{-5} \log^2 \left(n \eps^{-1}\right) \right).
    \end{align*}
    
    The space follows directly from the space needed for the approximate distance data structure. All other data structures necessary for the algorithm use $\cO(n)$ or $\cO\left(\varepsilon^{-2}\right)$ space. Hence, the space is $\cO\left(n\varepsilon^{-4}\log^2\left(\varepsilon^{-1}\right)\right)$, as described in \cite{Driemel2012JaywalkingYD}. The correctness of the output is guaranteed by Theorem~\ref{approx:apxratio}. 
\end{proof}

\subsection{Modified algorithm for $c$-packed curves}\label{approx:c-packed}

In the case that the input curves are $c$-packed, for some constant $c$, we can modify the algorithm and achieve in near-linear running time in $n$. For this, we follow the approach of Driemel and Har-Peled~\cite{Driemel2010ApproximatingTF} to first simplify the curves. 

\begin{definition}[$\mu$-simplification \cite{Driemel2010ApproximatingTF}]
	Let a polygonal curve $X$ and a parameter $\mu>0$ be given. First mark the initial vertex of $X$ and set it as the current vertex. Now scan the polygonal curve from the current vertex until it reaches the first vertex that is in distance at least $\mu$ from the current vertex. Mark this vertex and set it as the current vertex. Repeat this until reaching the final vertex of the curve and also mark this final vertex. The $\mu$-simplification of $X$ denoted by $\simpl(X,\mu)$ is the resulting curve $X'$ that connects only the marked vertices in the order along $X$.
\end{definition}

For simplifications of $c$-packed curves the complexity of the free-space is linear:

\begin{definition}[Free-space complexity]
	Let $T$ and $B$ be two polygonal curves, and $\delta$ a given parameter. Define
	\[\cN_{\leq\delta}(T,B) = \#\left\{(i,j)\in\{1,\ldots,n_1\}\times\{1,\ldots,n_2\}\;\Big|\;\dfree{i}{j}\neq\emptyset\right\}\]
	as the number of cells in the parametric space, with non-empty $\delta$-free-space.
\end{definition}

\begin{observation}
	Given two polygonal curves $T$ and $B$ of total complexity $n_1 + n_2 = n$, then $\cN_{\leq\delta}(T,B)\leq n_1 n_2\in \cO(n^2)$.
\end{observation}

\begin{lemma}[\cite{Driemel2010ApproximatingTF}]\label{approx:c-pac-lin-comp}
	For any two $c$-packed curves $X$ and $Y$ in $\bR^d$ of total complexity $n$, and two parameters $0<\varepsilon<1$ and $\delta>0$, we have that $\mathcal{N}_{\leq \delta}(\simpl(X,\varepsilon\delta),\simpl(Y,\varepsilon\delta)) \leq n\left(9c + 6c\varepsilon^{-1}\right)= \cO\left(cn\varepsilon^{-1}\right)$.
\end{lemma}

\begin{corollary}\label{approx:coro-s}
	For all $s\geq 0$ we have that
	\[\mathcal{N}_{\leq s\delta}(\simpl(X,\varepsilon\delta),\simpl(Y,\varepsilon\delta)) \leq n\left(9c + 6cs\varepsilon^{-1}\right)= \cO\left(cn+scn\varepsilon^{-1}\right).\]
\end{corollary}

The following lemma by Driemel and Har-Peled shows that the shortcut Fr\'echet distance is approximately preserved under simplifications.

\begin{lemma}[\cite{Driemel2010ApproximatingTF}]\label{approx:approx-simp}
	Given a simplification parameter $\mu$ and two polygonal curves $X$ and $Y$, let $X'=\simpl(X,\mu)$ and $Y'=\simpl(Y,\mu)$ denote their $\mu$-simplifications respectively. For all $k\in\bN$ it holds that \[d^k_{\mathcal{S}}(X',Y') - 2\mu \leq d^k_{\mathcal{S}}(X,Y)\leq d^k_{\mathcal{S}}(X',Y') + 2\mu.\]
\end{lemma}

\subsubsection{Modifications}

The four major modifications we have to apply to Algorithm \ref{approx:algorydm}, in order the achieve near-linear running time, are the following. Instead of using $\eps'=\eps/9$ and $\delta$ over the course of the algorithm as approximation value and distance threshhold, we instead use $\eps'=\eps/20$ and $\delta'=\delta/(1-2\eps')$. This is in order to retrieve a $(3+\eps)$-approximate result from the algorithm.\\ Secondly we $\eps'\delta'$-simplify both input curves $T$ and $B$, such that $\dreach[\delta']{}{}{}(T',B')$ only has $\cO(cn\eps^{-1})$ non-empty cells by Lemma~\ref{approx:c-pac-lin-comp}.\\
Thirdly, instead of iterating over all cells, we only want to iterate over these non-empty cells. We solve this with an output-sensitive algorithm for computing the intersections of edges and the boundary of $\delta'$-neighbourhoods of these edges. For this we first compute the $\cO(n)$ boundaries of neighbourhoods of edges, whose geometric shape we refer to a ``capsule'' in $\cO(n)$ time. We then compute the intersections between all edges and capsules of $B$ and $T$ with a slight modification (to handle capsules) of the sweep line algorithm presented by Bentley and Ottmann~\cite{Bentley1979AlgorithmsfRaCGI}. From these intersections we can then reconstruct, which cells have non-empty $\delta$-free-space.
A detailed description of this straight-forward modification can be found in Appendix \ref{appendix:intersectionfinder}.\\
Lastly, in order to store and retrieve the left- and rightmost points in a column below and in the lower-left quadrant of a cell, we use two dimensional range trees described in \cite{deBerg2008}. Both storing and retrieving takes logarithmic time, but now we are able to retrieve these points, while only storing and updating these points, whenever we are in a non-empty cell.

\subsection{Analysis for $c$-packed curves}

We now turn to analysing the modified algorithm as described in Section~\ref{approx:c-packed}.

\subsubsection{Correctness}

\begin{theorem}\label{approx:c-packed-apxratio}
    Given two $c$-packed curves $T$ and $B$ in the plane, as well as parameters $0<\eps\leq1$ and $\delta>0$, the algorithm correctly computes a decision of either $d_\mathcal{S}^k(T,B)>\delta$ or $d_\mathcal{S}^k(T,B)\leq (3+\eps)\delta$.
\end{theorem}
\begin{proof}
    The algorithm defines $\eps'=\eps/20$, and $\delta'=\delta/(1-2\eps')$, and $\eps'\delta'$-simplifies $T$ and $B$, resulting in $T'$ and $B'$ respectively. As the main part of the algorithm is not modified, Lemma \ref{approx:apxratiohelper} guarantuees a correct decision of either  $d_\mathcal{S}^k(T',B')>\delta'$ or $d_\mathcal{S}^k(T',B')\leq 3(1+\eps')^2\delta'$. By Lemma \ref{approx:approx-simp}, this decision implies a correct decision of either $d_\mathcal{S}^k(T,B)>(1-2\eps')\delta'-2\eps'\delta'$ or $d_\mathcal{S}^k(T',B')\leq 3(1+\eps')^2\delta'+2\eps'\delta'$. From the choices of $\eps'$ and $\delta'$ as well as from the fact that $\eps<1$ it follows that $\delta'-2\eps'\delta'=\delta$, and $3(1+\eps')^2\delta'+2\eps'\delta'<3+\eps$, thus implying the claim.
\end{proof}

\subsubsection{Running Time}
In this section we analyse the running time of the algorithm.
We begin by proving, that we can find all non-empty cells in an output-sensitive manner.

\begin{lemma}\label{approx:bounded-intersection}
    Let $B$ and $T$ be two polygonal $c$-packed curves and $0<\varepsilon\leq 1$ and $\delta>0$ be given. Then the number of pairwise intersections in the set consisting of edges and $\delta$-neighbourhood boundaries of edges of $\simpl(B,\varepsilon\delta)$ and $\simpl(T,\varepsilon\delta)$  is in $\cO\left(cn\varepsilon^{-1}\right)$.
\end{lemma}
\begin{proof}
    This lemma is proven by repeated applications of Corollary \ref{approx:coro-s}.
    Let $B'=\simpl(B,\varepsilon\delta)$ and $T'=\simpl(T,\varepsilon\delta)$.
    For $s=0$, $B$ and $T$, Corollary \ref{approx:coro-s} implies that there are only $\cO\left(cn\varepsilon^{-1}\right)$ many intersections between edges of $B'$ and edges of $T'$.
    For $s=0$, $B$ and $B$, Corollary \ref{approx:coro-s} implies that there are only $\cO\left(cn\varepsilon^{-1}\right)$ many intersections between edges of $B'$ and edges of $B'$, similarly for $T'$.
    For $s=1$, $B$ and $T$, Corollary \ref{approx:coro-s} implies that there are only $\cO\left(cn\varepsilon^{-1}\right)$ many intersections between edges of $B'$ and neighbourhoods of edges of $T'$. Similarly for edges of $B'$ and neighbourhoods $B'$, and edges $T'$ and neighbourhoods $T'$.
    And lastly for $s=2$, Corollary \ref{approx:coro-s} implies the same for intersections of neighbourhoods and neighbourhoods, implying the claim.
\end{proof}

\begin{corollary}\label{approx:intersection-algorithm}
    Let $B$ and $T$ be two polygonal $c$-packed curves in the plane and $0<\varepsilon\leq1$ and $\delta>0$ be given. Then all $\cO\left(cn\varepsilon^{-1}\right)$ non-empty cells in the $\delta$-free-space of $\simpl(B,\varepsilon\delta)$ and $\simpl(T,\varepsilon\delta)$ can be found in $\cO\left(cn\varepsilon^{-1}\log\left(cn\varepsilon^{-1}\right)\right)$ time.
\end{corollary}
\begin{proof}
    We first compute the intersections of the set of edges of $B$ and boundaries of $\delta$-neighbourhoods of the edges of $T$.
    For this we can use an output-sensitive intersection-finding algorithm, such as the Bentley-Ottman algorithm \cite{Bentley1979AlgorithmsfRaCGI}, see also Appendix~\ref{appendix:intersectionfinder}. 
    From these intersections we can then reconstruct the edge-pairs of $B$ and $T$ with non-empty $\delta$-free-space in $\cO\left(cn\varepsilon^{-1}\log\left(cn\varepsilon^{-1}\right)\right)$ time, by tracing the edges of $B$ in the arrangement of $\delta$-neighborhoods of $T$. 
    Note that we also need to compute the edges that lie completely inside the $\delta$-neighbourhoods of other edges, however these are surrounded by edges intersecting (i.e. entering and leaving) the boundary of the same neighbourhood. Hence, we can process the edges of $B$ in the order along $B$ and find all edge-pairs of $B$ and $T$ with non-empty $\delta$-free-space.
\end{proof}

\ckapproximation*

\begin{proof}
    The non-trivial steps of the algorithm are: (i) Precomputation on the curves, (ii) finding all non-empty cells, (iii) iterating over these cells, (iv) the \textsc{apxDiagonalTunnel} procedure and (v) storing and restoring the rightmost gate.
    
    As $\eps' = \frac{\eps}{20}$, we can replace $\eps'$ with $\eps$ in the running time. For the precomputation we initialize the datastructure presented by Driemel and Har-Peled~\cite{Driemel2012JaywalkingYD} from Lemma \ref{distance-oracle}. This precomputation takes $\cO\left(\varepsilon^{-4}\log^2\left(\eps^{-1}\right)n\log^2(n)\right)$ time.
    In Corollary~\ref{approx:intersection-algorithm} we showed, that finding all intersections can be done in $\cO\left(cn\varepsilon^{-1}\log\left(cn\varepsilon^{-1}\right)\right)$ time. Sorting these intersections in $\cO\left(cn\varepsilon^{-1}\log\left(cn\varepsilon^{-1}\right)\right)$ time alphanumerically by the two indices, allows us to iterate over these cells as described in the algorithm.
    In Section~\ref{sec:apx-alg} we described the \textsc{apxDiagonalTunnel} procedure. This procedure iterates over $\cO\left(\varepsilon^{-2}\right)$ gridpoints, thus querries the data structure $\cO\left(\varepsilon^{-2}\right)$ times where each querry takes $\cO\left(\varepsilon^{-2}\log n\log\log n\right)$ time. Finally, we construct a convex hull and intersect it with a line, taking $\cO\left(\varepsilon^{-2}\log \varepsilon^{-1}\right)$ time as the complexity of the convex hull is $\cO\left(\varepsilon^{-2}\right)$. Thus the overall running time of the \textsc{apxDiagonalTunnel} procedure is $\cO\left(\varepsilon^{-4}\log n\log\log n\right)$. We call this procedure $\cO\left(kcn\varepsilon^{-1}\right)$ times, $k$ times for each nonempty cell.
    
    To store the rightmost gate in the lower left quadrant we can use two dimensional range trees as described in \cite{deBerg2008}. We build this tree with $\cO\left(cn\varepsilon^{-1}\right)$ points at the end of each outer loop storing all right- and left-most points for the next iteration in $\cO\left(cn\varepsilon^{-1}\log\left(cn\varepsilon^{-1}\right)\right)$ time. As we do this $k$ times, this results in an overall running time of $\cO\left(kcn\varepsilon^{-1}\log\left(cn\varepsilon^{-1}\right)\right)$, where the space used is $\cO\left(cn\varepsilon^{-1}\log\left(cn\varepsilon^{-1}\right)\right)$.
    
    Thus the overall running time is
    
    \begin{align*}
    &\cO\left(cn\varepsilon^{-1}\log\left(cn\varepsilon^{-1}\right) + \varepsilon^{-4}\log^2\left(\varepsilon^{-1}\right)n\log^2(n) + kcn\varepsilon^{-1}\left(\varepsilon^{-4}\log n\log\log n\right)\right) \\
    =\;\;     &\cO\left(cn\varepsilon^{-1}\log\left(n\varepsilon^{-1}\right) + \varepsilon^{-5} n \log^2(n) + kcn\varepsilon^{-5}\log n\log\log n\right) \\
    =\;\; &\cO\left(kcn\eps^{-5} \log^2 \left(n \eps^{-1}\right) \right).
    \end{align*}
    
    The space follows directly from the space needed for the approximate distance data structure. All other data structures necessary for the algorithm use $\cO(n)$ or $\cO\left(\varepsilon^{-2}\right)$ space. Hence, the space is $\cO\left(n\varepsilon^{-4}\log^2\left(\varepsilon^{-1}\right)\right)$, as described in \cite{Driemel2012JaywalkingYD}. The approximation ratio is guaranteed by Theorem~\ref{approx:c-packed-apxratio}. 
\end{proof}

\section{Hardness}\label{non-fpt}

We prove that deciding whether the $k$-shortcut Fréchet distance is less than or equal to a given value can not be done in $n^{o(k)}$ time, unless ETH fails. For this we construct a $(4k+2)$-shortcut Fréchet distance instance based on a $k$-Table-SUM instance, where the distance is exactly $1$ if and only if the $k$-Table-SUM instance has a solution and more than $1$ otherwise.

\begin{restatable}[$k$-Table-SUM]{definition}{kTableSUM}

We are given $k$ lists $S_1,\ldots,S_k$ of $n$ non-negative integers $\{s_{i,1},\ldots,s_{i,n}\}$ and a non-negative integer $\sigma$. We want to decide whether there are indices $\iota_1,\ldots,\iota_k$ such that $\sum_{i=1}^{k}s_{i,\iota_i} = \sigma$. We call $\sigma_j=\sum_{i=1}^{j}s_{i,\iota_i}$ the $j$th partial sum.
\end{restatable}

\subsection{General idea}

A $k$-Table-SUM instance consists of $k$ lists of integers and a target value and asks whether the target value can be rewritten as a sum of values, one from each list.
Based on such an instance we describe how to construct a $(4k+2)$-shortcut Fréchet distance instance consisting of the target curve $T$ and the base curve $B$ with the described property, that they have a distance of $1$ if and only if the underlying instance has a solution.

The target curve $T$ will lie on a horizontal line going to the right.
The set of points in $\bR^2$ which have a distance of at most $1$ to the target curve we will call the \textit{hippodrome}.
The base curve will consist of several horizontal edges going to the left on the boundary of the hippodrome.
All other edges of the base curve will lie outside the hippodrome.
Any shortcut curve of B that has Fréchet distance of at most $1$ to T we will call \textit{feasible}.
It is easy to see that any feasible shortcut curve must lie completely in the hippodrome.
Since any edge of the base curve inside the hippodrome lies on the boundary of it and is oriented in the opposite direction of the base curve, no feasible shortcut curve consists of any subcurve of the target curve.
Hence, every shortcut on a feasible shortcut curve has to start where the previous shortcut ended.
To restrict the set of feasible shortcut curves even further, we place so called \textit{twists} on the target curve.
Twists force shortcuts traversing it to go through precisely one point, called its \textit{focal point} or \textit{projection centre}.
For a simplified structural view of the curves refer to Figure \ref{hardness:structure}.
These twists are constructed by going a distance of $2$ to the left, before continuing rightwards.
We will not place any edges of the base curve too close to twists, so that a shortcut must be taken to traverse these.

Intuitively we can think of the horizontal edges of the base curve as mirrors that disperse incoming light in all directions and focal points as a wall with a hole, like in a pinhole camera.
A shortcut curve can be thought of as the path of a photon that tries to traverse this instance.
It bounces from mirror to mirror, always passing through a focal point.
A feasible shortcut curve exists if and only if it is possible to send a photon from the beginning of the base curve to the very end.

We can transport information throughout the instance by comparing two different feasible shortcut curves.
Assume both shortcut curves traverse the instance up to some edge, with the points of contact being some distance $d$ apart.
Then the points of contact on the next edge following shortcuts through a focal point will be $d$ apart again.
We could then keep track of a shortcut curve encoding a partial sum of $0$ as a reference point, where the distance from any shortcut curve to this reference curve encodes the partial sum of the particular shortcut curve.

\begin{figure}
	\includegraphics[width=\textwidth]{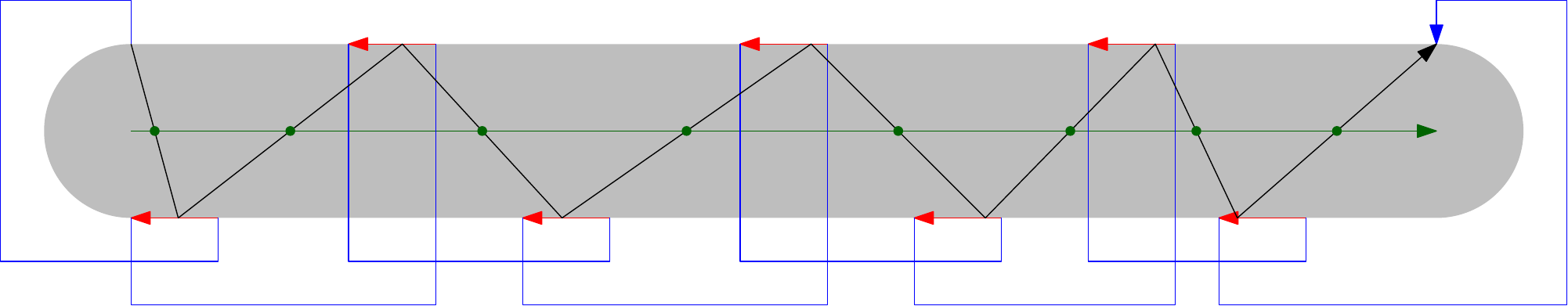}
	\caption{Simplified global layout of the target curve and its focal points in green, and the base curve consisting of red mirror edges and blue connector edges. A feasible shortcut curve is drawn in black, hippodrome in gray.}
	\label{hardness:structure}
\end{figure}
For two shortcut curves to take different paths we need to introduce a choice corresponding to taking an item from a list in the $k$-Table-SUM instance.
For this we place multiple edges, one for each item, at distances between $\frac{1}{2}$ and $1$ of the base curve instead of a single edge on the boundary of the hippodrome.
These can be thought of as semi-transparent mirrors.
Since the distance from these edges to the target curve may be less than $1$, it may happen that a feasible shortcut curve traverses the edge before taking the next shortcut.
Therefore the relative position along an edge no longer encodes precise values but approximates the partial sums.
We can introduce a scaling in the horizontal direction to contain this error.
A second problem that occurs is that edges may overlap in the vertical direction, such that photons may visit multiple edges.
We will fix this by stretching the instance even further.

\subsection{$k$-Table-SUM}
We begin by defining the $k$-Table-SUM problem and looking at equivalent variants of it which we want to work with.

\begin{definition}[$k$-SUM]
We are given a list $S$ of $n$ non-negative integers $\{s_{1},\ldots,s_{n}\}$ and a non-negative integer $\sigma$. We want to decide whether there is a subset $S'\subset S$ of size $k$, such that $\sum_{s\in S'}s = \sigma$.
\end{definition}

The following theorem is well-known. We provide a proof for the sake of completeness.

\begin{theorem}[Folklore]\label{hardness:thm1}
	Assuming the exponential time hypothesis, $k$-Table-SUM can not be solved in $2^{O(k)}\log(n)\,n^{o(k)}$ time, i.e. for fixed $k$, not in $\log(n)\, n^{o(k)}$ time.
\end{theorem}

\begin{proof}
The \textit{exponential time hypothesis} states that the well known $3$-SAT problem in $n$ variables can not be solved in $2^{o(n)}$ time \cite{Impagliazzo1999}.
Assuming the exponential time hypothesis, P\u{a}tra\c{s}cu and Williams in \cite{Patrascu2010OnPFSA} showed, that $k$-SUM cannot be solved in $n^{o(k)}$ time.

To reduce a $k$-SUM instance to a $k$-Table-SUM instance, we begin by randomly partitioning the original integer list into $k$ non-empty parts. With probability $k!/k^k>e^{-k}$ any given solution is then split, with one item in each of the $k$ lists. This can be derandomized, by computing a $k$-perfect family of hash functions, introduced by Schmidt and Siegel in \cite{Schmidt1990SpatialCoOk}. A family of such derandomizations has been introduced by Alon, Yuster and Zwick \cite{Alon1995ColorCoding}. The derandomization introduces a factor of $2^{O(k)}\log n$.
\end{proof}

In the construction step we are interested in a slight variation of the $k$-Table-SUM, where each table has a minimum value of $0$.
This is equivalent to the above stated $k$-Table-SUM problem by subtracting the minimum value of each list from every value of that list as well as the sum of all minimum values from $\sigma$.
Another slight modification we need to introduce is that all $k$ lists need to be sorted.
This reduction takes $\cO(kn\log n)$ time by for example $k$ applications of a suitable sorting algorithm.
During the construction we will always refer to this sorted version of $k$-Table-SUM.

\subsection{Construction}\label{hardness:construction}

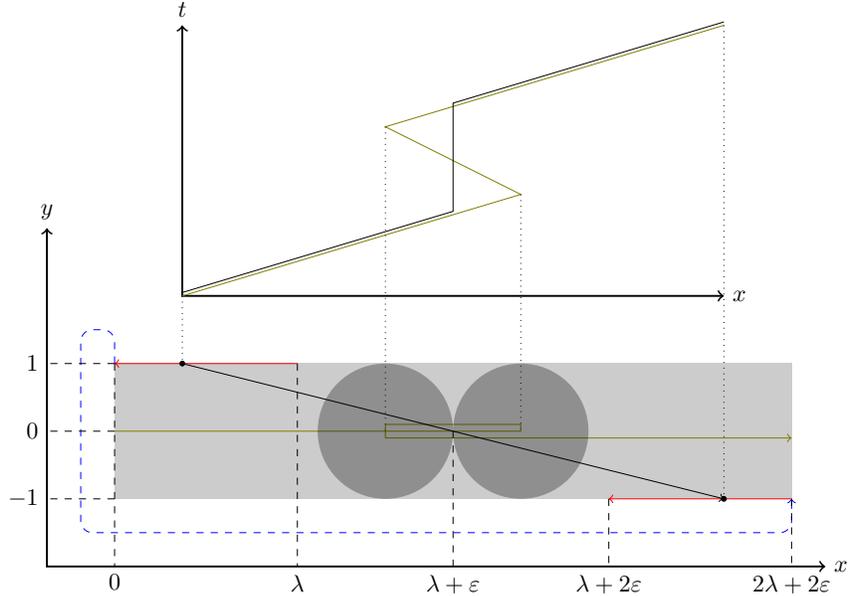
\begin{figure}
\centering
\scalebox{0.9}{
	\begin{tikzpicture}
	\begin{scope}[shift={(0,2)}]
	\draw [<->,thick] (-4,4) node (yaxis) [above] {$t$}
	|- (4,0) node (xaxis) [right] {$x$};
	\draw[green!50!red] (-4,0) -- (1,1.5) -- (-1,2.5) -- (4,4);
	\draw[black] (-4,0.05) -- (0,1.25) -- (0,2.85) -- (4,4.05);
	\draw[dotted] (1,1.5) -- (1,-2);
	\draw[dotted] (-1,2.5) -- (-1,-2);
	\draw[dotted] (-4,0) -- (-4,-1);
	\draw[dotted] (4,4) -- (4,-3);
	\end{scope}
	
	\fill[opacity = 0.2] (-5,-1) rectangle ++(10,2);
	\draw[->,green!50!red] (-5,0) -- (1,0) -- (1,0.1) -- (-1,0.1) -- (-1,-0.1) -- (5,-0.1);
	\draw[->,red] (-2.3,1) -- (-5,1);
	\draw[->,blue,dashed,rounded corners=5pt] (-5,1) -- (-5,1.5) -- (-5.5,1.5) -- (-5.5,-1.5) -- (5,-1.5) -- (5,-1);
	\draw[->,red] (5,-1) -- (2.3,-1);
	\fill[opacity=0.3](1,0) circle (1);
	\fill[opacity=0.3](-1,0) circle (1);
	\filldraw (-4,1) circle (1pt);
	\draw[->] (-4,1) -- (4,-1);
	\filldraw (4,-1) circle (1pt);
	
	\draw [<->,thick] (-6,3) node (yaxis) [above] {$y$}
	|- (5.5,-2) node (xaxis) [right] {$x$};
	\draw[dashed] (-5,1) -- (-5,-2);
	\draw[dashed] (-2.3,1) -- (-2.3,-2);
	\draw[dashed] (0,0) -- (0,-2);
	\draw[dashed] (5,-1) -- (5,-2);
	\draw[dashed] (2.3,-1) -- (2.3,-2);
	
	\draw[dashed] (-5,1) -- (-6,1);
	\draw[dashed] (-5,0) -- (-6,0);
	\draw[dashed] (-5,-1) -- (-6,-1);

	\node[below] at (-5,-2) {$0$};
	\node[below] at (-2.3,-2) {$\lambda$};
	\node[below] at (-0,-2) {$\lambda+\varepsilon$};
	\node[below] at (2.3,-2) {$\lambda+2\varepsilon$};
	\node[below] at (5,-2) {$2\lambda+2\varepsilon$};

	\node[left] at (-6,1) {$1$};
	\node[left] at (-6,0) {$0$};
	\node[left] at (-6,-1) {$-1$};
	\end{tikzpicture}
	}
	
	\caption{Traversal of a shortcut through a twist. The twist is not drawn correctly, to emphasize the structure of the twist.}
	\label{hardness:twist}
\end{figure}

In this section we describe the construction of the curves $T$ and $B$ given a $k$-Table-SUM instance.

We first describe the overall layout of the instance.
We will construct $k+2$ 'gadgets', an initialization gadget $g_0$, $k$ encoding-gadgets $g_1,\ldots,g_k$ that encode the individual lists of the $k$-Table-SUM instance and a terminal gadget $g_{k+1}$ used to verify that the target value $\sigma$ has been reached.
Each gadget $g_i$ will consist of two curves $T_i$ and $B_i$, which we concatenate to get $T$ and $B$ in the end.
We denote by $H_y$ the horizontal line at $y$ in $\bR^2$ and by $H_{> y}$ all points above $H_y$.
Similarly for $\geq,\leq$ and $<$.
And finally $H_{\geq a}^{<b} = H_{\geq a} \cap H_{<b}$.
The target curve $T$ will lie in $H_0$.

The base curve will have leftwards horizontal edges in $H_{\geq 1/2}^{\leq 1}$ and $H_{\geq -1}^{\leq -1/2}$, we will call \textit{mirror edges}.
All other edges of $B_i$ that connect these mirror edges we will call \textit{connector edges}.
The connector edges will mostly lie outside of the hippodrome.
The placement for connector edges that lie outside of the hippodrome is irrelevant.
We have to carefully look at any exception, since we want any feasible shortcut curve to only interact with the mirror edges.
Since all points lie on a small set of horizontal lines, we will occasionally denote the $x$-coordinate of a point and the point itself with the same variable but in different fonts.
For example the point $\textrm{x}^i_j$ has $x$-coordinate $x^i_j$.

The edges of the target curves $T_i$ will, with the exception of twists, be oriented in positive $x$-direction.
A twist centred at the focal point $(p,0)$ is a subcurve defined by the points $(p-1,0),(p+1,0)$ and $(p-1,0)$ connected by straight lines.
Around each focal point we introduce a buffer rectangle of length $2\varepsilon = 5$ and height $3$, where we let $\varepsilon$ be a global constant for the construction.
The base curve will never intersect these buffer zones, which is important for the twists to restrict the feasible shortcut curves as intended.

The instance will have two more global parameters. The first parameter $\gamma\geq1$ is a global scaling factor in $y$-direction, which ensures that feasible shortcut curves will never enter connector edges.
Furthermore it will ensure that the approximate encoding of two different partial sums will stay disjoint. The parameter $\gamma$ will be in $\cO(k)$. Lastly $\beta$ is a spacing parameter ensuring that edges are far enough apart from one another.

Before we look at the precise construction, let us convince ourselves of the correctness of twists.
For the following paragraph refer to Figure~\ref{hardness:twist}.
Assume we have two mirror edges of length $\lambda$, one placed from $(\lambda,1)$ to $(0,1)$, the other from $(2\lambda + 2\varepsilon,-1)$ to $(\lambda + 2\varepsilon,-1)$, which are connected by connector edges. 
We have a twist centred at $(\lambda + \varepsilon,0)$ on an otherwise rightwards facing target curve.
Assume furthermore that we have a partial feasible shortcut curve, which reaches some point $(p,1)$ on the first mirror edge.
Since the distance to the target curve is precisely $1$, any reparametrization with a distance at most $1$ for the shortcut Fréchet distance has to pair the point $(p,1)$ to $(p,0)$.
Since the target curve is oriented in the opposite direction to the mirror edge, the only way to continue the feasible shortcut curve is by a shortcut to the right.
It can not jump to any point on the first mirror edge, since all those points lie left of $(p,1)$.
The shortcut has to traverse the buffer zone of the twist.
And since there are no edges of the base curve in the buffer zone, the shortcut has to traverse it completely.
To analyse all shortcuts at a distance of at most $1$, we place two auxiliary disks centred at $(\lambda + \varepsilon\pm1,0)$ of radius $1$.
Any feasible shortcut curve traversing the buffer zone must traverse both of these disks, since otherwise no reparametrization can pair to the points $(\lambda + \varepsilon\pm1,0)$ at distance at most $1$, which are part of the target curve.

Since the twists first goes to $(\lambda + \varepsilon+1,0)$ and then to $(\lambda + \varepsilon-1,0)$, any feasible shortcut curve must also traverse the disks in this order.
The first disk lies to the right of the second disk, and we try to traverse these disks from the left.
The only possible way to traverse them with a straight line is through the intersection of the disks. 
And the only point in the intersection is exactly the focal point.
So any shortcut of a feasible shortcut curve that traverses the buffer zone of a twist must traverse its focal point.
A possible partial traversal is given in the upper plot in Figure~\ref{hardness:twist}.
Note that it is in $t$-$x$-space, corresponding to how the two points paired by the reparametrizations traverse the curves in the $x$-direction.

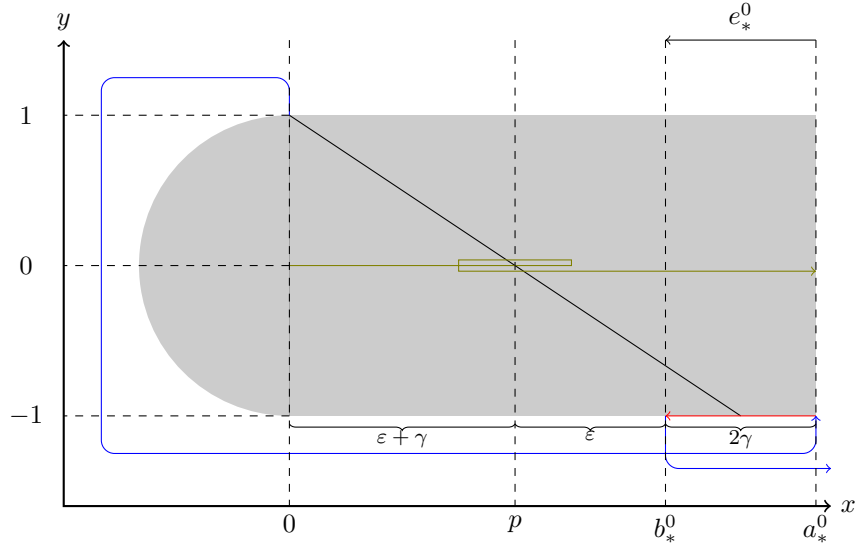
\begin{figure}
	\centering
	\begin{tikzpicture}
\def\one{2.0}
\def\g{1.0}
\def\eps{2.0}

\def\poreps{0.5}
\def\porinc{0.2}

\def\twisth{0.075}
\def\twistl{0.75}

\begin{scope}
\clip  ({-\eps-\g},\one) rectangle (-3*\one,-2*\one);
\fill [opacity = 0.2] ({-\eps-\g},0) circle(\one);
\end{scope}

\coordinate (P) at (0,0);
\coordinate (s) at ({-\eps-\g},\one);
\coordinate (t) at ({\eps+\g},-\one);
\coordinate (a1) at ({\eps},-\one);
\coordinate (b1) at ({\eps+2*\g},-\one);

\fill [opacity = 0.2] (s) rectangle (b1);

\draw (s) -- (t);
\draw[->,red] (b1) -- (a1);
\draw[->,blue,rounded corners=5pt] (s) -- ($(s) + (0,\poreps)$) -- ($(s) + (-\one-\poreps,\poreps)$) -- ($(s) + (-\one-\poreps,-2*\one - \poreps)$) -- ($(b1) - (0,\poreps)$) -- (b1);

\draw[->,blue,rounded corners=5pt] (a1) -- ($(a1) + (0,-\poreps - \porinc)$) -- ($(b1) + (\porinc,-\poreps-\porinc)$);

\draw[decoration={brace,raise=3pt,mirror},decorate] ({-\eps-\g},-\one) -- node[midway,below=2pt] {\footnotesize $\varepsilon+\gamma$} (0,-\one);

\draw[decoration={brace,raise=3pt,mirror},decorate] (0,-\one) -- node[midway,below=2pt] {\footnotesize $\varepsilon$} (a1);

\draw[decoration={brace,raise=3pt,mirror},decorate] (a1) -- node[midway,below=2pt] {\footnotesize $2\gamma$} (b1);


\draw[->,green!50!red] ($(s) - (0,\one)$) -- ($(P) + (\twistl,0)$) --  ($(P) + (\twistl,\twisth)$) -- ($(P) + (-\twistl,\twisth)$) -- ($(P) - (\twistl,\twisth)$) --  ($(b1) + (0,\one-\twisth)$);

\draw [<->,thick] ($(s) + (-\one-2*\poreps,2*\poreps)$) node (yaxis) [above] {$y$}
|- ($(b1) + (\porinc,-2*\poreps-\porinc)$) node (xaxis) [right] {$x$};

\node at (-3*\one - \poreps, 0) {$0$};
\node at (-3*\one - \poreps, -\one) {$-1$};
\node at (-3*\one - \poreps, \one) {$1$};

\draw [dashed] (-1.5*\one,0) -- (-3*\one,0);
\draw [dashed] (-1.5*\one,\one) -- (-3*\one,\one);
\draw [dashed] (-1.5*\one,-\one) -- (-3*\one,-\one);

\draw [dashed] (s |- 52,\one+2*\poreps) -- (s |- 52,-\one-2*\poreps-\porinc);
\draw [dashed] (P |- 52,\one+2*\poreps) -- (P |- 52,-\one-2*\poreps-\porinc);
\draw [dashed] (a1 |- 52,\one+2*\poreps) -- (a1 |- 52,-\one-2*\poreps-\porinc);
\draw [dashed] (b1 |- 52,\one+2*\poreps) -- (b1 |- 52,-\one-2*\poreps-\porinc);
\draw [->] (b1 |- 52,\one+2*\poreps) -- (a1 |- 52,\one+2*\poreps)  node [midway, above] {$e^0_*$};

\node[below] at (s |- 52,-\one-2*\poreps-\porinc) {$0$};
\node[below] at (P |- 52,-\one-2*\poreps-\porinc) {$p$};
\node[below] at (a1 |- 52,-\one-2*\poreps-\porinc) {$b^0_*$};
\node[below] at (b1 |- 52,-\one-2*\poreps-\porinc) {$a^0_*$};
\end{tikzpicture}
	\caption{Construction of the Initialization gadget. The first forced shortcut is drawn in black. Mirror edges are red, connector edges are blue, and the target curve is green.}
	\label{hardness:init}
\end{figure}

\subsubsection{Initialization gadget $g_0$}
For the construction refer to Figure~\ref{hardness:init}.
Both curves $T_0$ and $B_0$ will start at $x$-coordinate $0$ placing the start point for the base curve at $(0,1)$, and the start point for the target curve at $(0,0)$.
The target curve will go rightwards, up to the first twist centred at $(\varepsilon + \gamma,0)$ and continue rightwards after that.
The base curve will immediately leave the hippodrome to the left and connect to the first mirror edge from $\textrm{a}^0_* = (3\gamma+2\varepsilon,-1)$ to $\textrm{b}^0_* = (\gamma+2\varepsilon,-1)$.

\subsubsection{Encoding gadget $g_i$}
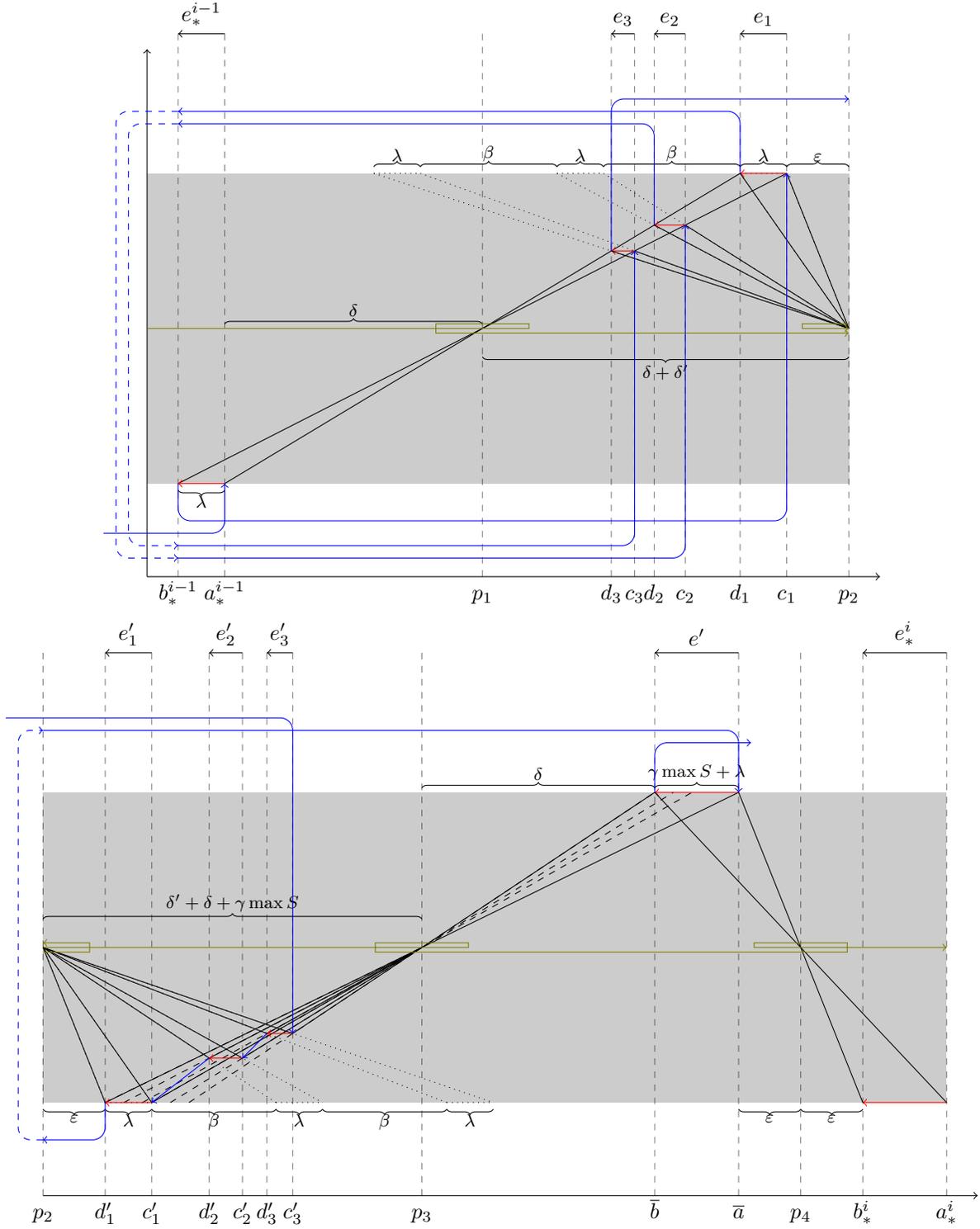
\begin{figure}
	\centering
	\makebox[\textwidth][c]{	\begin{tikzpicture}
	
	\def\one{2.5}
	\def\l{0.75}
	\def\b{2.2}
	\def\n{3} 
	\def\eps{1}
	\def\d{\l+\eps}
	\def\del{(\n-1)*(\l+\b)-(\d)}
	\def\roun{2mm}
	\def\coordoffset{0.5}
	
	\def\desoffset{1.5}
	
	\def\twisth{0.075}
	\def\twistl{0.75}

	\pgfmathdeclarefunction{a}{1}{%
		\pgfmathparse{\del-(#1-1.0)*(\b+\l)}%
	}
	\pgfmathdeclarefunction{b}{1}{%
		\pgfmathparse{\del-(#1-1.0)*(\b+\l)+\l}%
	} 
	\pgfmathdeclarefunction{x}{1}{%
		\pgfmathparse{(\del)*(\del+\d)/(\del+\d+(#1-1.0)*(\l+\b))}%
	}
	\pgfmathdeclarefunction{y}{1}{%
		\pgfmathparse{(\del+\l)*(\del+\d)/(\del+\d+(#1-1.0)*(\l+\b))}%
	}
	
	\coordinate (P) at (0,0);
	\coordinate (Q) at ({\del+\d},0);
	
	\draw[opacity=0.5,dashed](P |- 52,\one+1.5*\desoffset) -- (P |- 52,-\one-\desoffset);
	\node[below,text height=0.7em] at (P |- 52,-\one-\desoffset) {$p_1$};
	\draw[opacity=0.5,dashed](Q |- 52,\one+1.5*\desoffset) -- (Q |- 52,-\one-\desoffset);
	\node[below,text height=0.7em] at (Q |- 52,-\one-\desoffset) {$p_2$};
	
	\fill[opacity = 0.2] ($(\eps-\coordoffset,-\one) - (Q)$) rectangle ($(Q) + (0,\one)$);
	
	\coordinate (null) at ($(\eps-\coordoffset,-\one) - (Q)$);

	\draw[->](null |- 52,-\one-\desoffset) -- ($(null |- 52,\one+\desoffset) + (0,\coordoffset)$);
	\draw[->](null |- 52,-\one-\desoffset) -- ($(Q |- 52,-\one-\desoffset) + (\coordoffset,0)$);

	\foreach \i in {1,...,\n}
	{
		\coordinate (a\i) at ({a(\i)},\one);
		\coordinate (b\i) at ({b(\i)},\one);
		\draw[dotted] (a\i) -- (b\i);
	}
	
	\foreach \i in {1,...,\n}
	{
		\draw[decoration={brace,raise=3pt},decorate] (a\i) -- node[midway,above=2pt] {\footnotesize $\lambda$} (b\i);
	}
	
	\foreach \i in {2,...,\n}
	{
		\draw[decoration={brace,raise=3pt},decorate] (b\i) -- node[midway,above=2pt] {\footnotesize $\beta$}($ (b\i) + (\b,0) $);
	}
	
	\draw[decoration={brace,raise=3pt},decorate] (b1) -- node[midway,above=2pt] {\footnotesize $\varepsilon$}({\del+\d},\one);
	
	\draw[decoration={brace,raise=13pt,mirror},decorate] (P) -- node[midway,below=13pt] {\footnotesize $\delta+\delta'$}(Q);
	
	\draw[decoration={brace,raise=2pt,mirror},decorate] (P) -- node[midway,above=2pt] {\footnotesize $\delta$}($(0,\one) - (a1)$);
	
	\draw (P) -- (a1);
	\draw (P) -- (b1);
	
	\foreach \i in {1,...,\n}
	{
		
		\coordinate (x\i) at ({x(\i)},{(\one*x(\i))/(\del)});
		\coordinate (y\i) at ({y(\i)},{(\one*y(\i))/(\del+\l)});
		\draw[->,red] (y\i) -- (x\i);
		\draw[dotted] (a\i) -- (Q);
		\draw[dotted] (b\i) -- (Q);
		\draw (Q) -- (intersection of a\i--Q and P--a1);
		\draw (Q) -- (intersection of b\i--Q and P--b1);

		\draw[opacity=0.5,dashed](x\i |- 52,\one+1.5*\desoffset) -- (x\i |- 52,-\one-\desoffset);
		\node[below,text height=0.7em] at (x\i |- 52,-\one-\desoffset) {$d_\i$};
		\draw[opacity=0.5,dashed](y\i |- 52,\one+1.5*\desoffset) -- (y\i |- 52,-\one-\desoffset);
		\node[below,text height=0.7em] at (y\i |- 52,-\one-\desoffset) {$c_\i$};
		
		\draw[->](y\i |- 52,\one+1.5*\desoffset) -- (x\i |- 52,\one+1.5*\desoffset) node [midway, above] {$e_\i$};
	}
	
	\coordinate (s1) at ($(0,0) - (b1)$);
	\coordinate (s2) at ($(0,0) - (a1)$);
	
	\draw[opacity=0.5,dashed](s1 |- 52,\one+1.5*\desoffset) -- (s1 |- 52,-\one-\desoffset);
	\node[below,text height=0.7em] at (s1 |- 52,-\one-\desoffset) {$b^{i-1}_*$};
	\draw[opacity=0.5,dashed](s2 |- 52,\one+1.5*\desoffset) -- (s2 |- 52,-\one-\desoffset);
	\node[below,text height=0.7em] at (s2 |- 52,-\one-\desoffset) {$a^{i-1}_*$};
	\draw[->](s2 |- 52,\one+1.5*\desoffset) -- (s1 |- 52,\one+1.5*\desoffset) node [midway, above] {$e^{i-1}_*$};
	
	\draw (P) -- (s2);
	\draw (P) -- (s1);
	\draw[->,red] (s2) -- (s1);		
	\draw[decoration={brace,raise=3pt},decorate] (s2) -- node[midway,below=2pt] {\footnotesize $\lambda$} (s1);
	
	\draw[->,green!50!red] ($(s1) + (-\coordoffset,\one)$) -- ($(P) + (\twistl,0)$) --  ($(P) + (\twistl,\twisth)$) -- ($(P) + (-\twistl,\twisth)$) -- ($(P) - (\twistl,\twisth)$) -- ($(Q) - (0,\twisth)$);
	\draw[->,green!50!red]($(Q) + (0,\twisth)$) -- ($(Q) + (-\twistl,\twisth)$) -- ($(Q) - (\twistl,0)$)-- (Q);
	
	\def\portalx{0-(b(1))}	
	\def\porbuff{\one+0.6}
	\def\porinc{0.2}
	\def\pordotbuff{0.6}
	
	\foreach \i [count=\j from 1] in {2,...,\n}
	{
		\coordinate (p\i) at ({\portalx},{\porbuff + (\n-\j)*\porinc});
		\coordinate (q\i) at ({\portalx},{0-(\porbuff + (\n-\j+1)*\porinc)});
		\coordinate (w\i) at ($(p\i) - ({\pordotbuff+(\n-\j)*\porinc},0)$);
		\draw[->,blue,rounded corners=\roun] (x\j) -- (x\j|-p\i) -- (p\i);
		\draw[->,blue,dashed,rounded corners=\roun] (p\i) -- (w\i) -- (w\i |- q\i) -- (q\i);
		\draw[->,blue,rounded corners=\roun] (q\i) -- (y\i|-q\i) -- (y\i);
	}
	
	\coordinate (p1) at ({\portalx},{\porbuff+(\n)*\porinc)});
	\coordinate (q0) at ({(\portalx)-(\pordotbuff)-(\n*\porinc)},{0-(\porbuff+\porinc)});
	\coordinate (q-1) at ({(\portalx)-(\pordotbuff)-(\n*\porinc)},{0-(\porbuff)});
	\draw[->,blue,rounded corners=\roun] (q0) -- (s2|-q0) -- (s2);
	\draw[->,blue,rounded corners=\roun] (s1) -- (s1|-q-1) -- (y1|-q-1) -- (y1);
	\draw[->,blue,rounded corners=\roun] (x\n) -- (x\n|-p1) -- (Q|-p1);
	
	\end{tikzpicture}}
	\makebox[\textwidth][c]{\begin{tikzpicture}

	\def\one{2.5}
	\def\l{0.75}
	\def\b{2}
	\def\n{3} 
	\def\roun{2mm}
	\def\desoffset{1.5}
	\def\coordoffset{0.5}
	
	\def\twisth{0.075}
	\def\twistl{0.75}

	\def\valscale{0.3}
	
	\def\n{3} 
	\def\values{{2.0},{1.0},{0.0}} 
	\def\maxval{2} 
	
	\def\eps{1.0}
	\def\d{\l+\eps}
	\def\del{(\n-1)*(\l+\b)-(\d)}
	
	\def\poff{\del+\d}

	\pgfmathdeclarefunction{f}{1}{%
	    \pgfmathparse{\poff+\eps+(#1-1.0)*(\b+\l)}%
	}
	\pgfmathdeclarefunction{g}{1}{%
 	   \pgfmathparse{\poff+\eps+(#1-1.0)*(\b+\l)+\l}%
	}
	
	\coordinate (P) at ({\poff},0);
	\coordinate (Q) at ({\d+\valscale*\maxval+\del+\poff},0);
	\coordinate (R) at ($(Q) + ({\del + \eps + \l + \valscale*\maxval},0)$);

	\draw[opacity=0.5,dashed](P |- 52,\one+1.5*\desoffset) -- (P |- 52,-\one-\desoffset);
	\node[below,text height=0.7em] at (P |- 52,-\one-\desoffset) {$p_2$};
	\draw[opacity=0.5,dashed](Q |- 52,\one+1.5*\desoffset) -- (Q |- 52,-\one-\desoffset);
	\node[below,text height=0.7em] at (Q |- 52,-\one-\desoffset) {$p_3$};
	\draw[opacity=0.5,dashed](R |- 52,\one+1.5*\desoffset) -- (R |- 52,-\one-\desoffset);
	\node[below,text height=0.7em] at (R |- 52,-\one-\desoffset) {$p_4$};
	
	\coordinate (t1) at ($(Q) + ({\del},{\one})$);
	\coordinate (t2) at ($(Q) + ({\del+\l+\valscale*\maxval},{\one})$);

	\fill[opacity = 0.2] ($(P) + (0,\one)$) rectangle ($(R) + (R) - (t1)$);
	
	\foreach \i in {1,...,\n}
	{
		\coordinate (c\i) at ({f(\i)},-\one);
		\coordinate (d\i) at ({g(\i)},-\one);
		\draw[dotted] (c\i) -- (d\i);
		\draw[dotted] (P) -- (d\i);
		\draw[dotted] (c\i) -- (P);
	}
	
	\foreach \v [count=\i] in \values
	{
		\coordinate (r\i) at ($(c1) + ({\valscale*(\maxval-\v)},0)$);
		\coordinate (s\i) at ($(d1) + ({\valscale*(\maxval-\v)},0)$);
		\draw[dashed] (s\i) -- ($(Q) + (Q) - (s\i)$);
		\draw[dashed] (r\i) -- (Q);
	}
	
	\foreach \i in {1,...,\n}
	{
		\coordinate (xs\i) at (intersection of c\i--P and r\i--Q);
		\coordinate (ys\i) at (intersection of d\i--P and s\i--Q);
		\draw (P) -- (xs\i) -- (Q);
		\draw (P) -- (ys\i) -- (Q);
		\draw[->,red] (ys\i) -- (xs\i);

		\draw[opacity=0.5,dashed](xs\i |- 52,\one+1.5*\desoffset) -- (xs\i |- 52,-\one-\desoffset);
		\node[below,text height=0.7em] at (xs\i |- 52,-\one-\desoffset) {$d'_\i$};
		\draw[opacity=0.5,dashed](ys\i |- 52,\one+1.5*\desoffset) -- (ys\i |- 52,-\one-\desoffset);
		\node[below,text height=0.7em] at (ys\i |- 52,-\one-\desoffset) {$c'_\i$};
		
		\draw[->](ys\i |- 52,\one+1.5*\desoffset) -- (xs\i |- 52,\one+1.5*\desoffset) node [midway, above] {$e'_\i$};
	}
	
	\draw (Q) -- (t1);
	\draw (Q) -- (t2);
	\draw[->,red] (t2) -- (t1);

	\draw[opacity=0.5,dashed](t2 |- 52,\one+1.5*\desoffset) -- (t2 |- 52,-\one-\desoffset);
	\node[below,text height=0.7em] at (t2 |- 52,-\one-\desoffset) {$\overline{a}$};
	\draw[opacity=0.5,dashed](t1 |- 52,\one+1.5*\desoffset) -- (t1 |- 52,-\one-\desoffset);
	\node[below,text height=0.7em] at (t1 |- 52,-\one-\desoffset) {$\overline{b}$};
	\draw[->](t2 |- 52,\one+1.5*\desoffset) -- (t1 |- 52,\one+1.5*\desoffset) node [midway, above] {$e'$};
	
	
	\draw[decoration={brace,raise=3pt},decorate] (R |- 52,-\one) -- (t2 |- 52,-\one)  node[midway,below=2pt] {\footnotesize $\varepsilon$}({\poff},-\one);
	
	\foreach \i in {1,...,\n}
	{
		\draw[decoration={brace,raise=3pt,mirror},decorate] (c\i) -- node[midway,below=2pt] {\footnotesize $\lambda$} (d\i);
	}
	
	\foreach \i in {2,...,\n}
	{
		\draw[decoration={brace,raise=3pt},decorate] (c\i) -- node[midway,below=2pt] {\footnotesize $\beta$}($ (c\i) - (\b,0) $);
	}
	
	\draw[decoration={brace,raise=3pt},decorate] (c1) -- node[midway,below=2pt] {\footnotesize $\varepsilon$}({\poff},-\one);
	
	\draw[decoration={brace,raise=13pt},decorate] (P) -- node[midway,above=13pt] {\footnotesize $\delta'+\delta+\gamma\max S$}(Q);
	
	\draw[decoration={brace,raise=2pt},decorate] ($(Q)-(0,-\one)$) -- node[midway,above=2pt] {\footnotesize $\delta$}(t1);
	
	\draw[decoration={brace,raise=2pt},decorate] (t1) -- node[midway,above=2pt] {\footnotesize $\gamma\max S+\lambda$}(t2);
	
	
	\def\porbuff{\one+0.6}
	\def\porinc{0.2}
	\def\pordotbuff{0.6}	
	
	\coordinate (ep1) at ({\poff-\pordotbuff},{(\porbuff) + (\n)*(\porinc)});
	\coordinate (ep2) at ({\poff-\pordotbuff},{(\porbuff) + (\n-1)*(\porinc)});
	\coordinate (ep3) at ({\poff-\pordotbuff},{(\porbuff) + (\n-2)*(\porinc)});
	
	\coordinate (exitx) at ($(t2) + (\porinc,0)$);
	
	\draw[->,blue,rounded corners=\roun] (ep1) -- (ys\n|-ep1) -- (ys\n);
	\foreach\i [count=\j] in {2,...,\n}
	{
		\draw[->,blue,rounded corners=\roun] (xs\i) -- (ys\j);
	}
	
	\coordinate (qs1) at ({\poff-\pordotbuff+3*\porinc},{0-(\porbuff)});
	\coordinate (qs2) at ({\poff-\pordotbuff+\porinc},{0-(\porbuff)});
	
	\draw[->,blue,rounded corners=\roun] (xs1) -- (xs1|-qs1) -- (qs1);
	\draw[->,dashed,blue,rounded corners=\roun] (qs1) -- (qs2) -- (qs2|-ep2) -- (qs1|-ep2);
	\draw[->,blue,rounded corners=\roun] (qs1|-ep2) -- (t2|-ep2) -- (t2);
	\draw[->,blue,rounded corners=\roun] (t1) -- (t1|-ep3) -- (exitx|-ep3);
	
	\coordinate (v1) at ($(R) + (R) - (t1)$);
	\coordinate (v2) at ($(R) + (R) - (t2)$);
	\draw[->,red] (v1) -- (v2);
	\draw (v1) -- (t1);
	\draw (v2) -- (t2);
	
	\draw[->] (P |- 52,-\one-\desoffset) -- ($(v1 |- 52,-\one-\desoffset) + (\coordoffset,0)$);

	\draw[opacity=0.5,dashed](v2 |- 52,\one+1.5*\desoffset) -- (v2 |- 52,-\one-\desoffset);
	\node[below,text height=0.7em] at (v1 |- 52,-\one-\desoffset) {$a^i_*$};
	\draw[opacity=0.5,dashed](v1 |- 52,\one+1.5*\desoffset) -- (v1 |- 52,-\one-\desoffset);
	\node[below,text height=0.7em] at (v2 |- 52,-\one-\desoffset) {$b^i_*$};
	\draw[->](v1 |- 52,\one+1.5*\desoffset) -- (v2 |- 52,\one+1.5*\desoffset) node [midway, above] {$e^i_*$};
	
	\draw[->,green!50!red]($(P) - (0,\twisth)$) -- ($(P) + (\twistl,-\twisth)$) -- ($(P) + (\twistl,\twisth)$)-- ($(P) + (0,\twisth)$);
	\draw[->,green!50!red] (P) -- ($(Q) + (\twistl,0)$) -- ($(Q) + (\twistl,\twisth)$) -- ($(Q) + (-\twistl,\twisth)$) -- ($(Q) - (\twistl,\twisth)$) -- ($(R) + (\twistl,-\twisth)$) --($(R) + (\twistl,\twisth)$) -- ($(R) + (-\twistl,\twisth)$) --($(R) - (\twistl,0)$) -- ($(v1) + (0,\one)$);
	
	\draw[decoration={brace,raise=3pt},decorate] (v2 |- 52,-\one) -- (R |- 52,-\one)  node[midway,below=2pt] {\footnotesize $\varepsilon$}({\poff},-\one);

\end{tikzpicture}}
	\caption{Construction of the encoding gadget. Mirror edges are red, connector edges blue and the target curve is green. Projection cones are black.}
	\label{hardness:splitgadget}
\end{figure}

\begin{table}[h]%
	\centering
		\begin{tabular}{ |c l l| } %
			\hline
			Step 1: & $\ttt{p}_1=\ttt{a}^{i-1}_*+(\delta,1)$ &  \\ 
			& $\ttt{d}_1=\overline{\ttt{a}^{i-1}_*\ \ttt{p}_1}\cap H_1$ & \\ 
			& $\ttt{c}_1=\overline{\ttt{b}^{i-1}_*\ \ttt{p}_1}\cap H_1$ & \\ 
			\hline
			Step 2: & $\ttt{p}_2=\ttt{p}_1 + (\delta+\delta',0)$ & \\ 
			& $\ttt{d}_j=\overline{\ttt{a}^{i-1}_*\ \ttt{p}_1}\cap \overline{(\ttt{d}_1-\ttt{s}_j)\ \ttt{p}_2}$ &with $\ttt{s}_j=((j-1)(\beta+\lambda) , 0)$ \\  
			& $\ttt{c}_j=\overline{\ttt{b}\ \ttt{p}_1}\cap \overline{(\ttt{c}_1-s_j)\ \ttt{p}_1}$ &\\
			\hline
			Step 3: & $\ttt{p}_3=\ttt{p}_2 + (\delta+\delta'+\gamma\max S_i,0)$ & \\ 
			& $\ttt{c}'_1=\overline{\ttt{d}_1\ \ttt{p}_2}\cap H_{-1}$ & \\ 
			& $\ttt{d}'_1=\overline{\ttt{c}_1\ \ttt{p}_2}\cap H_{-1}$ & \\ 
			& $\overline{\ttt{b}}=(\overline{\ttt{c}'_1\ \ttt{p}_3}\cap H_{1}) - (\gamma\max S_i,0)$ & \\ 
			& $\overline{\ttt{a}}=\overline{\ttt{d}'_1\ \ttt{p}_3}\cap H_{1}$ &\\ 
			& $\ttt{c}'_j = \overline{\ttt{d}_i\ \ttt{p}_2}\cap\overline{\overline{\ttt{b}}_i\ \ttt{p}_3}$&with $\overline{\ttt{b}}_i = \overline{\ttt{b}} + \gamma(\max S_i - s_{i,j},0)$\\
			& $\ttt{d}'_j = \overline{\ttt{c}_i\ \ttt{p}_2}\cap\overline{\overline{\ttt{a}}_i\ \ttt{p}_3}$&with $\overline{\ttt{a}}_i = \overline{\ttt{a}} - \gamma(s_{i,j},0)$\\
			\hline
			Step 4: & $\ttt{p}_4=\ttt{p}_3+(\delta + \gamma\max S_i + \delta',0)$ & \\ 
			& $\ttt{b}^i_*=\overline{\overline{\ttt{a}}\ \ttt{p}_4}\cap H_{-1}$ & \\ 
			& $\ttt{a}^i_*=\overline{\overline{\ttt{b}}\ \ttt{p}_4}\cap H_{-1}$ & \\ 
			\hline
		\end{tabular}%
		\caption{Precise construction of the $i$th encoding gadget. Index $i$ is omitted in most cases.}%
		\label{hardness:thetable}%
	\end{table}%
The overall structure of a gadget $g_i$ for some $1\leq i\leq k$ is depicted in Figure~\ref{hardness:splitgadget}.
This gadget will encode the $i$th table $S_i = \{s_{i,1},\ldots,s_{i,n}\}$ of the $k$-Table-SUM instance.
The construction of the precise values is given in Table~\ref{hardness:thetable}.
As for the parameters, $\lambda^i$ is the length of the entry edge, determined by the previous gadget $g_{i-1}$, and $\beta$ is the global spacing parameter.
The parameters $\delta^i$ and $\delta'^i$ are auxiliary parameters, with ${\delta'^i = \lambda^i + \varepsilon}$ and ${\delta^i = \max(2\varepsilon +1 , (n-1)(\lambda^i + \beta) - \delta'^i)}$. 
Excluding the entry edge the base curve $B_i$ consists of $2n+2$ mirror edges and $\cO(n)$ connector edges.
For $1\leq j\leq n$ the first $n$ mirror edges $e^i_j$ are defined by $\textrm{c}^i_j$ and $\textrm{d}^i_j$, and the second $n$ mirror edges $e'^i_j$ are defined by $\textrm{c}'^i_j$ and $\textrm{d}'^i_j$.
The last two mirror edges are defined by $\bar{\textrm{a}}^i$ and $\bar{\ttt{b}}^i$, and $\ttt{a}^i_*$ and $\ttt{b}^i_*$.
All of these mirror edges lie in either $H_{\geq 1/2}^{\leq 1}$ or $H_{\geq -1}^{\leq -1/2}$ by construction.
The target curve $T_i$ has four twists centred at $\ttt{p}^i_1,\ldots,\ttt{p}^i_4$.
Since the index $i$ will not change other than for the entry and exit edge, we will omit these indices in the construction of this gadget.

The intuition behind the construction is as follows:
The first two steps place the first projection point at a distance from the entry edge, such that $n$ copies of the entry edge fit into the \textit{projection cone}.
By projection cone we denote the cone we get, when projecting the edge through a projection centre.
The edges must satisfy further constraints, namely that all of the edges lie in $H_{\geq 1/2}^{\leq 1}$ and they have sufficient distance in $x$-direction.

These edges offer the choice, which item should be taken.
Step $3$ places an edge $e'$ from $\bar{\ttt{a}}$ to $\bar{\ttt{b}}$, where all the diverging paths have to meet, and then places $n$ copies of the entry edge in the $n$ disjunct projection cones such that their projections onto $e'$ have a relative offset according to the values in the list.
Step $4$ defines the entry edge to the next gadget, since we have to mirror the data once more to not introduce sign errors, due to every focal point 'flipping' the 'image' i.e. the values, like a pinhole camera would.
The edges in Step $3$ and $4$ are used to recombine the diverging paths making sure that the offset between the paths corresponds to the value of the items in the list $S_i$.

A shortcut curve traversing this gadget will look as follows. 
A shortcut curve reaches some point in the entry edge $e^{i-1}_*$.
From here it takes a shortcut to some $e^i_j$.
The next shortcuts are forced to land on $e'^i_j$, then $e'^i$ and finally $e^i_*$.


\begin{figure}
	\centering
				    \begin{tikzpicture}
\def\one{2.0}
\def\l{2.0}
\def\eps{2.0}

\def\poreps{0.5}
\def\porinc{0.2}

\def\twisth{0.075}
\def\twistl{0.75}

\coordinate (P) at ({\l+\eps},0);
\coordinate (s) at ({\l/2.0},-\one);
\coordinate (t) at ({2*\eps+1.5*\l},\one);
\coordinate (a1) at (0,-\one);
\coordinate (b1) at (\l,-\one);

\begin{scope}
\clip  (t) rectangle ($(t) + (\one,-2*\one)$);
\fill [opacity = 0.2] ($(t) - (0,\one)$) circle(\one);
\end{scope}

\fill [opacity = 0.2] (0,-\one) rectangle (t);

\draw[->,red] (b1) -- (a1);

\draw[->,blue,rounded corners=5pt] (a1) -- ($(a1) - (0,\poreps)$) -- ($(t) + (\one+\poreps,-2*\one - \poreps)$) -- ($(t) + (\one+\poreps,+ \poreps)$) -- ($(t) + (0,+ \poreps)$) -- (t);

\draw[->,blue,rounded corners=5pt] ($(a1) + (-\porinc,-\poreps - \porinc)$) -- ($(b1) + (0,-\poreps - \porinc)$) -- (b1);


\draw[->,green!50!red] (0,0) -- ($(P) + (\twistl,0)$) --  ($(P) + (\twistl,\twisth)$) -- ($(P) + (-\twistl,\twisth)$) -- ($(P) - (\twistl,\twisth)$) -- ($(t) - (0,\one+\twisth)$);

\draw[decoration={brace,raise=3pt,mirror},decorate] (0,-\one) -- node[midway,below=2pt] {\footnotesize $\lambda$} (\l,-\one);

\draw[decoration={brace,raise=3pt,mirror},decorate] (\l,-\one) -- node[midway,below=2pt] {\footnotesize $\varepsilon$} ({\l+\eps},-\one);

\draw[decoration={brace,raise=3pt,mirror},decorate] ({\l+\eps},-\one) -- node[midway,below=2pt] {\footnotesize $\varepsilon+\lambda-\gamma(\sigma+1)$} ({2*\eps+1.5*\l},-\one);

\draw [<->,thick] (0,\one+2*\poreps) node (yaxis) [above] {$y$}
|- ($(t) + (\one+\poreps,-2*\one - 2*\poreps)$) node (xaxis) [right] {$x$};

\draw [dashed] (t |- 52,\one+2*\poreps) -- (t |- 52,-\one-2*\poreps);
\draw [dashed] (P |- 52,\one+2*\poreps) -- (P |- 52,-\one-2*\poreps);
\draw [dashed] (\l,-\one |- 52,\one+2*\poreps) -- (\l,-\one |- 52,-\one-2*\poreps);
\draw[->]  (\l,-\one |- 52,\one+2*\poreps) -- (0,0 |- 52,\one+2*\poreps) node [midway, above] {$e^k_*$};

\node[below] at (t |- 52,-\one-2*\poreps) {$t$};
\node[below] at  (P |- 52,-\one-2*\poreps) {$p$};
\node[below] at  (\l,-\one |- 52,-\one-2*\poreps) {$a^k_*$};
\node[below] at  (0,0 |- 52,-\one-2*\poreps) {$b^k_*$};

\end{tikzpicture}
	\caption{Construction of the Terminal-Gadget. Mirror edges are red, connector edges are blue and the target curve is green.}
	\label{hardness:terminal}
\end{figure}

\subsubsection{Terminal gadget $g_{k+1}$}
The terminal gadget $g_{k+1}$ is the dual to the initialization gadget (refer to Figure \ref{hardness:terminal}). 
The entry edge from $(b^k_*+\lambda^k,-1)$ to $(b^k_*,-1)$ is defined by the previous gadget.
The target curve $T_{k+1}$ has a single twist at $(b^k_*+\lambda+\varepsilon,0)$ and ends at $(b^k_*+2\lambda+2\varepsilon-\gamma(\sigma + 1),0)$.
The base curve $B_{k+1}$ connects the entry edge to $(b^k_*+2\lambda+2\varepsilon-\gamma(\sigma + 1),1)$ from outside the hippodrome.
The final vertex $B(1)$ of the base curve is placed such that a shortcut from the entry edge $e^k_*$ has to start precisely at $x$-coordinate $b^k_* + \gamma(\sigma +1)$ to hit the vertex.


\subsection{Correctness}
We now want to argue that this construction is correct.
That is, there exists a feasible shortcut curve with $(4k+2)$ shortcuts if and only if the original $k$-Table-SUM instance has a solution.
We begin by showing this for a subset of shortcut curves we call \textit{one-touch}. For general shortcut curves this will be shown in Section~\ref{hardness:general}. These one-touch shortcut curves consist of only shortcuts and will never take subcurves of the base curve $B$.
In the following section we often have to argue with distances that get preserved, when getting projected through a projection point. This argument is captured in the following observation.
\begin{observation}\label{hardness:congruent}
	If an edge lies on some $H_{-\beta}$ with length $\lambda$, and some point $p$ on $H_0$ is given, we can then project the edge through $p$ onto some $H_{\alpha}$ of length $\lambda'$. This forms two congruent triangles such that $\lambda'=\frac{\alpha\lambda}{\beta}$. Refer to $e^{i-1}_*$ and $e_3$ in Figure \ref{hardness:splitgadget} as an example.
\end{observation}

\begin{wrapfigure}[13]{r}{0cm}
    \includegraphics[scale=0.5]{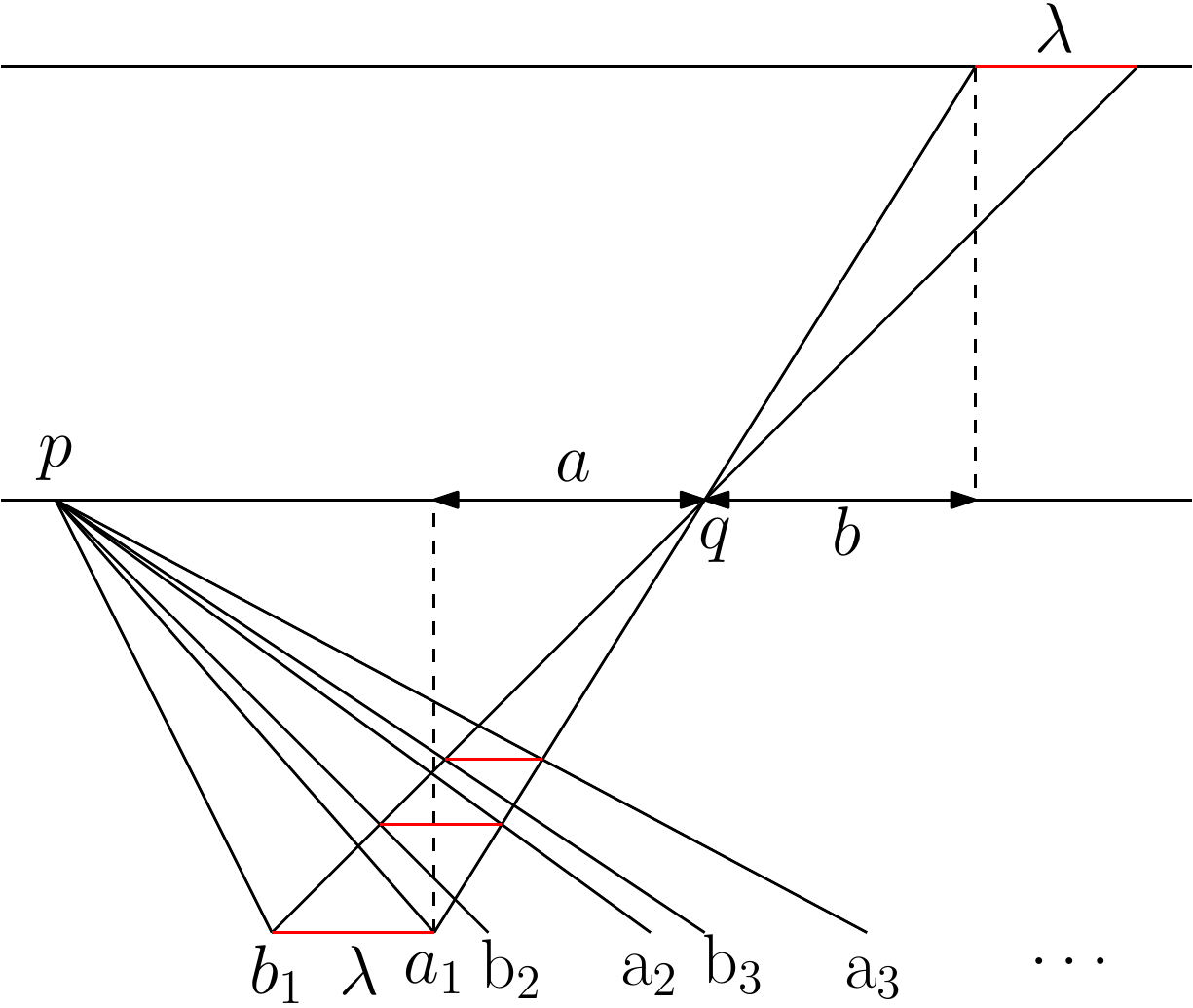}
\end{wrapfigure}
In the construction of the encoding gadgets we have many different instances of these congruent triangles.
Here we carefully adjust the distances of the edges so that the whole gadget has certain properties.
These adjustments can be thought of as shifting one triangle or cone.
These modifications are captured in the following observation.
\begin{observation}\label{hardness:coneshifter}
	Consider the set-up in the image to the right. We have $n$ cones starting in point $p$ going to points $\ttt{a}_i$ and $\ttt{b}_i$ each, which are $\lambda$ apart.
	Another cone is placed at some distance parametrized by $a=b\geq 0$ from $q$ to $\ttt{a}_1$ and $\ttt{b}_1$. This cone intersects the $n$ previous cones and forms $n$ edges (in red).
	Like in the construction of the encoding-gadget we will call the endpoints of these edges $\ttt{c}'_i$ and $\ttt{d}'_i$. There are two modifications we want to look at:
	Firstly increasing $a$ and $b$ equally at the same time.
	This modifies the edges such that for the resulting points $\ttt{c}^*_i$ and $\ttt{d}^*_i$ it holds that $d^*_i-c^*_{i-1} > d'_i-c'_{i-1}$.
	The second modification consists of increasing $a$ but not $b$.
	This simply scales the whole instance such that $d^*_i-c^*_{i-1} = s(d'_i-c'_{i-1})$, for some $s>1$.
\end{observation}

\begin{definition}[One-touch encoding]\label{hardness:one-touch-encoding}
	Let $I={\iota_1,\ldots,\iota_k}$ be an index set of a $k$-Table-SUM instance. We construct a one-touch shortcut curve $B_I$ of the base curve incrementally.
	The first two vertices on the initial gadget are defined as follows.
	We choose the first vertex of the base curve $B(0)$ for $\ttt{v}_0^0$, then we project it through the first projection center $\ttt{p}^0$ onto $e^0_*$ to obtain $\ttt{v}_*^0$.
	Now for $1\leq i\leq k$ we project $\ttt{v}_*^{i-1}$ through $\ttt{p}^i_1$ to land on $e^i_{\iota_i}$ to obtain $\ttt{v}_1^i$.
	We continue by projecting $\ttt{v}_l^i$ through $\ttt{p}_l+1$ onto $B_i$ to obtain $\ttt{v}^i_{l+1}$ for $1\leq l \leq 3$ (refer to Figure~\ref{hardness:fig-bounce}).
	Since these projections are all forced, no choices have to be made.
	Let $\ttt{v}_*^i=\ttt{v}_4^i$.
	We continue this construction throughout all gadgets in order of $i$.
	Finally, we choose $B(1)$ as the last vertex of our shortcut curve.
\end{definition}
\begin{lemma}\label{hardness:lemma1}
	For any $1\leq i\leq k$ and $1\leq j\leq n$ let $\ttt{x}^i_j$ be leftmost point on $e^i_*$ reachable by projections starting on edge $e^i_j$. Then $x^i_j-b_*^i = \gamma s_{i,j}$.
\end{lemma}
\begin{proof}
	This follows directly from the construction (refer to Figure~\ref{hardness:fig-bounce} and Table~\ref{hardness:thetable}) and repeated application of Observation \ref{hardness:congruent}.
	$x_j$ is determined by the projection of $\ttt{c}_j$ through $\ttt{p}_2$, which is $\ttt{d}'_j$. 
	Projecting this through $\ttt{p}_3$ lands on $\overline{\ttt{a}}_j$ which by another projection through $\ttt{p}_4$ lands on $\ttt{x}_j$.
	The offset between $\overline{\ttt{a}}_j$ and $\overline{\ttt{a}}$ is precisely the offset between $\ttt{x}_j$ and $\ttt{b}'$.
	And this offset is by construction $\gamma s_{i,j}$.
\end{proof}
\begin{lemma}\label{hardness:one-touch-induction}
	Given a shortcut curve $B_I$, which is a one-touch encoding, let $\ttt{v}^i_*$ be the vertex of $B_I$ on the entry-edges $e^i_*$ of gadgets $g_i$ for all applicable $i$. Then $||\ttt{v}^i_* - \ttt{b}^i_*|| = \gamma (\sigma_i +1)$, where $\sigma_i$ is the $i$th partial sum of the index set $I$ encoded by $B_I$.
\end{lemma}
\begin{proof}
	We prove this via induction.
	For $i=0$ this is correct by construction of the initialization gadget.
	Refer for the following argument to Figure \ref{hardness:fig-bounce} and Observation \ref{hardness:congruent}.
	For all choices of $j$ we have $||\ttt{v}^{i-1}_* - \ttt{b}^{i-1}_*|| = ||\ttt{v}^i_*-\ttt{x}^i_j||$.
	This follows immediately from following the projections:
	\[||\ttt{v}^{i-1}_* - \ttt{b}^{i-1}_*|| = \alpha_j||\ttt{v}_1^i - \ttt{c}_j|| = \alpha'_j||\ttt{v}^i_2 - \ttt{d}'_j|| = ||\overline{\ttt{v}}_3^i - \overline{\ttt{a}}_j|| = ||\ttt{v}^i_* - \ttt{x}^i_{j}||.\]
	Together with Lemma \ref{hardness:lemma1} we have
	\[||\ttt{v}^i_* - \ttt{b}^i_*||=||\ttt{v}^i_* - \ttt{x}^i_j|| + ||\ttt{x}^i_j - \ttt{b}^i_*|| = \gamma(\sigma_{i-1} +1) + \gamma s_{i,j} = \gamma(\sigma_i+1).\]
\end{proof}

\begin{figure}
	\includegraphics[width=\textwidth]{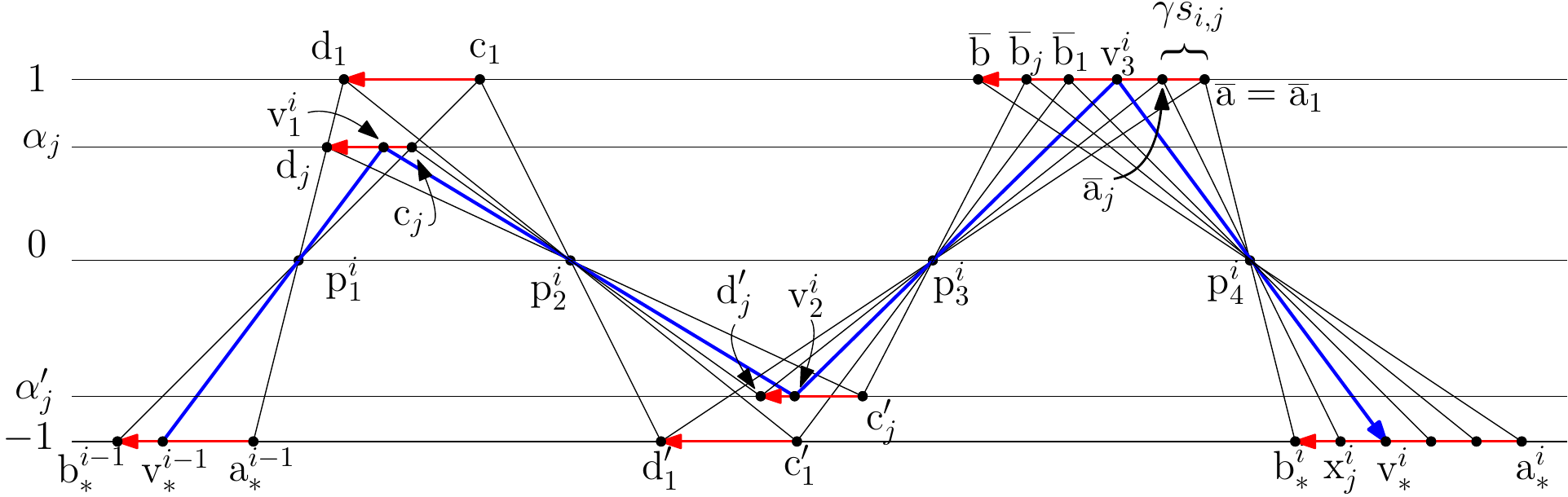}
	\caption{The path of a shortcut curve through the gadget $g_i$ in the case, where $s_{i,j}$ is selected from the $i$th list  (Lemma \ref{hardness:one-touch-induction}). Most top indices $i$ are omitted. Furthermore, for presentation the mirror edges have horizontal overlap, which in the construction they do not.}
	\label{hardness:fig-bounce}
\end{figure}


\begin{lemma}\label{hardness:bufferzone}
	The constructed base curve never enters any buffer zone centred at a projection centre.
\end{lemma}
\begin{proof}
	We will not consider the connector edges outside the hippodrome, since they can easily be placed such that they do not enter buffer zones.
	
	We first consider the encoding gadget $g_i$.
	For this we omit the top index $i$, as we only look at a single gadget at a time.
	For the buffer zones centred at $\ttt{p}_2$ and $\ttt{p}_4$ for $1\leq i\leq k$ the claim is implied by construction.	
	For the buffer zone centred at $\ttt{p}_1$ the closest edge in $x$-direction inside the hippodrome is by construction $e_n$. And of this edge, $\ttt{d}_3$ is the closest point. By construction this edge is at a height of $\geq\frac{1}{2}$, since $\delta + \delta' \geq (n-1)(\lambda + \beta)$ (this will again be observed in Lemma \ref{hardness:big4}). Hence, $d_n \geq p_1 + \frac{\delta}{2} > p_1 + \varepsilon$, implying the claim for this projection centre as well.
	
	For the buffer zone centred at $\ttt{p}_3$ the closest edge in $x$-direction inside the hippodrome is by construction $e'_n$. And of this edge $\ttt{c}_n'$ is the closest point. We begin by bounding the $y$-coordinate of $\ttt{c}_n'$. For this we compare two triangles. The first triangle consists of the points $\ttt{p}_1$, $\ttt{p}_2$ and $\ttt{d}_n$, for which the height of $\geq\frac{1}{2}$ is already known from the first part of the proof. The second triangle is given by the points $\ttt{p}_2$, $\ttt{p}_3$ and $\ttt{c}'_n$. Note that the slope of $\overline{\ttt{p}_2\ \ttt{c}'_n}$ is the slope of $\overline{\ttt{d}_n\ \ttt{p}_2}$, since they lie on a common line by construction. The slope of $\overline{\ttt{c}'_n\ \ttt{p}_3}$ and the slope of $\overline{\ttt{p}_1\ \ttt{d}_n}$ is by construction $\frac{1}{\delta}$. Hence, the two triangles are congruent. And since the last edge of the second triangle is larger than that of the first triangle, the height must be larger.
	Hence, the $y$-coordinate of $\ttt{c}'_n$ is at most $-\frac{1}{2}$. Now again, since the slope of $\overline{\ttt{c}'_n\ \ttt{p}_3}$ is $\frac{1}{\delta}$, $p_3-c'_n\geq\frac{\delta}{2}>\varepsilon$.
	This implies the claim for the projection centre $\ttt{p}_3$.
	
	The last two buffer zones left to analyse are around the projection centre in the initialization and end gadget, however the claim follows directly from construction.
\end{proof}

\begin{lemma}[4-monotonicity \cite{Buchin2013ComputingTF}]\label{hardness:monotone}
	Any feasible shortcut curve is rightwards 4-monotone. That is, if $x_1$ and $x_2$ are the $x$-coordinates of two points that appear on the shortcut curve in that order, then $x_2+4\geq x_1$. Furthermore, it lies inside or on the boundary of the hippodrome.
\end{lemma}
\begin{proof}
	Any point on the feasible shortcut curve has to lie within distance $1$ to some point of the target curve, thus the curve cannot leave the hippodrome. As for the monotonicity, assume for the sake of contradiction that there exist two such points with $x_2 + 4 < x_1$. Let $\hat{x}_1$ be the $x$-coordinate of the point on the target curve matched to $x_1$ and let $\hat{x}_2$ be the one for $x_2$. By the Fréchet matching it follows that $\hat{x}_2-1+4<\hat{x}_1 + 1$. This would imply that the target curve is not $2$-monotone, which contradicts the way we constructed it.
\end{proof}

\begin{lemma}\label{hardness:big4}
	If $\beta \geq 32 + 4\lambda^{i-1}$, $\varepsilon < 16$ and $n\geq 2$, then any two mirror edges of the gadget $g_i$ are at least $4$ apart. Additionally, if $\beta \geq \lambda^i$, then all mirror edges lie inside the hippodrome.
\end{lemma}
\begin{proof}
	In this proof we omit the top index $i-1$ from $\lambda$ unless stated otherwise.
	The top index $i$ from $c_j,c_j',d_j,d_j'$ are omitted as well.
	We begin by computing $\ttt{c}_j$ and $\ttt{d}_j$ for $1\leq j\leq n$. To not worry about offsets, we translate the instance such that $\ttt{p}_1$ coincides with the origin. Then $\ttt{d}_j$ is defined as the intersection of a line $l_1$ from $(0,0)$ to $(\delta,1)$ and $l_2$ from $(\delta - (j-1)(\lambda + \beta),1)$ to $(\delta + \delta',0)$. Note that $l_1$ has a slope of $\frac{1}{\delta}$, and $l_2$ has a slope of $\frac{1}{(j-1)(\lambda+\beta)+\delta'}$. Thus the $x$-coordinate of the intersection point satisfies	
	\[\frac{d_j}{\delta}=1-\frac{d_j-(\delta - (j-1)(\lambda + \beta))}{\delta' + (j-1)(\lambda + \beta)}.\]	
	From this follows
	\[d_j = \frac{\delta(\delta' + \delta)}{\delta + \delta' + (j-1)(\lambda + \beta)}.\]	
	Since $l_1$ has slope $\frac{1}{\delta}$, we thus have	
	\[\ttt{d}_j = \frac{1}{\delta + \delta' + (j-1)(\lambda + \beta)}\bigg(\delta(\delta' + \delta),(\delta' + \delta)\bigg).\]
	Via similar calculations we get 	
	\[\ttt{c}_j = \frac{1}{\delta + \delta' + (j-1)(\lambda + \beta)}\bigg((\delta+\lambda)(\delta' + \delta),(\delta' + \delta)\bigg).\]	
	Since $\delta+\delta'\geq(n-1)(\lambda + \beta)$, and $j\leq n$, the $y$-coordinate of these edges is at least $\frac{1}{2}$.	
	Now we want to prove that $c_j+4<d_{j-1}$ holds for $2\leq j\leq n$. We will define auxiliary variables $o_j=(j-1)(\lambda+\beta)$ and $a_j = \delta + \delta' + o_j$. Thus we need to show	
	\[\frac{(\delta + \lambda)(\delta + \delta ')}{ \delta + \delta' + o_j} + 4 < \frac{\delta(\delta + \delta ')}{ \delta + \delta' + o_{j-1}}.\]	
	Multiplying both sides with $a_ja_{j-1}$ we get	
	\[(\delta + \lambda)(\delta + \delta')a_{j-1} + 4a_ja_{j-1} < \delta(\delta + \delta')a_j.\]
	Since $a_j = a_{j-1} + \lambda + \beta$ this is equivalent to	
	\[\lambda(\delta + \delta')(\delta + \delta' + o_{j-1})+ 4a_ja_{j-1} < \delta(\delta+\delta')(\lambda + \beta),\]	
	simplifying to	
	\[\lambda(\delta + \delta')(\delta' + o_{j-1})+ 4a_ja_{j-1} < \beta\delta(\delta+\delta'). \]
	Since $o_j \leq o_n = (n-1)(\lambda + \beta) \leq \delta + \delta '$ and thus $a_j \leq 2(\delta + \delta ')$ we get	
	\[2\lambda(\delta+\delta')^2 + 16(\delta+\delta')^2 < \beta\delta(\delta+\delta'),\]	
	which is equivalent to 	
	\[(16 + 2\lambda)\left(1+\frac{\delta'}{\delta}\right) < \beta.\]	
	Assuming $n\geq 2$ and $\beta > \lambda + 2 \varepsilon$ we have $\frac{\delta'}{\delta}<\frac{\varepsilon + \lambda}{(n-1)\beta + (n-2)\lambda - \varepsilon}<\frac{\varepsilon + \lambda}{\beta - \varepsilon} \leq 1$.	
	Then
	\[(32 + 4\lambda) < \beta.\]	
	Hence, any two mirror edges $e_j$ are at least $4$ apart from each other, if $\beta>\max(32+4\lambda,\lambda + 2\varepsilon) = 32+4\lambda$, assuming $\varepsilon < 16$.
	
	From the previous argument together with Observation \ref{hardness:coneshifter} it follows that the edges $e_j'$ would be at least $4$ apart, with $\overline{b}_j = \overline{b}_i$ and $\overline{a}_j = \overline{a}_i$ for all $1\leq j\leq n$ and any fixed $1\leq i\leq n$.
	If we now imagine $\overline{b}_j$ and $\overline{a}_j$ sweeping from $\overline{b}_{j-1}$ and $\overline{a}_{j-1}$ to its actual value defined in the table, we see that all the edge-end points move in positive $x$ direction, where each $e_j'$ moves further 'right' than $e_{j-1}'$ since $s_{i,j-1} < s_{i,j}$, thus the distance in $x$-direction only ever increases proving the claim.
	
	The last thing to prove is that all edges lie inside the hippodrome. For $e_j$ this follows immediately from construction. Similarly for $e_1'$, since $s_{i,1} = 0$. For $e_j'$ with $j\geq 2$ we need to look at the construction a bit more carefully.
	The $y$-coordinate of the edge $e_j'$ equals	
	\[-\frac{\delta + \delta' + \gamma\max S_i}{\delta+ \delta' + o_j + \gamma(\max S_i - s_{i,j})}.\]	
	As $s_{i,j}<\max S_i$, the $x$-coordinate is bounded from below by 	
	\[-\frac{\delta + \delta' + \gamma\max S_i}{\delta+ \delta' + o_j},\]	
	which is at least $-1$, if $o_j \geq \gamma\max S_i$. This holds, since $o_j\geq \beta \geq \lambda^i = \lambda^{i-1}+\gamma\max S_i$.
\end{proof}

\begin{lemma}\label{hardness:projectioncenter}
	If $16 > \varepsilon > 2$ and $\beta > \max(32+4\lambda^k,\lambda^{k+1})$, then a feasible shortcut curve passes through every buffer zone of the target curve via its projection centre and furthermore it does so from left to right.
\end{lemma}
\begin{proof}
	Any feasible shortcut curve has to start at $B(0)$ and end at $B(1)$, and all of its vertices must lie in the hippodrome or on its boundary. By Lemma \ref{hardness:bufferzone} the base curve does not enter any of the buffer zones and therefore the feasible shortcut curve has to pass through the buffer zone by using a shortcut. 
	If we choose the width of a buffer zone $2\varepsilon > 4$, then the only way to do this, while matching to the two associated vertices of the target curve in their respective order, is to go through the intersection of their unit disks.
	The intersection lies at the centre of the buffer zone, as we saw in the beginning.
\end{proof}

\begin{lemma}\label{hardness:one-touch}
	If $16 > \varepsilon > 2$ and $\beta > \max(32+4\lambda^k,\lambda^{k+1})$, then a feasible shortcut curve that is one-touch visits exactly one of the edges $e^i_j$ and exactly one of the edges $e'^i_j$ for $1\leq j \leq n$ in every gadget $g_i$ for $1\leq i\leq k$.
	Furthermore it visits all edges $e_*^i$ for $0\leq i\leq k$.
\end{lemma}
\begin{proof}
	By Lemma \ref{hardness:monotone} any feasible shortcut curve is $4$-monotone. Furthermore, it starts at $B(0)$ and ends at $B(1)$. By Lemma \ref{hardness:projectioncenter} it goes through all projection centres of the target curve from left to right. We first want to argue that it visits at least one mirror edge between two projection centres, i.e. that it cannot 'skip' such a mirror edge by matching to two twists in one shortcut. Such a shortcut would have to lie on $H_0$, since it has to go through the two corresponding projection centres lying on $H_0$. By construction, the only possible endpoints of such a shortcut lie on the connector edges that connect to mirror edges. Assume such a shortcut could be taken by a shortcut curve starting from $B(0)$. Then there must be a connector edge, which intersects a line from a point on a mirror edge through the projection centre. In particular, since the curve has to go through all projection centres, one or more of the following must be true for some $1\leq i\leq k$:
\begin{compactenum}
		\item[-] there exists a line through $\ttt{p}^i_1$ intersecting a mirror edge $e_*^{i-1}$ and a connector edge of $e^i_j$,
		\item[-] there exists a line through $\ttt{p}^i_2$ intersecting a mirror edge $e_j^i$ and a connector edge of $e'^i_l$ for some $l$, or
		\item[-] there exists a line through $\ttt{p}^i_3$ intersecting a mirror edge $e'^i_j$ and a connector edge of $e'^i_l$ for some $l>j$.
\end{compactenum}
	However, this was prevented by the careful placement of these connector edges.
	
	It remains to prove that the shortcut curve can not visit more than one $e^i_j$ and cannot visit more than one $e'^i_j$ and therefore visits exactly one mirror edge between two projection centres.
	The shortcut curve has to lie inside or on the boundary of the hippodrome and is $4$-monotone (Lemma \ref{hardness:monotone}).
	At the same time, we constructed the gadget such that the mirror edges between two consecutive projection centres have distance at least $4$ to one another by Lemma \ref{hardness:big4}.
	Furthermore, inside the projection cone from $e'^i_j$ to $\ttt{p}^i_3$ all mirror edges come before (as parametrized by the base curve) $e'^i_j$, implying the claim.
\end{proof}
\begin{corollary}
	A feasible shortcut curve consists of exactly $4k+2$ shortcuts. One shortcut for the initialization and end gadgets, and $4$ shortcuts in each encoding gadget.
\end{corollary}
Putting the above lemmas together implies the correctness of the reduction for shortcut curves that are one-touch i.e., which visit every edge in at most one point.
\begin{lemma}\label{hardness:itallcomestogethernow}
	If $16 > \varepsilon > 2$ and $\beta > \max(32+4\lambda^k,\lambda^{k+1})$, then for any feasible one-touch shortcut curve $B_\diamond$, it holds that the index set $I$ encoded by $B_\diamond$ sums to $\sigma$. Furthermore, for any index set $I$ that solves the $k$-Table-SUM instance there is a feasible one-touch shortcut curve that encodes it.
\end{lemma}
\begin{proof}
	Lemma \ref{hardness:projectioncenter} and Lemma \ref{hardness:one-touch} imply that $B_\diamond$ must be a one-touch encoding as defined in Definition \ref{hardness:one-touch-encoding}. By Lemma \ref{hardness:one-touch-induction} the second last vertex of $B_\diamond$ is the point on the edge $e^k_*$ which is at distance $\gamma(\sigma_\diamond + 1)$ to $\ttt{b}^k_*$, where $\sigma_\diamond$ is the sum encoded by the subset selected by $B_\diamond$. The last vertex of $B_\diamond$ is equal to $B(1)$, which we placed in distance $\gamma(\sigma + 1)$ to the projection of $\ttt{b}^k_*$ through $\ttt{p}^{k+1}_1$. Thus the last shortcut of $B_\diamond$ passes through the last projection centre of the target curve, if and only if $\sigma_\diamond=\sigma$. It follows that if $\sigma_\diamond\neq\sigma$, then $B_\diamond$ cannot be feasible. For the second part of the claim, we construct a one-touch encoding as defined in Definition \ref{hardness:one-touch-encoding}. By the above analysis it will be feasible, if the subset sums to $\sigma$, since the curve visits every edge of $B$ in at most one point and in between uses shortcuts, which pass through every buffer zone from left to right and via the buffer zones projection centre.
\end{proof}

\subsection{Size of coordinates}

\begin{lemma}\label{hardness:size}
	Let $\lambda^i,n$ and $S_i=\{\sigma_1,\ldots,\sigma_n\}$ be given, and $\varepsilon < 16$. Then the length of an encoding gadget is in $\cO(\lambda^i n + \gamma\max S)$, and $\lambda^{i+1} = \lambda^i + \gamma\max S_i$.
\end{lemma}
\begin{proof}
	From the construction we get $\delta^i=\max(2\varepsilon + 1,(n-1)(\lambda^i+\beta)-\delta'^i)$ and $\delta'^i = \lambda^i + \varepsilon$.
	From Lemma \ref{hardness:itallcomestogethernow} we get $\beta\geq\max(4\lambda^i+32,\lambda^{i+1})$.
	From the construction we see that $p_1 - b_*^{i-1}=\lambda^i+\delta^i$, $p_2-p_1=\delta^i+\delta'^i$, $p_3-p_2 = \gamma\max S_i + \delta^i + \delta'^i$, $p_4-p_3=\delta'^i+\gamma\max S_i + \delta^i$ and $b_*^{i+1} - p_4 = \varepsilon$, with the length of the next edge being $\lambda^i + \gamma\max S_i$.
	With all these values we get
	\begin{align*}
	b_*^{i+1}- b_*^{i-1} &= \lambda^i + 4\delta + 3\delta'+2\gamma\max S_i+\varepsilon\\
	&< 4\delta + 4\delta'+2\gamma\max S_i\\
	&< 4(n-1)(\beta + \lambda ^i) + 8\varepsilon + 4 +2\gamma\max S_i\in \cO(\lambda^i n + \gamma\max S_i).
	\end{align*}
\end{proof}

\begin{lemma}
	The curves can be constructed in $\cO(kn)$ time. Furthermore, if we choose $\varepsilon =\frac{5}{2}$ and $\beta = \max(32+4\lambda^k,\lambda^{k+1}) + 1$, then the coordinates used are in $\cO(kn\gamma\sum_{i=0}^k\max S_i)$.
\end{lemma}
\begin{proof}
	Each of the constructed gadgets uses $\cO(n)$ vertices, since we need to place $\cO(n)$ mirror and connector edges. Because we construct $k+2$ gadgets, the overall number of vertices used is in $\cO(kn)$. The curves $T$ and $B$ can be constructed using a single iteration from left to right, therefore the overall construction takes $\cO(kn)$ time.	
	By Lemma \ref{hardness:size} the length of a gadget $g_i$ is in $\cO(\lambda^i n + \gamma\max S_i)$. Since $\lambda^i = \lambda^{i-1} + \gamma\max S_{i-1}$ and $\lambda^0 = 2\gamma$, the maximum length of any $\lambda^i$ is in $\cO(\gamma\sum_{i=0}^k\max S_i)$. Hence, the claim follows.
\end{proof}

\subsection{Correctness for general shortcut curves}\label{hardness:general}

When we consider general feasible shortcut curves that might not necessarily be one-touch, they might follow a mirror edge for a short while instead of immediately taking the next shortcut.
This results in a small error when comparing the shortcut curve with a one-touch curve encoding the same index set.
We now want to contain this incremental error introduced with the control parameter $\gamma$.

\begin{lemma}\label{hardness:partial-sum-error}
	Choose $16 > \varepsilon > 2$ and $\beta > \max(32+4\lambda^k,\lambda^{k+1})$. Given a feasible shortcut curve $B_{\diamond}$, let $v_i$ be any point of $B_\diamond$ on the exit-edge $e^{i}_*$ of the gadget $g_i$. For all $i$ let $\sigma_i$ be the partial sum encoded by $B_\diamond$. If we choose $\gamma>\xi_i$, then
	
	\[b^i_*+\gamma(\sigma_i +1) - \xi_{i} \leq v_i \leq b^i_*+\gamma(\sigma_i +1) + \xi_{i}\]	
	holds, where $\xi_i = 16i+5$ is an upper bound of the maximum error possible for any shortcut curve traversing up to gadget $g_i$.
\end{lemma}	

\begin{proof}
	We prove this claim by induction on $i$. For $i=0$ the claim follows by construction of the initialization gadget: As $B_\diamond$ has to start at $B(0)$ and it has to lie completely in the hippodrome, it has to take a shortcut, and since it has to pass a twist, it must traverse its projection centre. The only point, where this shortcut can end is on the entry edge of $g_1$. By construction this point is at a distance of $\gamma$ from $b^0_*$. Since the edge is oriented leftward, $B_\diamond$ can only walk in that direction. However, $B_\diamond$ is rightwards $4$-monotone. It follows that	
	\[b^0_*+\gamma-4\leq v_1\leq b^0_*+\gamma.\]
	Since $\xi_0=5>4$ and $\sigma_0 = 0$, this implies the claim for $i=0$.
	
	For $i>0$ the curve $B_\diamond$ entering gadget $g_i$ from edge $e^{i-1}_*$ has to pass the first twist, and has to do so through the projection point. By induction	
	\[b^{i-1}_*+\gamma(\sigma_{i-1} +1) - \xi_{i-1} \leq v_{i-1} \leq b^{i-1}_*+\gamma(\sigma_{i-1} +1) + \xi_{i-1}.\]	
	Since $\gamma>\xi_i=\xi_{i-1}+16$, it follows that the distance of $v_i$ to the endpoints of the edge is	
	\[\gamma(\sigma_{i-1} +1) - \xi_{i-1} \geq \gamma - \xi_{i-1}> 16\]	
	and	
	\[\gamma(\sigma_{i-1} +1) + \xi_{i-1} \leq (\lambda^{i-1}-\gamma)+\xi_{i-1} \leq \lambda^{i-1}-16,\]
	thus $v_{i-1}$ lies at a distance greater than $4$ from the endpoints of the entry-edge of gadget $g_i$.
	
	Therefore, the only edges that can be hit through the projection point $p_1^i$ are $e_j$. 
	Denote by $o_{\max}=\gamma(\sigma_{i-1} +1) + \xi_{i-1}$ and $o_{\min} = \gamma(\sigma_{i-1} +1) - \xi_{i-1}$ the maximal and minimal offset $v_i$ may have from $b^{i-1}_*$.
	Furthermore, let $\alpha_j$ be the $y$-coordinate of the edge $e^i_j$, and similarly $\alpha_j'$ for the edge $e'^i_j$.
	We will again omit the top index of $i$, since it is fixed for the gadget $g_i$ from now on.
	Then the interval of $x$-coordinates, where the shortcut may end on $e_j$ is 	
	\[[c_j - \alpha_j o_{\max}\,,\,\,c_j - \alpha_j o_{\min}].\]	
	The length of the edge $e_j$ is $\alpha_j\lambda^{i-1}$. Thus the endpoint lies inside the edge. Now $B_\diamond$ may walk on this edge as well. Again, it can do so only leftwards. As the curve is rightwards $4$-monotone, it may do so a distance of at most 4. But since $\alpha_j\geq\frac{1}{2}$, and we already saw that $o_{\max}>\lambda+16$, the shortcut curve can not leave this edge by walking. Thus all possible points for $B_\diamond$ are determined by the interval	
	\[[c_j - \alpha_j o_{\max} - 4\,,\,\,c_j - \alpha_j o_{\min}].\]	
	Hence, the shortcut curve must leave this edge via a shortcut through $\ttt{p}_2$. It then may again walk up to $4$ to the left resulting in the interval	
	\[\left[d'_j + \alpha'_j o_{\min} -4\,,\,\,d'_j + \alpha'_j \left(o_{\max} +\frac{4}{\alpha_j}\right)\right].\]	
	Repeated application for the next two edges results in the interval for the edge $e'$	
	\[\left[\bar{a} - \gamma s_{i,j} -  \left(o_{\max} +\frac{4}{\alpha_j}\right) - 4\,,\,\,\bar{a} - \gamma s_{i,j} - \left(o_{\min} -\frac{4}{\alpha_j'}\right)\right].\]	
	Note that $\bar{a}_j = \bar{a} - \gamma s_i$ by construction. And for $e_*$ it lands in the interval	
	\[\left[b_* + o_{\min} + \gamma s_{i,j} - \frac{4}{\alpha_j'} - 4\,,\,\,b_* +  o_{\max} + \gamma s_{i,j} +\frac{4}{\alpha_j} + 4\right].\]	
	Since $\alpha_j\geq\frac{1}{2}$ and $\alpha'_j\geq\frac{1}{2}$, we get for the item $s_{i,j}$ taken by the shortcut curve	
	\[b^i_* + \gamma (\sigma_{i} +1) - \xi_{i} = b^i_*+\gamma s_{i,j} + \gamma (\sigma_{i-1} +1) - \xi_{i-1} - 16 \leq v_i \]	
	as well as 	
	\[v_i \leq b^i_*+\gamma s_{i,j} + \gamma (\sigma_{i-1} +1) + \xi_{i-1} + 16 =b^i_* + \gamma (\sigma_{i} +1) + \xi_{i},\]	
	implying the claim.
\end{proof}	

\hardness*
\begin{proof}
    Let some $k$-Table-SUM instance be given.
    Let $16>\eps>2$ and $\beta>\max(32+4\lambda^k,\lambda^{k+1})$, as well as $\gamma\geq16(k+1) + 5$.
    Let $B_\diamond$ be any feasible shortcut curve of the constructed instance for the $k$-Table-SUM instance.
	Since $B_\diamond$ is feasible, it must visit the exit edge of the last gadget $g_k$ at distance $\gamma(\sigma + 1)$ to $b^k_*$, since this is the only point that connects to $B(1)$ via a shortcut. Let $v_k=b^k_*+\gamma(\sigma + 1)$ be the $x$-coordinate of this visiting point, and let $\sigma_k$ be the sum of the subset encoded by $B_\diamond$. Lemma \ref{hardness:partial-sum-error} implies that 	
	\[b^k_* + \gamma(\sigma_k+1)-\xi_k\leq v_i = b^k_* + \gamma(\sigma_k+1) \leq b^k_* + \gamma(\sigma_k+1)+\xi_k,\]	
	since $\gamma = 16(k+1)+5>\xi_k$. Therefore,	
	\[\sigma_k - \frac{\xi_k}{\gamma}\leq\sigma\leq\sigma_k + \frac{\xi_k}{\gamma}.\]	
	Since $\gamma >\xi_k$ it follows that $\sigma_k$ must be $\sigma$, since both are integers.
	Hence, any feasible shortcut curve solves the $k$-Table-SUM instance, implying the claim.
\end{proof}

\bibliography{mybib}{}
\bibliographystyle{plain}

\newpage
\appendix

\section{Intersection Finder}\label{appendix:intersectionfinder}
\begin{algorithm}
    \caption{Intersection Finder}
	\label{appendix:alg_intersectionfinder}
    \begin{algorithmic}[1]
        \Procedure{IntersectionFinder}{$\varepsilon\delta$-simplifications $X'$ and $Y'$, $\delta$}
        
        \State Insert every start- and end-point of edges on $X'$, and every start- and end-points of arcs of $\delta$-capsules like in Figure~\ref{appendix:fig:sweeper} into $E$, a priority queue
        
        \State Sort $E$ by firstly the $x$-coordinate, and secondly by the $y$-coordinate
        
        \State Let $A$ be a self-balancing empty binary tree, and $I$ an empty array
        
        \While{$E\neq\emptyset$}
        
            \State pop the head object off $E$ into $x$
            
            \If{$x$ is the first vertex to be inserted of an object $o$}
            
                \State sorted insert $o$ into $A$ by its current $y$-coordinate
                
                \State compute the point of intersection of $x$ and its at most two neighbours
                
                \State sorted insert this point by its $x$-coordinate into $E$
            
            \EndIf
            \If{$x$ is the second vertex to be inserted of an object $o$}
                \State remove $o$ from $A$
            \EndIf
            \If{$x$ corresponds to an intersection between $o$ and $o'$}
                \State insert the intersection to $I$
                \State swap $o$ and $o'$ in $A$, and compute the point of intersection with their new neighbours updating $E$
            \EndIf
        \EndWhile
        
        \State Return $I$
        
        \EndProcedure
    \end{algorithmic}
\end{algorithm}

For completeness sake we provide a detailed description of the modified version of the Bentley-Ottman sweep-line algorithm \cite{Bentley1979AlgorithmsfRaCGI}.
\begin{lemma}\label{appendix:lemma-intersection}
	Given two polygonal curves $X$ and $Y$ in $\bR^2$, a parameter $\delta\geq0$, and let $X'$ and $Y'$ be their $\varepsilon\delta$-simplifications. One can find all $\mathcal{N}_{\leq\delta}(X',Y')$ cells in the free-space diagram that have non-empty $\delta$-free-space in $\cO(\frac{cn}{\varepsilon}\log(\frac{cn}{\varepsilon}))$ time.
\end{lemma}
\begin{wrapfigure}[13]{R}{0cm}
	\includegraphics[scale=0.6]{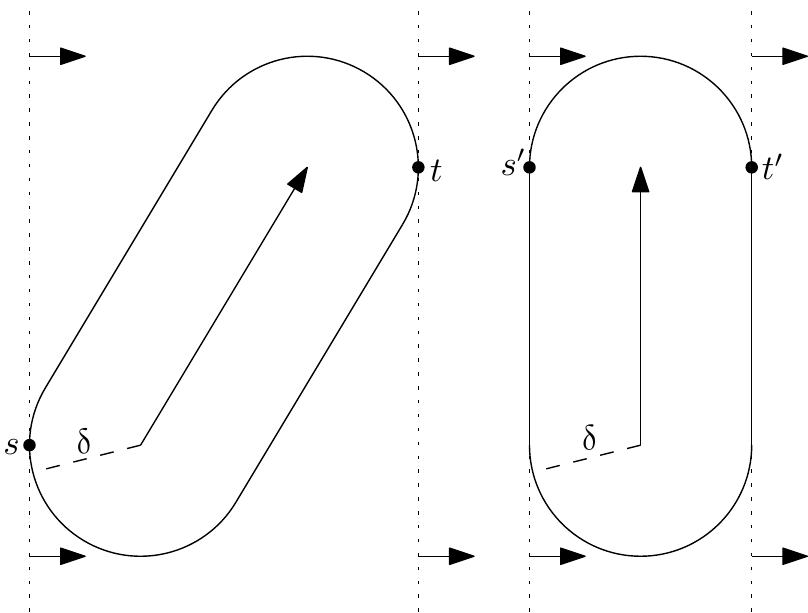}
	\caption{Arc definition for the intersection procedure.}
	\label{appendix:fig:sweeper}
\end{wrapfigure}
\begin{proof}
    Without loss of generality it suffices to find all edges of the curve $X'$ that enter and exit a $\delta$-neighbourhood of any edge of $Y'$, since any edge that is completely contained
    in this neighbourhood lies between two edges entering and leaving the neighbourhood.
	In the special case that the start or end vertex of $X'$ lies in such a neighbourhood it is easily checked by looking whether the first (resp. last) such edge is entering or leaving the neighbourhood. Entering and exiting such a neighbourhood is the same as intersecting its boundary. Thus we can modify for example the Bentley-Ottmann algorithm \cite{Bentley1979AlgorithmsfRaCGI} to find all intersections in a set of edges (refer to Algorithm~\ref{appendix:alg_intersectionfinder}).
	The main idea is to sweep along the $x$-axis and keep track of all objects that cross the sweeping line in an array of size $\cO(n)$.
	Every time a new object enters the array it checks with its at most two neighbours how far the sweeping line would have to sweep to get to the intersection point of the new object. 
	If an intersection occurs at some time in the future, we add this event to the event queue of the sweeping line.
	If the sweeping line is at an intersection event, it swaps the two objects in question and updates all new $\cO(1)$ neighbours.
	We can modify this easily to work with capsules (the geometric shape of the $\delta$-neighbourhood of an edge) by introducing two sections of the capsule into the array instead of a single line, as can be seen in Figure~\ref{appendix:fig:sweeper}.
	Intersections with its neighbours can still be checked and updated in $\cO(1)$.
	The algorithm runs in $\cO((n+k)\log(n+k))$ time, for $k$ intersecting objects.
	From Lemma \ref{approx:bounded-intersection} we know, that the number of intersections of the described objects is bound in $\cO(\frac{cn}{\eps})$.
	Hence, Lemma \ref{approx:c-pac-lin-comp} implies the claim.
\end{proof}
\section{Modified Algorithm}\label{appendix:buchin}

Here we restate the algorithm presented by Buchin, Driemel and Speckmann~\cite{Buchin2013ComputingTF}, with our modifications, resulting in an improved running time.

\begin{algorithm}[htp]
\caption{Modified Algorithm}
\label{appendix:buchinalgorydm}
\begin{algorithmic}[1]
    \Procedure{ModifiedDecider}{curves $T$ and $B$, $\delta > 0$ and $0<\varepsilon\leq 1$}
    \State Let $\eps' = \frac{\eps}{20}$
    \State Assert that $||T(0)-B(0)||\leq\delta$ and $||T(1)-B(1)||\leq\delta$
    \State Let $\cA,\overline{\cA},g_r,\overline{g}_r,g_l,\overline{g}_l$ be arrays of size $n_1$
    \For{$j=1,\ldots,n_2$}
        \State Update $\overline{\cA}\leftarrow\cA,\overline{g}_l\leftarrow g_l,\overline{g}_r\leftarrow g_r$
        \For{$i=1,\ldots,n_1$}
            \If{$i=1$ and $j=1$}
                \State $P_{i,j}=\dfree{i}{j}$
            \Else 
                \State Retrieve $VR_{(i,j-1)}$ and $HR_{(i-1,j)}$ from $\overline{\cA}[i]$ and $\cA[i-1]$
                \State Compute $N_{i,j}$ from $V\,R_{(i,j-1)}$ and $H\,R_{(i,j-1)}$
    		    \State Let $V_{i,j}=\textsc{verticalTunnel}(\overline{g}_l[i],C_{i,j},\delta)$
    		    \State Let $D_{i,j}=\textsc{apxDiagonalTunnel}(\overline{g}_r[i-1],C_{i,j},\eps',3\delta)$
    		    \State Let $P_{i,j}=Q(N_{i,j} \cup D_{i,j}) \cup V_{i,j}\cap\dfree{i}{j}$
            \EndIf
            \If{$P_{i,j}\neq\emptyset$}
                \State Update $g_l[i]$ and $g_r[i]$ using $P_{i,j}$
    		    \State Compute $V\,R_{(i,j)}$ and $H\,R_{(i,j)}$ and store them in $\cA[i]$
            \Else
    		    \State Update $g_r[i]$ using $g_r[i-1]$
            \EndIf
        \EndFor
    \EndFor
    \If{$(1,1)\in\cA[n_1]$}
    	\State Return '$d_\cS(T,B)\leq(3+\eps)\delta$'
	\Else
		\State Return '$d_\cS(T,B)>\delta$'
	\EndIf
    \EndProcedure
\end{algorithmic}
\end{algorithm}

\end{document}